\documentclass[journal, onecolumn]{IEEEtran}
\usepackage[utf8]{inputenc}
\usepackage[T1]{fontenc}
\usepackage[english]{babel}
\usepackage{amsmath,amssymb,mathtools}
\usepackage{hyperref}
\usepackage{booktabs}
\usepackage{thm-restate}
\usepackage{cleveref}
\usepackage{bbold}
\usepackage{bm}
\usepackage[normalem]{ulem}
\usepackage[numbers]{natbib} 
\usepackage[font=footnotesize]{caption}
\usepackage{subcaption}
\usepackage{enumitem,calc}
\usepackage{comment}
\usepackage{cases}
\usepackage{graphicx}
        \graphicspath{{figs/}{../figs/}}
\usepackage{tikz}

\DeclareMathOperator*{\argmin}{\mathrm{argmin}}

\DeclareMathOperator*{\Hyp}{\mathrm{Hyp}}
\DeclareMathOperator*{\Bin}{\mathrm{Bin}}
\newcommand{\er}{Erd\H{o}s--R\'{e}nyi}
\newcommand{\Sn}{\mathcal{S}_n}
\newcommand{\E}{\mathsf{E}}
\newcommand{\prob}{\mathsf{P}}
\newcommand{\var}{\mathsf{Var}}
\newcommand{\cov}{\mathsf{Cov}}
\newcommand{\R}{\mathbb{R}}

\newcommand{\nontrivpout}{\Tilde{t} ^{\mathrm{a}}}
\newcommand{\nontrivpin}{\Tilde{t} ^{\mathrm{u}}}
\newcommand{\vpin}{t ^{\mathrm{u}}}
\newcommand{\vpout}{t ^{\mathrm{a}}}
\newcommand{\Density}{R}
\newcommand{\density}{r}
\newcommand{\nonDensity}{\Tilde{R}}
\newcommand{\nondensity}{\Tilde{r}}

\newcommand{\vu}{\mathcal{V}_{\mathrm{u}}}
\newcommand{\va}{\mathcal{V}_{\mathrm{a}}}
\newcommand{\eu}{\mathcal{E}_{\mathrm{u}}}
\newcommand{\ea}{\mathcal{E}_{\mathrm{a}}}
\newcommand{\psiu}{\psi_{\mathrm{u}}}
\newcommand{\psia}{\psi_{\mathrm{a}}}

\newcommand{\SnTn}{\mathcal{S}_{n,\Tilde{n}}}
\newcommand{\muu}{\mu}
\newcommand{\mua}{\nu}
\newcommand{\orbit}{\mathcal{O}}
\newcommand{\N}{\mathbb{N}}
\newcommand{\rom}[1]{\uppercase\expandafter{\romannumeral #1\relax}}

\usepackage{amsthm}
\usepackage{thmtools}
\makeatletter
\newtheorem*{rep@theorem}{\rep@title}
\newcommand{\newreptheorem}[2]{%
\newenvironment{rep#1}[1]{%
 \def\rep@title{#2 \ref{##1}}%
 \begin{rep@theorem}}%
 {\end{rep@theorem}}}
\makeatother

\newtheorem{theorem}{Theorem}
\newreptheorem{theorem}{Theorem}
\newtheorem{lemma}{Lemma}
\newreptheorem{lemma}{Lemma}
\newtheorem{fact}{Fact}
\newreptheorem{fact}{Fact}
\newtheorem{corollary}{Corollary}

\theoremstyle{definition}
\newtheorem{remark}{Remark}

\hypersetup{
    colorlinks=true,
    linkcolor=black,
}

\makeatletter
\def\input@path{{Journal/}{Journal/sections/}{Journal/figs/}}
\makeatother
\title{Attributed Graph Alignment}
\author{Ning Zhang,
\and Ziao Wang,
\and Weina Wang,
\and Lele Wang
\thanks{Ning Zhang is with the Department of Statistics, University of Oxford, Oxford OX1 3LB, UK (email: ning.zhang@stats.ox.ac.uk). This work was done when she was with the Department of Electrical and Computer Engineering, University of British Columbia, Canada.}
\thanks{Ziao Wang is with the Department of Electrical and Computer Engineering, University of British Columbia, Vancouver, BC V6T1Z4, Canada (email: ziaow@ece.ubc.ca).}
\thanks{Weina Wang is with the Computer Science Department, Carnegie Mellon University, Pittsburgh, PA 15213, USA, (email: weinaw@cs.cmu.edu).}
\thanks{Lele Wang is with the Department of Electrical and Computer Engineering, University of British Columbia, Vancouver, BC V6T1Z4, Canada (email: lelewang@ece.ubc.ca).}
\thanks{This work was presented in part at the 2021 IEEE International Symposium on Information Theory.}
}
\allowdisplaybreaks
\begin{document}
\maketitle

\begin{abstract}
Motivated by various data science applications including de-anonymizing user identities in social networks, we consider the graph alignment problem, where the goal is to identify the vertex/user correspondence between two correlated graphs. Existing work mostly recovers the correspondence by exploiting the user-user connections. However, in many real-world applications, additional information about the users, such as user profiles, might be publicly available. In this paper, we introduce the attributed graph alignment problem, where additional user information, referred to as attributes, is incorporated to assist graph alignment.  We establish both the
{achievability and converse results}
% sufficient and necessary conditions
on recovering vertex correspondence exactly, 
% where the conditions match for a wide range of practical regimes. 
where the conditions match for certain parameter regimes.
Our results span the full spectrum between  models that only consider user-user connections and models where only attribute information is available.
\end{abstract}

% \begin{IEEEkeywords}
% Graph alignment, \er\ random graph,
% \end{IEEEkeywords}
% \IEEEpeerreviewmaketitle
% \section{Temporary}
% \input{sections/new-converse}

\section{Introduction}\label{sec:introduction}
% * Applications of the general graph alignment problem.\\
\IEEEPARstart{T}{he} graph alignment problem, also known as graph matching problem or noisy graph isomorphism problem, has received increasing attention in recent years, brought into prominence by applications in a wide range of areas \cite{Sin-Xu-global2008, Cho-Lee-progressive2012,Hag-Ng-robust2005}.  For instance, in social network deanonymization \cite{Nar-Shm-deanonymizing2009,Kor-Lat-Reconciliation2014}, two graphs are given, each of which represents the user relationship in a social network (e.g., Twitter, Facebook, Flickr, etc.). One graph is anonymized and the other graph has user identities as public information.
Then the graph alignment problem, whose goal is to find the best correspondence of the two graphs with respect to a certain criterion, can be used to de-anonymize users in the anonymous graph by finding the correspondence between them and the users with public identities in the other graph.

% * Current models (two correlated Erdos--Renyi graphs), the information-theoretic limits and computational limits on exact alignment (explain exact alignment) \\
The graph alignment problem has been studied under various random graph models, among which the most popular one is the \emph{\er\ graph pair} model (see, e.g., \cite{Ped-Gro-privacy2011, Cul-Kiy-improved2016, settling-TIT}).
In particular,
two \er\ graphs on the same vertex set, $G_1$ and $G_2$, are generated in a way such that their edges are correlated.
Then $G_1$ and an anonymous version of $G_2$, denoted as $G_2'$, are made public, where $G_2'$ is modeled as a vertex-permuted $G_2$ with an unknown permutation.
Under this model, typically the goal is to achieve the so-called \emph{exact alignment}, i.e., recovering the unknown permutation and thus revealing the correspondence for all vertices exactly.

A fundamental question in the graph alignment problem is: \emph{when is exact alignment possible?}  More specifically, \emph{what conditions on the statistical properties of the graphs are required for achieving exact alignment when given unbounded computational resources?}
Such conditions, usually referred to as \emph{information-theoretic limits}, have been established for the \er\ graph pair in a line of work \cite{Ped-Gro-privacy2011,Cul-Kiy-improved2016,settling-TIT,Cul-Kiy-exact2017}. 
{The best known information-theoretic limits are proved in~\cite{settling-TIT,Cul-Kiy-exact2017}, where the authors establish nearly matching achievability and converse bounds.}

In many real-world applications, additional information about the anonymized vertices might be available. For example, Facebook has user profiles on their website about each user's age, birthplace, hobbies, etc.
Such associated information is referred to as attributes (or features), which, unlike  user identities, are often publicly available.
Then a natural question to ask is: \emph{Can the attribute information help recover the vertex correspondence?} If so, \emph{can we quantify the amount of benefit brought by the attribute information?}
The value of attribute information has been demonstrated in the work of aligning Netflix and IMDb users by \citet{Nar-Shm-robust2008}.  They successfully recovered some of the user identities in the anonymized Netflix dataset based only on users' ratings of movies, without any information on the relationship among users. 
%In the example of aligning Netflix and IMDb databases, where each user's ratings on movies are considered as attributes of that user, \citet{Nar-Shm-robust2008} successfully recovered some of the user identities in the anonymized Netflix dataset by only comparing attributes of users without knowing the relationship network among users. 
In this paper, we incorporate attribute information to generalize the graph alignment problem. We call this problem the \emph{attributed graph alignment} problem.

%* Briefly describe our model and relate it other models (Correlated Erdos--Renyi, seeded graph alignment, database alignment)\\
%\ning{---modified intro---}

To investigate the attributed graph alignment problem, we extend the current \er\ graph pair model and we refer to this new random graph model as the attributed \er\ pair model $\mathcal{G}(n,\bm{p};m,\bm{q})$.
For a pair of graphs, $G_1$ and $G_2$, generated from the attributed \er\ pair model, each graph contains $n$ \emph{user vertices} and $m$ \emph{attribute vertices} (see Figure~\ref{fig:model}). 
Here, the user vertices represent the entities that need to be aligned; while the attribute vertices are all pre-aligned, reflecting the public availability of the attribute information. 
There are two types of edges in each graph, i.e., edges between user vertices and edges between user vertices and attribute vertices.
Here, edges between user vertices represent the relationship between users (e.g., friendship relations in a social network); edges between user vertices and attribute vertices encode the side information attached to each user (e.g., user profiles in a social network).
% there are no edges between attribute vertices. 
These two types of edges are correlatedly generated in the following way:
for a user-user vertex pair $(i,j)$, the edges connecting them follow a distribution $\bm{p}=(p_{11}, p_{10}, p_{01}, p_{00})$, where $p_{11}$ is the probability that $i$ and $j$ are connected in both $G_1$ and $G_2$, and $p_{10}, p_{01}, p_{00}$ represent the three remaining cases respectively: $i,j$ are only connected in $G_1$, only connected in $G_2$, and not connected in neither $G_1$ nor $G_2$;
for a user-attribute vertex pair, the edges connecting them are generated in a similar way following a distribution $\bm{q}=(q_{11}, q_{10}, q_{01}, q_{00})$.
This random process creates an identically labeled graph pair $(G_1,G_2)$ with similarity in both the graph topology part (user-user edges) and the attribute part (user-attribute edges). 
The graph $G_2$ is then anonymized by applying a random permutation on its \emph{user vertices} and the anonymized graph is denoted as $G_2'$. Under this formulation, our goal of attributed graph alignment is to recover this unknown permutation from $G_1$ and $G_2'$ by exploring both the topology similarity and attribute similarity. 
\begin{figure}
    \centering
    \fontsize{8pt}{2pt}
    \def\svgwidth{0.5\columnwidth}
    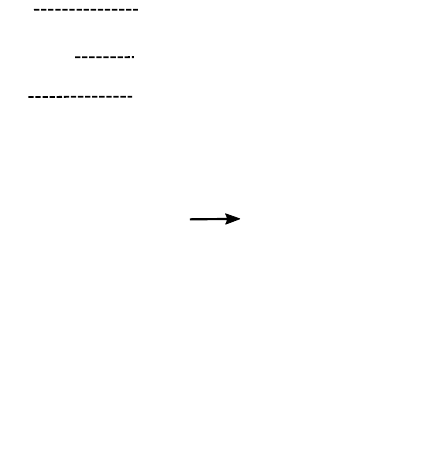
    \vspace{-1.2in}
    \caption{Example of attributed \er\ graph pair: Graph $G_1$ and $G_2$ are generated on the same set of vertices. Anonymized graph $G_2'$ is obtained through applying $\Pi^* = (1)(2,3)$ only on $\va$ of $G_2$ (permutation $\Pi^*$ is written in cycle notation).\vspace{-0.25in}}
    \label{fig:model}
\end{figure}

Under our attributed \er\ pair model, we use the maximum a posterior (MAP) estimator for aligning $(G_1,G_2')$, and establish the achievability and converse results for exact alignment. 
To get an intuitive understanding of how the existence of attribute information contributes to exact graph alignment, we present a simplified result by restricting the graph parameters to a certain regime, while deferring the general result to Section~\ref{sec:main-result}. In this regime we assume that the correlation coefficient of the user-user edges is at least $\Omega(\frac{(\log n)^2}{\sqrt{n}})$ and correlation coefficient of the user-attribute edges is at least $\Omega(\frac{(\log n)^{3/2}}{\sqrt{m}})$. Together with two other conditions on the edge sparsity, we establish the following asymptotically matching achievability and converse results as $n\rightarrow \infty$ (See Corollary~\ref{Coro:achievability} for the formal statement).
\begin{itemize}[leftmargin=14pt]
    \item If $n p_{11} +mq_{11}-\log{n} \rightarrow \infty$, then there exists an algorithm that achieves exact alignment with high probability (w.h.p.).
    \item If $np_{11}+mq_{11}-\log{n} \rightarrow -\infty$, then no algorithm guarantees exact alignment w.h.p.
\end{itemize}
The achievability and converse results are {illustrated} in Figure~\ref{fig:feasible-infeasible-region}. Here, $np_{11}$ is the average number of common users between $G_1$ and $G_2$ that are connected to an identical user vertex, and $mq_{11}$ is the average number of common attributes. 
Intuitively, the key quantity $np_{11}+mq_{11}$ (average common vertex degree) quantifies the topology and attribute similarity between $G_1$ and $G_2$.
The above results simply show that if this similarity measure is large enough, then exact alignment is achievable, or otherwise no algorithm can exactly recover the true alignment.
It is also worth noting that the average common vertex degree in attribute $mq_{11}$ highlights the extra benefit from attribute information, compared to the achievability result $np_{11} - \log n \to \infty$ when the attribute is not available.

%% Old Version 2022-Jul-26
% Here $np_{11}$ is the average number of users connected to a single user (average vertex degree in user) in the intersection graph $G_1 \cap G_2$ , and
% $mq_{11}$ is the average number of attributes connected to a single user in  $G_1 \cap G_2$ (average vertex degree in attribute). 
% Intuitively, $np_{11}+mq_{11}$ (overall average vertex degree in $G_1 \cap G_2$) quantifies the structure and attribute similarity between $G_1$ and $G_2$. The above results show that if this overall average vertex degree is large enough, then exact alignment is achievable, or otherwise no algorithm can exactly recover the true alignment.
% It is also worth noticing that the vertex degree in attribute out of the overall vertex degree highlights the extra benefit from attribute information. 
% The achievability and converse results are {illustrated} in Figure~\ref{fig:feasible-infeasible-region}. 
\begin{figure}
    \centering
    \fontsize{9pt}{9pt}
        \includegraphics[width=0.5\columnwidth]{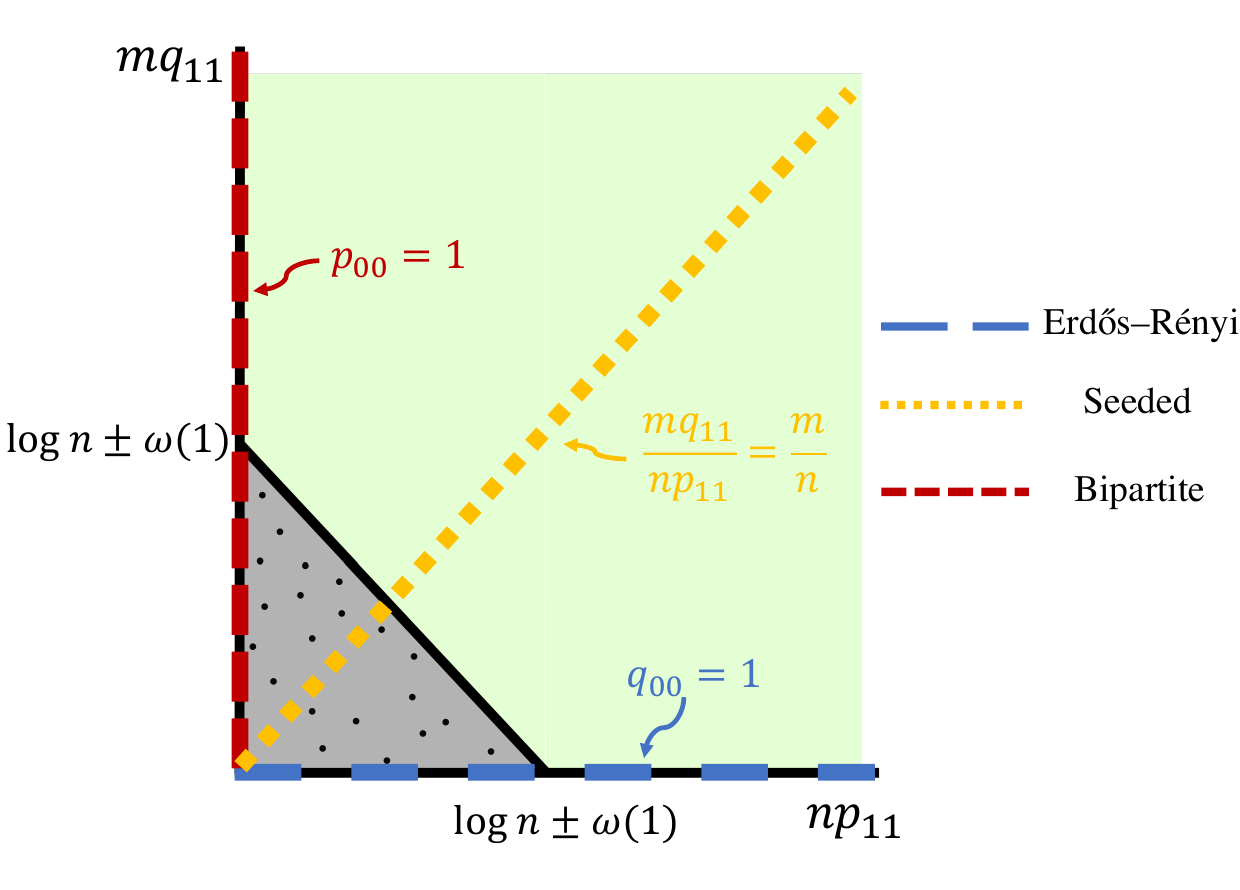}
    \caption{The green region in the figure is information theoretically achievable and the shaded grey region is not achievable. The three lines in the figure represent three specialized settings: 
    the blue line (correlated \er\ model) is obtained by setting $q_{00}=1$;
    the yellow line (seeded \er\ model) is obtained by setting $\bm{p}=\bm{q}$;
    the red line (correlated bipartite model) is obtained by setting $p_{00}=1$.
    Their intersections with the achievable and non-achievable region give the information-theoretic limits of the correlated \er\ model, seeded \er\ model and the correlated bipartite model separately.
    \vspace{-0.25in}
    }
    \label{fig:feasible-infeasible-region}
\end{figure}

From the information-theoretic limits we derive for the attributed graph pair, we could obtain information-theoretic limits on other existing random graph models as special cases (see Figure~\ref{fig:feasible-infeasible-region}). 
Below we highlight how the specializations of our results compare with the existing graph alignment literature under three specialized settings.
% Below we highlight how our results, by comparing with the three specialized settings, help answer some of the existing problems in the graph alignment literature.
The detailed comparison is given in Section~\ref{sec:comparison}. 
\begin{itemize}[leftmargin=14pt]
    \item 
    Specializing our model by setting %$m=0$, we remove the attribute vertices and obtain the correlated \er\ graph pair.
    % \ning{How about we just say $q_{00}=1$?}
    $q_{00}=1$, we remove the effect of the attribute vertices and get the correlated \er\ graph pair model.
    % Our specialized results recover the information-theoretic limits on \er\ graph alignment in \cite{settling-TIT,Cul-Kiy-exact2017}. 
    Our specialized results recover the information-theoretic limits on \er\ graph alignment in \cite{Cul-Kiy-exact2017}.
    % Comparing the specialized and un-specialized results allows us to quantify the benefit brought by the attribute information.
    \item Specializing our model by setting $\bm{p}=\bm{q}$, we can then treat the $m$ attribute vertices as pre-aligned user vertices and obtain the seeded \er\ model. 
    % Compared to the best-known information-theoretic limits in~\cite{Cul-Kiy-exact2017,Mos-Xu-seeded2020,wang-2022-feasible}, our specialized results strictly improves the best-known achievability and converse, establishing the tight threshold for the seeded graph alignment problem in the sparse regime.
    % Combining our specialized results and the best-known results in \cite{Cul-Kiy-exact2017,Mos-Xu-seeded2020,wang-2022-feasible}, the state-of-the-art information-theoretic limits of seeded graph alignment are strictly improved.
    Our specialized result reveals certain achievable and converse region that is unknown in the literature \cite{Cul-Kiy-exact2017,Mos-Xu-seeded2020,wang-2022-feasible}.
    % \sout{In contrast, existing work on the seeded problem mainly focuses on characterizations of computational feasibility bounds, and there are no known results for the information-theoretic limits.}
    \item Specializing our model by setting $p_{00}=1$, we remove all of the user-user edges and obtain the correlated bipartite graph pair model. 
    Our specialized achievability result recovers the best-known achievable region for bipartite alignment in the literature~\cite{Cul-Mit-fundamental2018}, and reveals certain converse region that is unknown in~\cite{Cul-Mit-fundamental2018}.

\end{itemize}

The main contributions of this paper are summarized as follows.
\begin{enumerate}[leftmargin=14pt]
    \item \emph{\textbf{Model Formulation.}} 
    We {propose} the attributed \er\ pair model, which incorporates both the graph topology similarity and the attribute similarity. Such model formulation allows us to align graphs with the assistance of publicly available side information.
    Moreover, our model serves as a unifying setting in the graph alignment literature and includes several popular models as its special cases.

    \item \emph{\textbf{Information theoretic limits.}} We {establish} achievability and converse results on exactly aligning random attributed graphs, where the conditions are tight under certain parameter regimes.
    Our results span the full spectrum from the traditional \er\ pair model where only the user relationship networks are available to models where only attribute information is available, unifying the existing results in each of these settings. 
    % When specialized to the seeded graph alignment and bipartite graph alignment models, our bounds strictly improve the best-known information-theoretic limits for these models.
    % When specialized to the seeded graph alignment and bipartite graph alignment models, by combining our bounds with the existing results, the state-of-the-art information-theoretic limits are strictly improved.
    When specialized to the seeded graph alignment and bipartite graph alignment models, our result reveals certain achievable and converse region that is unknown in the literature.
    \item \emph{\textbf{Proof techniques.}} The proof techniques for the achievability results are mainly inspired by the previous study on \er\ graph alignment \cite{Cul-Kiy-exact2017}. For the converse results, we study the phase-transition phenomenon on the existence of \emph{indistinguishable vertex pairs}, which may be of independent interest.

\end{enumerate}

\subsection{Related work}
% empirical
% seeded: noisy/accurate
% SBM with community label
% determinstic study 
% \ning{[All in present tense]}\\
The exact graph alignment problem has been studied under various random graph models.
One of the most popular random graph models is the correlated \er\ pair model   $\mathcal{G}(n,\bm{p})$, which generates simple graph pairs without any side information.
% The most popular setting in graph alignment without side information is the correlated \er\ pair model $\mathcal{G}(n,\bm{p})$.
Under this model, the optimal alignment strategy, derived from the MAP estimator, is enumerating all possible permutations in order to make the two graphs achieve the maximum edge overlap.
While the optimal strategy requires exponential time complexity, numerous studies have proposed polynomial-time approaches that exactly solve the graph alignment problem with high probability \cite{dai-2019-analysis,Ding-2020-efficient,Fan-2020-spectral,Mao-2021-constant,Mao2022}.

Here, we do not attempt to provide further detailed discussions on efficient algorithms, but focus on surveying the information-theoretic limits of exact alignment.
Currently, the best-known information-theoretic limits on \er\ graph alignment are shown in \cite{Cul-Kiy-exact2017} and \cite{settling-TIT} by analyzing error event of the MAP estimator. 
In \cite{Cul-Kiy-exact2017}, the authors prove achievability in the regime $n (\sqrt{p_{11}p_{00}}-\sqrt{p_{10}p_{01}})^2 \geq (2+\epsilon)\log n$. Under certain sparsity conditions, they also show that the achievable region can be improved to $np_{11} \geq \log n +\omega(1)$. In \cite{settling-TIT}, the authors consider a special case of the \er{} graph pair model called \emph{symmetric subsampling model}. In this model, it is assumed that 
\begin{equation}
\label{eq:subsamp_mod}
    p_{11}=ps^2,\; p_{01}=p_{10}=ps(1-s),\; p_{00}=1-2ps+ps^2
\end{equation}
for some $p,s\in [0,1]$.
Under this model, the authors prove the achievability in the regime $n(\sqrt{p_{11}p_{00}}-\sqrt{p_{10}p_{01}})^2 \geq (1+\epsilon)\log n$.
For the converse, \cite{Cul-Kiy-exact2017} proves that the permutation cannot be exactly recovered with high probability if $np_{11}\leq \log n -\omega(1)$ by showing the existence of isolated vertices in the intersection graph $G_1 \wedge G_2$. 
Under the general \er\ graph pair model, \cite{settling-TIT} shows the impossibility of exactly recovering the permutation in the regime $n(\sqrt{p_{11}p_{00}} - \sqrt{p_{10}p_{01}})^2 \leq (1-\epsilon)\log n$ by showing the existence of permutations that fails the MAP estimator by swapping two vertices.
To summarize the aforementioned results, matching achievability and converse for exact recovery is derived under certain sparsity assumptions in~\cite{Cul-Kiy-exact2017}, 
% and almost tight achievability and converse with a gap of $\Theta(\log n)$ is derived for the special case of symmetric subsampling model is presented in~\cite{settling-TIT}. 
and for the special case of the symmetric subsampling model, \cite{settling-TIT} provides almost tight achievability and converse bounds, with a gap of width $2\epsilon\log n$ between the established bounds.
Closing the gap for the general \er{} graph pair model is still an open problem.

Recently, there has been a growing interest in studying graph alignment with side information.
For example, in the seeded alignment setting, the side information appears in the form of a partial observation of the latent alignment.
For the seeded graph alignment problem, there have been a number of studies concentrating on designing polynomial-time algorithms with performance guarantees\cite{lyzinski-2014-seeded, Fis-Don-Ada-seeded2019, Mos-Xu-seeded2020}.
Some other more general settings treat any form of side information as vertex attributes and formulate this as the attribute graph alignment problem \cite{zhang-tong-2018attributed}.

% The side information, in the seeded alignment problems, appears in the form of partial observation of the latent alignment.
% There has been various studies focus on designing polynomial-time algorithms and providing performance guarantees for the seeded graph alignment problem  \cite{lyzinski-2014-seeded, Fis-Don-Ada-seeded2019, Mos-Xu-seeded2020}.
% Some other settings, which are more general than the seeded models, treat any form of side information as vertex attributes and formulate this as the attribute graph alignment problem.
There is a line of empirical studies on the attributed graph alignment \cite{zhang-tong-2016final, zhang-tong-2018attributed, zhou-ying-tong-2021attent}, yet, to the best of our knowledge, there is no known result on information-theoretic limits on graph alignment with attribute information. 

% We comment that our proposed attributed alignment problem can also be viewed as a graph alignment problem with part of vertices correctly pre-aligned, known as the \emph{seeded graph alignment} problem. Efficient algorithms and feasible regions for seeded graph alignment have been studied in \cite{KorLat_2014, Mos-Xu-seeded2020}. 

% \emph{\textbf{Reminder of the Landau notation.}}
% \begin{table}[!h]
% \centering
% \vspace{-0.1in}
% \begin{tabular}{ c c } 
%  \toprule
%  Notation & Definition\\ 
%  \midrule
%  $f(n)= \omega(g(n))$ & $\lim\limits_{n\rightarrow\infty} \frac{|f(n)|}{g(n)} = \infty$ \\
%  $f(n) = o(g(n))$ & $\lim\limits_{n \to \infty} \frac{|f(n)|}{g(n)} = 0$\\ 
%  $f(n)= O(g(n))$ & $\limsup\limits_{n\rightarrow\infty}\frac{|f(n)|}{g(n)} < \infty$\\
%  $f(n) = \Omega(g(n))$ & $\liminf\limits_{n\rightarrow \infty} \frac{|f(n)|}{g(n)} >0$\\
%  $f(n) = \Theta(g(n))$  & $f(n) = O(g(n))$ and $f(n) = \Omega(g(n))$\\
%  \bottomrule
% \end{tabular}
% \end{table}

% \section{related work}\label{sec:related}
% \input{sections/related}

\section{Model}\label{sec:model}

In this section, we describe the attributed \er\ graph pair model. 
{Under this model formulation, we formally define the exact attributed graph alignment problem.}
% We propose and formally define the attributed graph alignment problem under this model. 
An illustration of the model is given in Figure~\ref{fig:model}.

\emph{\textbf{User vertices and attribute vertices.}}
We first generate two graphs, $G_1$ and $G_2$, on the same vertex set $\mathcal{V}$.  The vertex set $\mathcal{V}$ consists of two disjoint sets of vertices, the \emph{user vertex set} $\vu$ and the \emph{attribute vertex set} $\va$, i.e., $\mathcal{V}=\vu\cup\va$. Assume that the user vertex set $\vu$ consists of $n$ vertices, labeled as $[n]\triangleq\{1, 2, 3,...,n\}$.  Assume that the attribute vertex set $\va$ consists of $m$ vertices, and $m$ scales as a function of $n$.

%\smallskip
\emph{\textbf{Correlated edges.}}
To describe the probabilistic model for edges in $G_1$ and $G_2$, we first consider the set of user-user vertex pairs $\eu\triangleq\vu\times\vu$ and the set of user-attribute vertex pairs $\ea\triangleq\vu\times\va$.  Then for each vertex pair $e\in \mathcal{E}\triangleq\eu\cup\ea$, we write $G_1(e)=1$ (resp.\ $G_2(e)=1$) if there is an edge connecting the two vertices in the pair in $G_1$ (resp.\ $G_2$), and write $G_1(e)=0$ (resp.\ $G_2(e)=0$) otherwise. Since we often consider the same vertex pair in both $G_1$ and $G_2$, we write $(G_1,G_2)(e)$ as a shortened form of $(G_1(e),G_2(e))$.

The edges of $G_1$ and $G_2$ are then correlatedly generated in the following way.  For each user-user vertex pair $e\in\eu$, $(G_1,G_2)(e)$ follows the joint distribution specified by
\begin{equation}
\label{Eq user pair distribution}
    (G_1,G_2)(e) =
    \begin{cases}
    (1,1)\;\ \text{w.p.} \;\ p_{11},\\
    (1,0)\;\ \text{w.p.} \;\ p_{10},\\
    (0,1)\;\ \text{w.p.} \;\ p_{01},\\
    (0,0)\;\ \text{w.p.} \;\ p_{00},
    \end{cases}
\end{equation}
where $p_{11},p_{10},p_{01},p_{00}$ are probabilities that sum up to~$1$. 
For each user-attribute vertex pair $e\in\ea$, $(G_1,G_2)(e)$ follows the joint probability distribution specified by
\begin{equation}
\label{Eq attribute pair distribution}
    (G_1,G_2)(e) =
    \begin{cases}
    (1,1)\;\ \text{w.p.} \;\ q_{11},\\
    (1,0)\;\ \text{w.p.} \;\ q_{10},\\
    (0,1)\;\ \text{w.p.} \;\ q_{01},\\
    (0,0)\;\ \text{w.p.} \;\ q_{00},
    \end{cases}
\end{equation}
where $q_{11},q_{10},q_{01},q_{00}$ are probabilities that sum up to~$1$. 
The correlation between $G_1(e)$ and $G_2(e)$ is measured by the correlation coefficient defined as
\begin{align*}
    \rho(e) &\triangleq \frac{\cov(G_1(e),G_2(e))}{\sqrt{\var[{G_1(e)}]}\sqrt{\var[G_2(e)]}},
\end{align*}
where $\cov(G_1(e),G_2(e))$ is the covariance between $G_1(e)$ and $G_2(e)$ and $\var[G_1(e)]$ and $\var[G_2(e)]$ are the variances. 
We assume that $G_1(e)$ and $G_2(e)$ are positively correlated, i,e., $\rho(e) > 0$ for every vertex pair $e$.
%\ning{A vertex pair $e$ in $G_1$ and $G_2$ represents the same entities, thus the correlation between $G_1(e)$ and $G_2(e)$ is non-negative, i,e., $\rho \geq 0$ for every vertex pair.}
Across different vertex pair $e$'s, the $(G_1,G_2)(e)$'s are independent.
Finally, recall that there are no edges between attribute vertices in our model.

For compactness of notation, we represent the joint distributions in \eqref{Eq user pair distribution} and \eqref{Eq attribute pair distribution} in the following matrix form:
\begin{align*}
    \bm{p} =
    \begin{pmatrix}
    p_{11} & p_{10}\\
    p_{01} & p_{00}
    \end{pmatrix}
    \quad \text{ and }\quad
    \bm{q} =
    \begin{pmatrix}
    q_{11} & q_{10}\\
    q_{01} & q_{00}
    \end{pmatrix}.
\end{align*}
We refer to the graph pair $(G_1,G_2)$ as an attributed \er\ pair $\mathcal{G}(n, \bm{p},m,\bm{q})$.
Note that this model is equivalent to the subsampling model in the literature \cite{Ped-Gro-privacy2011}.

\emph{\textbf{Anonymization and exact alignment.}}
In the attributed graph alignment problem, we are given $G_1$ and an anonymized version of $G_2$, denoted as $G_2'$.
The anonymized graph $G_2'$ is generated by applying a random permutation $\Pi^*$ on the user vertex set of $G_2$, where the permutation $\Pi^*$ is unknown.
More explicitly, each user vertex $i$ in $G_2$ is re-labeled as $\Pi^*(i)$ in $G_2'$.
The permutation $\Pi^*$ is chosen uniformly at random from $\mathcal{S}_n$, where $\mathcal{S}_n$ is the set of all permutations on $[n]$. Since $G_1$ and $G_2'$ are observable, we refer to $(G_1,G_2')$ as the \emph{observable pair} generated from the attributed \er\ pair $\mathcal{G}(n, \bm{p},m,\bm{q})$.

Then the graph alignment problem, i.e., the problem of recovering the identities/original labels of user vertices in the anonymized graph $G_2'$, can be formulated as a problem of estimating the underlying permutation $\Pi^*$.
The goal of graph alignment is to design an estimator $\hat{\pi}(G_1,G_2')$ as a function of $G_1$ and $G_2'$ to best estimate $\Pi^*$. We say $\hat{\pi}(G_1,G_2')$ achieves \emph{exact alignment} if $\hat{\pi}(G_1,G_2')=\Pi^*$. The probability of error for exact alignment is defined as $\prob(\hat{\pi}(G_1, G_2') \neq \Pi^*)$. We say exact alignment is achievable \emph{with high probability (w.h.p.)} if there exists $\hat{\pi}$ such that
$\lim_{n \to \infty}\prob(\hat{\pi}(G_1, G_2') \neq \Pi^*) = 0.$
% \begin{equation*}
% \lim_{n \to \infty}\prob(\hat{\pi}(G_1, G_2') \neq \Pi^*) = 0.
% \end{equation*}

\emph{\textbf{Reminder of the Landau notation.}}
\begin{table}[!h]
\centering
\vspace{-0.1in}
\begin{tabular}{ c c } 
 \toprule
 Notation & Definition\\ 
 \midrule
 $f(n)= \omega(g(n))$ & $\lim\limits_{n\rightarrow\infty} \frac{|f(n)|}{g(n)} = \infty$ \\
 $f(n) = o(g(n))$ & $\lim\limits_{n \to \infty} \frac{|f(n)|}{g(n)} = 0$\\ 
 $f(n)= O(g(n))$ & $\limsup\limits_{n\rightarrow\infty}\frac{|f(n)|}{g(n)} < \infty$\\
 $f(n) = \Omega(g(n))$ & $\liminf\limits_{n\rightarrow \infty} \frac{|f(n)|}{g(n)} >0$\\
 $f(n) = \Theta(g(n))$  & $f(n) = O(g(n))$ and $f(n) = \Omega(g(n))$\\
 \bottomrule
\end{tabular}
\vspace{-0.25in}
\end{table}

\section{Main results}\label{sec:main-result}

In this section, we state the achievability results (Theorem~\ref{Thm:achievability-general} and Theorem~\ref{Thm:achievability-sparse}) and the converse result (Theorem~\ref{Thm:converse}). To better demonstrate the benefit from attribute information, we also present a simplified version of the results 
under certain sparsity and correlation assumptions
% under some mild and practical assumptions 
as Corollary~\ref{Coro:achievability}.

Throughout the remainder of the paper, we define 
\begin{align}
\psiu &\triangleq (\sqrt{p_{11}p_{00}}-\sqrt{p_{10}p_{01}})^2\\
\psia &\triangleq (\sqrt{q_{11}q_{00}}-\sqrt{q_{10}q_{01}})^2.
\end{align}
\begin{theorem}[General achievability]
\label{Thm:achievability-general}
Consider the attributed \er\ pair $\mathcal{G}(n,\bm{p};m,\bm{q})$.
If
\begin{equation}
\label{Eq:general-feasible-region}
    \tfrac{1}{2}n\psiu + m\psia -\log{n} = \omega(1),
\end{equation}
% \sout{where $\psiu = (\sqrt{p_{11}p_{00}}-\sqrt{p_{10}p_{01}})^2$ and $\psia = (\sqrt{q_{11}q_{00}}-\sqrt{q_{10}q_{01}})^2$,}
then the MAP estimator achieves exact alignment w.h.p.
\end{theorem}
\begin{restatable}[Achievability in sparse region]{theorem}{LemmaSparse}
\label{Thm:achievability-sparse}
Consider the attributed \er\ pair $\mathcal{G}(n,\bm{p};m,\bm{q})$.
If
\begin{align}
    p_{11} &= O\left(\tfrac{\log{n}}{n}\right),\label{eq:lem31}\\ 
    p_{10}+p_{01} &= O\left(\tfrac{1}{\log{n}}\right),\label{eq:lem32}\\
    \frac{p_{10}p_{01}}{p_{11}p_{00}} &= O\left(\tfrac{1}{(\log{n})^3}\right),\label{eq:lem33}\\
    np_{11}+m \psia -\log{n} &= \omega(1),\label{eq:lem34}
\end{align}
then the MAP estimator achieves exact alignment w.h.p.
\end{restatable}

\begin{restatable}[Converse]{theorem}{ThmConverse}\label{Thm:converse}
Consider the attributed \er\ pair $\mathcal{G}(n, \bm{p}; m,\bm{q})$. If 
\begin{equation}
\label{eq:converse}
    -n\log(1-2p_{11}+2p_{11}^2)-m\log(1-2q_{11}+2q_{11}^2)-2\log n\rightarrow -\infty,
\end{equation}
then for any estimator, the probability of error is bounded away from zero.
\end{restatable}

To better illustrate the benefit of attribute information in the graph alignment problem, we present in Corollary~\ref{Coro:achievability} a simplified version of our achievability result by adding
% a few mild assumptions on the graph parameters motivated by practical applications. 
certain conditions on the sparsity and correlation of the two graphs.
To make the notation compact, we consider the equivalent expression from the subsampling model, where
\begin{align*}
    \begin{pmatrix}
    p_{11} & p_{10}\\
    p_{01} & p_{00}
    \end{pmatrix}
    &=\begin{pmatrix}
    p s_{\mathrm{u},1} s_{\mathrm{u},2} & p s_{\mathrm{u},1} (1-s_{\mathrm{u},2})\\
    p (1-s_{\mathrm{u},1}) s_{\mathrm{u},2} & p(1-s_{\mathrm{u},1})(1-s_{\mathrm{u},2}) + 1-p
    \end{pmatrix},
    % \\
    % \begin{pmatrix}
    % q_{11} & q_{10}\\
    % q_{01} & q_{00}
    % \end{pmatrix}
    % &=\begin{pmatrix}
    % q s'_1s'_2 & q s'_1 (1-s'_2)\\
    % q (1-s'_1) s'_2 & q(1-s'_1)(1-s'_2) + 1-q
    % \end{pmatrix}.
\end{align*}
and
\begin{align*}
    \begin{pmatrix}
    q_{11} & q_{10}\\
    q_{01} & q_{00}
    \end{pmatrix}
    &=\begin{pmatrix}
    q s_{\mathrm{a},1} s_{\mathrm{a},2} & q s_{\mathrm{a},1} (1-s_{\mathrm{a},2})\\
    q (1-s_{\mathrm{a},1}) s_{\mathrm{a},2} & q(1-s_{\mathrm{a},1})(1-s_{\mathrm{a},2}) + 1-q
    \end{pmatrix}.
    % \\
    % \begin{pmatrix}
    % q_{11} & q_{10}\\
    % q_{01} & q_{00}
    % \end{pmatrix}
    % &=\begin{pmatrix}
    % q s'_1s'_2 & q s'_1 (1-s'_2)\\
    % q (1-s'_1) s'_2 & q(1-s'_1)(1-s'_2) + 1-q
    % \end{pmatrix}.
\end{align*}
Under the subsampling model, the generation of $G_1$ and $G_2$ is modelled as a two-step random process. We first generate a base graph $G$, where an edge exists between each user-user pair with probability $p$ and an edge exists between each user-attribute pair with probability $q$. To generate graph $G_1$, each user-user edge in $G$ is kept with probability $s_{\mathrm{u},1}$ and each user-attribute edge in $G$ is kept with probability $s_{\mathrm{a},1}$. Similarly $G_2$ is generated by keeping each user-user edge in $G$ with probability $s_{\mathrm{u},2}$ and each user-attribute edge in $G$ with probability $s_{\mathrm{a},2}$. A random permutation is then applied on the user of $G_2$ to generate $G_2'$.
As a mild restriction on the sparsity of the graphs, we assume that the base graph edge probabilities $p$ and $q$ are not going to 1, i.e.,
\begin{align}
\label{eq:cond-nonedgeProb}
    1- p=\Theta(1), \\
    \label{eq:cond-attedgeProb}
    1- q=\Theta(1).
\end{align}
% In a typical social network, the degree of a vertex is much smaller than the total number of users. Based on this observation, we assume that the base graph edge probabilities $p$ and $q$ are not going to 1, i.e.,
% \begin{align}
% \label{eq:cond-nonedgeProb}
%     1- p=\Theta(1), \\
%     \label{eq:cond-attedgeProb}
%     1- q=\Theta(1).
% \end{align}
Moreover, we assume the following bounds on the vanishing speed of subsampling probabilities $s_{\mathrm{u},1}, s_{\mathrm{u},2}, s_{\mathrm{a},1}$ and $s_{\mathrm{a},2}$
\begin{equation}
    \label{eq:cond-rho}
    s_{\mathrm{u},1}s_{\mathrm{u},2}=\Omega\left(\frac{(\log n)^4}{n}\right).
\end{equation}
\begin{equation}
    \label{eq:cond-rhoa}
    s_{\mathrm{a},1}s_{\mathrm{a},2}=\Omega\left(\frac{(\log n)^3}{m}\right).
\end{equation}
\begin{corollary}[Simplified achievability]
\label{Coro:achievability}
Consider the attributed \er\ pair $\mathcal{G}(n,\bm{p};m,\bm{q})$.
Under conditions \eqref{eq:cond-nonedgeProb}-\eqref{eq:cond-rhoa}, we have tight achievability and converse for exact recovery. That is if 
\begin{equation}
\label{eq:coroeq0}
n p_{11} +mq_{11}-\log{n} \rightarrow \infty,
\end{equation}
then the MAP estimator achieves exact alignment \textit{w.h.p}, and if
\begin{equation}
\label{eq:coroeq1}
n p_{11} +mq_{11}-\log{n} \rightarrow -\infty,
\end{equation}
then the error probability of any estimator is bounded away from zero.
\end{corollary}

The proof of Corollary~\ref{Coro:achievability} can be found in Appendix~\ref{appx:coro-proof}. We visualize the matching achievability and converse results under conditions \eqref{eq:cond-nonedgeProb}-\eqref{eq:cond-rhoa} in Figure~\ref{fig:feasible-infeasible-region}. How to close the gap in the case where at least one of \eqref{eq:cond-nonedgeProb}-\eqref{eq:cond-rhoa} is not satisfied is an open problem.

% \begin{figure}
%     \centering
%      \fontsize{9pt}{9pt}
%         \includegraphics[width=0.5\columnwidth]{figs/ITlimit2.pdf}
%     \caption{The green region is information theoretically achievable and is specified by condition~\eqref{eq:coroeq2}; the shaded grey region in not achievable according the Theorem~\ref{Thm:converse}. The gap between the achievability and converse represent $q_{11}-\psia=O(q_{11}^{3/2})$. In particular, this gap is negligible up to $\pm\omega(1)$ when $m=\Omega((\log n)^3)$, which is the special case presented in Figure~\ref{fig:feasible-infeasible-region}.  }
%     \label{fig:corollary}
% \end{figure}

\section{Comparison}\label{sec:comparison}

% \ning{add references }
In this section, we specialize our main results (Theorems~\ref{Thm:achievability-general}, ~\ref{Thm:achievability-sparse}, and~\ref{Thm:converse}) on exact alignment of the attributed \er\ pair model to three closely related graph alignment problems: the \er\ graph alignment, the seeded \er\ graph alignment, and the bipartite graph alignment. We compare them with the best-known results in the literature. The main purpose of this comparison is to illustrate that our general results on the attributed graph alignment problem can recover most of the best-known existing results on these three specialized problems, and improve the state-of-the-art in certain cases. While it is possible that the best-known results can be further improved to get sharper bounds, this refinement is not the main focus of our comparison.

\subsection{\er\ graph pair}
The correlated \er\ pair model $\mathcal{G}(n,\bm{p})$ is the setting most commonly studied for graph alignment tasks that consider only graph topology similarity \cite{Ped-Gro-privacy2011, Cul-Kiy-improved2016,  settling-TIT,Cul-Kiy-exact2017, dai-2019-analysis}. This model generates graph pairs that contain only user vertices. For a pair of graphs $G_1$, $G_2$ obtained from this model $\mathcal{G}(n,\bm{p})$, we use $\vu$ to denote their vertex set and $|\vu|=n$. 
The edges in $G_1$ and $G_2$ are generated jointly in the following way: for a pair of users $e \in {\vu \choose 2}$, we have
\begin{equation}
\label{Eq ER pair distribution}
    (G_1,G_2)(e) =
    \begin{cases}
    (1,1)\;\ \text{w.p.} \;\ p_{11},\\
    (1,0)\;\ \text{w.p.} \;\ p_{10},\\
    (0,1)\;\ \text{w.p.} \;\ p_{01},\\
    (0,0)\;\ \text{w.p.} \;\ p_{00}.
    \end{cases}
\end{equation}
The anonymized graph $G_2'$ is obtained by applying a random permutation $\Pi^*$ on the vertices of $G_2$. This model can be specialized from the attributed graph pair model by setting the number of attributes $m = 0$ or $q_{00} = 1$. For aligning the correlated \er\ pair, the best-known information-theoretic limits are established in \cite{settling-TIT,Cul-Kiy-exact2017} and we state the combined results here for ease of comparison.

\begin{theorem}[Best-known information theoretic limits \cite{settling-TIT,Cul-Kiy-exact2017}]
\label{thm:sota ER}
Consider the correlated \er\ pair $\mathcal{G}(n,\bm{p})$.\\
\underline{Achievability:}\\
If 
\begin{align}
\label{eq:spER_A1}
    n \psiu \geq 2\log n + \omega(1),
\end{align}
or
\begin{align}
    p_{11} &= O\left(\tfrac{1}{\log{n}}\right),
    \label{eq:spER_A21}\\ 
    p_{10}+p_{01} &= O\left(\tfrac{1}{\log{n}}\right),
    \label{eq:spER_A22}\\
    \frac{p_{10}p_{01}}{p_{11}p_{00}} &= O\left(\tfrac{1}{(\log{n})^3}\right),
    \label{eq:spER_A23}\\
    np_{11} &= \log{n}+\omega(1),
    \label{eq:spER_A24}
\end{align}
then the MAP estimator achieves exact alignment w.h.p.\\
\underline{Converse:} \\
If there exist a constant $\epsilon \in (0,1)$ such that
\begin{align}
    n \psiu \leq (1-\epsilon) \log n, \label{eq:spERcon1}
\end{align}
or
\begin{align}
    np_{11} \leq \log n -\omega(1), \label{eq:spERcon2}
\end{align}
then for any estimator, the probability of error is bounded away from zero.
\end{theorem}

\begin{remark}
We point out that in the dense regime, i.e., at least one of conditions~\eqref{eq:spER_A21},~\eqref{eq:spER_A22} and~\eqref{eq:spER_A23} is not satisfied, Theorem 4 in~\cite{settling-TIT} provides a tighter achievability result. 
% However, the result is limited to the \emph{symmetric} \er{} pair model with $p_{01}=p_{10}$. 
However, as we mentioned in the introduction, the result is limited to the symmetric subsampling model~\eqref{eq:subsamp_mod}.
In this section, we focus on the comparison under the \emph{general} \er{} pair model, so the result from Theorem 4 in~\cite{settling-TIT} is not listed as one of the best known information theoretic limit. On the other hand, the converse result in~\cite{settling-TIT} is not limited to the symmetric subsampling model. Thus, we include the result as equation~\eqref{eq:spERcon1} in Theorem~\ref{thm:sota ER}.
\end{remark}
We now specialize the attributed \er{} pair model to the correlated \er{} pair by setting $q_{00} = 1$. Theorems~\ref{Thm:achievability-general}, ~\ref{Thm:achievability-sparse}, and~\ref{Thm:converse} simplify to the following.

\begin{theorem}[Specialization from attributed \er{} pair]
\label{thm:sp ER}
Consider the attributed \er\ pair $\mathcal{G}(n,\bm{p};m,\bm{q})$ with $q_{00}=1$.\\
\underline{Achievability:}\\
If
\begin{align}
\label{eq:spAER_A1}
    n \psiu \geq 2\log n + \omega(1),
\end{align}
or
\begin{align}
    p_{11} &= O\left(\tfrac{\log{n}}{n}\right),
    \label{eq:spAER_A21}\\ 
    p_{10}+p_{01} &= O\left(\tfrac{1}{\log{n}}\right),
    \label{eq:spAER_A22}\\
    \frac{p_{10}p_{01}}{p_{11}p_{00}} &= O\left(\tfrac{1}{(\log{n})^3}\right),
    \label{eq:spAER_A23}\\
    np_{11} &= \log{n}+\omega(1),
    \label{eq:spAER_A24}
\end{align}
then the MAP estimator achieves exact alignment w.h.p.\\
\underline{Converse:} 
If
\begin{align*}
    np_{11} \leq \log n -\omega(1),
\end{align*}
then for any estimator, the probability of error is bounded away from zero.
\end{theorem}

\begin{remark}
When specialized to the \er{} pair model, our achievability result
recovers the best-known achievability result from \cite{settling-TIT,Cul-Kiy-exact2017}, while our converse result is a strict subset of that given by conditions~\eqref{eq:spERcon1} and~\eqref{eq:spERcon2}.
% and our converse result is also a strict subset of that given by conditions~\eqref{eq:spERcon1} and~\eqref{eq:spERcon2}. }
To see the achievability results in Theorems~\ref{thm:sota ER} and~\ref{thm:sp ER} are equivalent, we observe that the difference between the region characterized by \eqref{eq:spER_A21}--\eqref{eq:spER_A24} and the region characterized by \eqref{eq:spAER_A21}--\eqref{eq:spAER_A24} is given by $p_{11}=\omega(\frac{\log n}{n})$ and $p_{11}=O(\frac{1}{\log n}).$ However, under the assumptions \eqref{eq:spAER_A22} and \eqref{eq:spAER_A23}, we know that $p_{00}=1-o(1)$ and $\frac{p_{10}p_{01}}{p_{00}p_{11}}=o(1)$. If $p_{11}=\omega(\frac{\log n}{n})$, then these further imply that $\psiu=(1+o(1))p_{11}p_{00}=\omega(\frac{\log n}{n})$, i.e., the difference between the two regions falls in the achievable region characterized by~\eqref{eq:spAER_A1}. Thus, the achievability region in Theorem~\ref{thm:sp ER} is exactly the same as that in Theorem~\ref{thm:sota ER}. 
For the converse, it is an open question whether a converse result for the attributed \er{} pair model can be established, which recovers condition~\eqref{eq:spERcon1} when specialized to the \er{} pair model.
\end{remark}

% Comparing the best-known result on \er\ pair alignment (Theorem~\ref{thm:sota ER}) with our specialized result (Theorem~\ref{thm:sp ER}), we can see that both our specialized achievability and converse region was contained by the best-known results.

\subsection{Seeded \er\ graph pair}
% \ning{compare to ER + assumptions on s=Theta(1)}
% \ning{use $m$ through out}\\
% % \ning{add reference}\\
% The seeded \er\ model, as a variant of the correlated \er\ model, is designed for graph alignment tasks when part of the vertices in the graph pair are already aligned.
% % with known side information about a partial alignment.
% Such partial alignment information may be formulated in a couple of different ways in different seeded models \cite{Kor-Lat-Reconciliation2014, Mos-Xu-seeded2020, Yar-Lyu-Gro-2013performance,lyzinski-2014-seeded,lubars-2018-correcting}, \ning{e.g., \cite{Kor-Lat-Reconciliation2014} studies a seeded model that has randomly generated seeds, \cite{Yar-Lyu-Gro-2013performance} studies a deterministic seeded model, etc.}
% \ning{Here, we focus on the deterministic seeded graph model $\mathcal{G}(n,m,\bm{p})$, which is closely related to our attributed \er\ setting.}

In the seeded graph model $\mathcal{G}(n,m,\bm{p})$, a pair of graphs $G_1, G_2$ are generated from the correlated \er\ pair model $\mathcal{G}(n+m,\bm{p})$. Then the anonymized graph $G_2'$ is obtained by applying a random permutation on the vertices of $G_2$. 
% Instead of observing only 
In addition to knowing $G_1$ and $G_2'$, 
in the seeded graph setting, we are also given
% also know 
the true alignment on a set of the user vertices, which is known as the seed set $\mathcal{V}_s$. The number of aligned pairs in $\mathcal{V}_s$ is a fixed number $m$. 
The seeded alignment problem has been studied by~\cite{Kor-Lat-Reconciliation2014, Mos-Xu-seeded2020, Yar-Lyu-Gro-2013performance,lyzinski-2014-seeded,lubars-2018-correcting}\footnote{In the literature, both random~\cite{Kor-Lat-Reconciliation2014} and deterministic~\cite{Yar-Lyu-Gro-2013performance} seed sets are considered. Here, we focus on the deterministic seed set setting which is closely related to our attributed \er\ pair model.}. Moreover, achievability results on unseeded graph alignment problem also trivially imply achievability results on seeded graph alignment problem. To the best of our knowledge, the best information-theoretic limits of the seeded alignment problem are given by~\cite{Cul-Kiy-exact2017,Mos-Xu-seeded2020, wang-2022-feasible}.
 For the simplicity of our discussion, we focus only on the 
symmetric subsampling model~\eqref{eq:subsamp_mod} with sparsity conditions
\begin{align}
    p=1-\Theta(1)\text{ and } s=\Omega\left(\frac{(\log n)^2}{\sqrt{n}}\right).\label{eq:sparsity}
\end{align}
\begin{theorem}[Best-known information-theoretic limits in the sparse and symmetric regime~\cite{Cul-Kiy-exact2017,Mos-Xu-seeded2020,wang-2022-feasible}]
\label{thm-sotaSeed}
Consider the seeded \er\ graph pair $\mathcal{G}(n,m,\bm{p})$ satisfying 
conditions~\eqref{eq:subsamp_mod} and~\eqref{eq:sparsity}.\\
\underline{Achievability from \cite{Cul-Kiy-exact2017}:} 
Assume 
\begin{align}
(n+m)p_{11}=\log(n+m)+\omega(1).
\label{eq:exi_ach1}
\end{align}
Then the unseeded MAP estimator achieves exact alignment w.h.p.\\
\underline{Achievability from~\cite{wang-2022-feasible}:} Assume $s=\Theta(1)$ and $p=o(1)$.
\begin{enumerate}
    \item In the regime where $mp_{11} = \Omega(\log n)$, if for a constant $\epsilon>0$, we have
    \begin{align}
       (n+m)p_{11} \geq (1+\epsilon) \log n,
       \label{eq:exi_ach2}
    \end{align}
    then the \textsc{AttrRich} algorithm in \cite{wang-2022-feasible} achieves exact alignment w.h.p.
    \item In the regime where $mp_{11} = o(\log n)$,
    if for a constant $\tau>0$, we have 
    \begin{align}
        np_{11} &=\log n+\omega(1),\\
        mp_{11} &\geq \frac{2\log n}{\tau \log (p_{11}/(p_{11}+p_{10})^2)},
    \end{align}
    then the \textsc{AttrSparse} algorithm in \cite{wang-2022-feasible} achieves exact alignment w.h.p. 
\end{enumerate}
\underline{Converse from \cite{Mos-Xu-seeded2020}} Consider the seeded \er\ graph pair $\mathcal{G}(n,m,\bm{p})$. If
\begin{align*}
    (n+m)p_{11}\leq \log{(n+m)}+O(1) \;\ \text{and} \;\ m=O(n),
\end{align*}
then for any estimator, the probability of error is bounded away from zero.
\end{theorem}
\begin{remark}[Efficient algorithms for seeded graph alignment]
    We comment that the seeded graph alignment algorithms proposed in~\cite{Mos-Xu-seeded2020} and~\cite{wang-2022-feasible} can be implemented polynomial-time, while the unseeded MAP estimator in~\cite{Cul-Kiy-exact2017} requires exponential time to implement. Under the seeded graph alignment problem, the best-known feasible range of graph parameters for achieving exact recovery by efficient algorithms is given together by \cite{Mos-Xu-seeded2020} and \cite{wang-2022-feasible}. 
    % Although the achievability results in~\cite{Mos-Xu-seeded2020} are not listed in Theorem~\ref{thm-sotaSeed} as best-known 
\end{remark}

To compare the best-known information-theoretic limits of the seeded \er\ alignment with our results, we specialize the attributed \er\ pair model by setting $\bm{p}=\bm{q}$, where $m$ attribute vertices are pre-aligned seeds. 
Notice that a small difference between the $\mathcal{G}(n,\bm{p}; m ,\bm{p})$ model and the seeded model $\mathcal{G}(n,m,\bm{p})$ is that there are no edges between the seeds in the specialized model but those edges exist in the seeded model. 
Such distinction may lead to a difference in the design of seeded graph alignment algorithms (e.g. algorithms from \cite{Mos-Xu-seeded2020} exploit seed-seed edges). It turns out that such seed-seed edges have no influence on the optimal MAP estimators for the two models, which leads to the next lemma.
% Thus, we have the following Lemma~\ref{lemma:seed-attri}.
\begin{lemma}
\label{lemma:seed-attri}
The information-theoretic limits on exact alignment in the seeded \er\ pair model $\mathcal{G}(n,m,\bm{p})$ and  the information-theoretic limits on exact alignment in the specialized attributed \er\ pair model $\mathcal{G}(n,\bm{p};m,\bm{p})$ are identical. 
\end{lemma}
\begin{proof}
See Appendix~\ref{appx:PfLem1}. 
\end{proof}

Based on Lemma~\ref{lemma:seed-attri}, we directly obtain the achievability and converse results on seeded graph alignment from Theorems~\ref{Thm:achievability-general},~\ref{Thm:achievability-sparse}, and~\ref{Thm:converse} by setting $\bm{p}=\bm{q}$. 
In the following, we demonstrate that the specialized result reveals certain achievable and converse region for seeded graph alignment that is unknown in the literature.

\begin{theorem}[Specialization from attributed \er{} pair]
\label{thm:sp Seed}
Consider the attributed \er\ pair $\mathcal{G}(n,\bm{p};m,\bm{p})$.\\
\underline{Achievability:}\\
If
\begin{align}
\label{eq:spSeed_A1}
    (n+m) \psiu \geq 2\log n + \omega(1),
\end{align}
or
\begin{align}
    p_{11} &= O\left(\tfrac{\log{n}}{n}\right),
    \label{eq:spSER_A21}\\ 
    p_{10}+p_{01} &= O\left(\tfrac{1}{\log{n}}\right),
    \label{eq:spSER_A22}\\
    \frac{p_{10}p_{01}}{p_{11}p_{00}} &= O\left(\tfrac{1}{(\log{n})^3}\right),
    \label{eq:spSER_A23}\\
    np_{11} +m\psiu &= \log{n}+\omega(1),
    \label{eq:spSER_A24}
\end{align}
then the MAP estimator achieves exact alignment w.h.p.\\
\underline{Converse:} \\
If
\begin{align*}
    (n+m)p_{11} \leq \log n -\omega(1),
\end{align*}
% or there exists a constant $\epsilon$ such that 
% \begin{align*}
%     (n+m)p_{11}p_{00} \leq (1-\epsilon) \log n,
% \end{align*}
then for any estimator, the probability of error is bounded away from zero.
\end{theorem}

% Moreover, if the two seeded graphs $G_1$ and $G_2'$ satisfy the two assumptions on typical social networks (see Section~\ref{sec:main-result}), then we have the following matching achievability and converse results. 
%(\ning{see proof Appendix~\ref{proof-seed_tight}}).
Moreover, if the two seeded graphs $G_1$ and $G_2'$ satisfy the two assumptions \eqref{eq:subsamp_mod} and~\eqref{eq:sparsity}, then we have the following matching achievability and converse results.
% \begin{corollary}[Threshold for sparse seeded \er\ pair]
% \label{coro:seed_tight}
% Consider the seeded \er\ pair model $\mathcal{G}(n,m,\bm{p})$
% under conditions \eqref{eq:cond-nonedgeProb} and \eqref{eq:cond-rho}. \\
% \underline{Achievability:}
% If 
% \begin{align}
% \label{eq:ach-seed}
%     (n+m)p_{11} \geq \log{n}+\omega(1),
% \end{align}
% then the MAP estimator achieves exact alignment w.h.p.\\
% \underline{Converse:}
% If 
% \begin{align}
% \label{eq:conv-seed}
%     (n+m)p_{11} \leq \log{n}-\omega(1),
% \end{align}
% then no algorithm guarantees exact alignment w.h.p.\\
% \end{corollary}
\begin{corollary}[Threshold for sparse seeded \er\ pair]
\label{coro:seed_tight}
Consider the seeded \er\ pair model $\mathcal{G}(n,m,\bm{p})$
under conditions \eqref{eq:subsamp_mod} and~\eqref{eq:sparsity}. \\
\underline{Achievability:}
If 
\begin{align}
\label{eq:ach-seed}
    (n+m)p_{11} \geq \log{n}+\omega(1),
\end{align}
then the MAP estimator achieves exact alignment w.h.p.\\
\underline{Converse:}
If 
\begin{align}
\label{eq:conv-seed}
    (n+m)p_{11} \leq \log{n}-\omega(1),
\end{align}
then for any estimator, the probability of error is bounded away from zero.
\end{corollary}
\begin{proof}
See Appendix~\ref{proof-seed_tight}.
\end{proof}

\begin{remark}
In Corollary~\ref{coro:seed_tight}, we obtain asymptotically tight achievability and converse
for seeded graph alignment under the symmetric subsampling model satisfying conditions \eqref{eq:subsamp_mod} and~\eqref{eq:sparsity}.
\end{remark}
\begin{remark}[Comparison between achievability results]
 The achievability result in Corollary~\ref{coro:seed_tight}
 reveals certain achievable region that is unknown in the literature~\cite{Cul-Kiy-exact2017,wang-2022-feasible}.
 % includes certain region that is not covered by the best-known achievability results for seeded alignment~\cite{Cul-Kiy-exact2017, wang-2022-feasible}.
 In the regime of $mp_{11}=\Omega(\log n)$, if $(n+m)p_{11}$ is at least $\log n+\omega(1)$ but less than $\log(n+m)+\omega(1)$ and $(1+\epsilon)\log n$ for any constant $\epsilon$, exact alignment is known to be achievable by Corollary~\ref{coro:seed_tight}, but not by Theorem~\ref{thm-sotaSeed}.
% The achievable region in Corollary~\ref{coro:seed_tight} strictly contains the union of feasible regions of all best-known seeded alignment algorithms from \cite{Mos-Xu-seeded2020, wang-2022-feasible}. 
In the following, we present an example which is in the achievable region of Corollary~\ref{coro:seed_tight}, but not in the achievable region of Theorem~\ref{thm-sotaSeed}. 
% Consider the sparse symmetric subsampling model satisfying~\eqref{eq:subsamp_mod} and~\eqref{eq:sparsity}.
% \ning{[I feel it need to be justified that for the example chosen, those condition are valid for some $p_{10}$ and $p_{01}$. That's why in the earlier version, we have also specified $p_{10}$, $p_{01}$.]}
Assume that 
\begin{align*}
    m=n^2,\; p_{11}=\frac{\log n+\log\log n}{m+n},\text{ and } p_{01} = p_{10} = 0.
\end{align*}
% m=n^2, p_{11}=\frac{\log n+\log\log n}{m+n}, p_{01} = p_{10} = 0$. 
We see that condition~\eqref{eq:ach-seed} holds because $(n+m)p_{11}=\log n+\log \log n=\log n+\omega(1)$. 
% Moreover, we have $1-(p_{11}+p_{01}) = 1-(p_{11}+p_{10}) = 1-\frac{\log n+\log\log n}{n^2+n} = \Theta(1)$ and $\rho_{\rm{u}}= \rho_{\rm{a}} = 1$, satisfying conditions~\eqref{eq:cond-nonedgeProb} and \eqref{eq:cond-rho}. 
Moreover, condition~\eqref{eq:sparsity} is satisfied because $s=1$ in this case.
However, condition~\eqref{eq:exi_ach1} $(n+m)p_{11}=\log(n+m)+\omega(1)$ in Theorem~\ref{thm-sotaSeed} does not hold because $(n+m)p_{11}<2\log n<\log(n+m)+\omega(1)$, and condition~\eqref{eq:exi_ach2} $(n+m)p_{11} \ge (1+\epsilon)\log n$ %\ning{(should here be $\geq$?) }
in Theorem~\ref{thm-sotaSeed} does not hold because $(n+m)p_{11}=\log n+\log\log n<(1+\epsilon)\log n$ for any positive constant $\epsilon$.  So this example lies in the achievable region of Corollary~\ref{coro:seed_tight}, but not in that of Theorem~\ref{thm-sotaSeed}.

However, we comment that the improvement mentioned above is natural. In~\cite{Cul-Kiy-exact2017}, the seedless graph alignment problem is considered. The results in~\cite{Cul-Kiy-exact2017} is included in the comparison because its achievable region for seedless graph alignment trivially implies achievable region for seeded graph alignment. When specializing the region in~\cite{Cul-Kiy-exact2017} to the scenario of seeded graph alignment, both seed vertices and non-seed vertices are viewed as vertices to align, and hence resulting in the $\log(m+n)$ terms on the right-hand side of~\eqref{eq:exi_ach1}. However, the identities of the seed vertices are already known, and there are actually just $n$ non-seed vertices to align. This causes the natural improvement to the $\log n$ term on the right-hand side of~\eqref{eq:ach-seed}. Moreover, note that this improvement is non-negligible only when $m=\omega(n)$, i.e., the number of seeds greatly surpass the number of non-seeds. The improvement beyond the achievable region in~\cite{wang-2022-feasible} is natural as well. This is because the MAP estimation studied in this work requires exponential computational time, while the algorithms in~\cite{wang-2022-feasible} can be implemented in polynomial time.
To make the comparison fair, we also mention that there exists certain known achievable region that is not covered by Corollary~\ref{coro:seed_tight}.

\end{remark}
\begin{remark}[Comparison between converse results]
% Our converse result in Corollary~\ref{coro:seed_tight} strictly improves the best-known converse result in~\cite{Mos-Xu-seeded2020}. 
Corollary~\ref{coro:seed_tight} includes certain converse region in the regime $m=\omega(n)$, while the converse region in Theorem~\ref{thm-sotaSeed} is exclusive to the regime $m=O(n)$. To make the comparison fair, we comment that the converse bound of Theorem~\ref{thm-sotaSeed} strictly contains the converse region of Corollary~\ref{coro:seed_tight} in the regime of $m=O(n)$.
\end{remark}
% \ning{comment on tightness}

\subsection{Bipartite graph pair}
In the bipartite graph pair model $\mathcal{G}(n,m,\bm{q})$, each graph is a bipartite graph on two disjoint set of vertices, i.e., the user vertex set $\vu$ and attribute vertex set $\va$. The edges between the two set of vertices are generated in a correlated way: for $e \in \vu \times \va$
\begin{equation}
\label{Eq bipa pair distribution}
    (G_1,G_2)(e) =
    \begin{cases}
    (1,1)\;\ \text{w.p.} \;\ q_{11},\\
    (1,0)\;\ \text{w.p.} \;\ q_{10},\\
    (0,1)\;\ \text{w.p.} \;\ q_{01},\\
    (0,0)\;\ \text{w.p.} \;\ q_{00}.
    \end{cases}
\end{equation}
The anonymized graph $G_2'$ is obtained by applying a random permutation $\Pi^*$ only on the user vertex set of $G_2$. Aligning such correlated bipartite graph is also known as a special case of the database alignment problem~\cite{Cul-Mit-Kiy-fund2018}. The best-known information-theoretic limits of database alignment are studied in \cite{Cul-Mit-Kiy-fund2018}. We restate the achievability and converse results from~\cite{Cul-Mit-Kiy-fund2018} when specialized to the case of bipartite graph pair alignment in Theorem~\ref{thm:DB_bipt}.

\begin{theorem}[Best-known information-theoretic limits~\cite{Cul-Mit-Kiy-fund2018}]
\label{thm:DB_bipt}
Consider the bipartite graph pair model $\mathcal{G}(n, m,\bm{q})$. \\
% Let $Q=\begin{psmallmatrix} q_{00} & q_{01}\\ q_{10} & q_{11}\end{psmallmatrix}$.\\
\underline{Achievability:}\\
If
\begin{align}
    -\frac m2 \log(1-2\psia)\ge \log n+\omega(1),
    % -\frac m2\log(1-2(\sqrt{q_{11}q_{00}}-\sqrt{q_{10}q_{01}})^2)\ge \log n +\omega(1),
    \label{eq:db_ach}
\end{align}
then the MAP estimator achieves exact alignment w.h.p.\\
\underline{Converse:}\\
If
\begin{align}
    -\frac m2 \log(1-2\psia)\le (1-\Omega(1))\log n,
    % -\frac m2\log(1-2(\sqrt{q_{11}q_{00}}-\sqrt{q_{10}q_{01}})^2)\le (1-\Omega(1))\log n,
    \label{eq:db_conv}
\end{align}
then for any estimator, the probability of error is bounded away from zero.
\end{theorem}
\begin{remark}
In~\cite{Cul-Mit-Kiy-fund2018}, the left-hand side of both inequalities in Theorem~\ref{thm:DB_bipt} are stated in a  different, yet equivalent, form. To state it, we need to introduce two definitions. Let $A=\begin{psmallmatrix} A_{00} & A_{01}\\ A_{10} & A_{11}\end{psmallmatrix}$ be a $2\times 2$ probability matrix with all non-negative entries and $A_{00}+A_{01}+A_{10}+A_{11}=1$. We define $A^{\otimes k}$ to be a $2^k\times 2^k$ probability matrix with rows and columns both indexed by $\{0,1\}^k$, and for $a,b\in \{0,1\}^k$, we have $A^{\otimes k}_{a,b}=\Pi_{i=1}^k A_{a_i,b_i}$. Matrix $A^{\otimes k}$ is called the $k$th \emph{tensor product} of $A$. For a probability matrix $A$ and an integer $l\ge 2$, we define the order-$l$ \emph{cycle mutual information} of $A$ to be $I_l^{\circ}(A)=\frac{1}{1-l}\log\mathrm{tr}((ZZ^T)^l)$, where $Z$ is a matrix with same dimension as $A$ and $Z_{ij}=\sqrt{A_{ij}}$ for any index pair $(i,j)$. In~\cite{Cul-Mit-Kiy-fund2018}, the left-hand side of both inequalities are given as $\frac12 I^{\circ}_2(Q^{\otimes m})$, where $Q=\begin{psmallmatrix} q_{00} & q_{01}\\ q_{10} & q_{11}\end{psmallmatrix}$.
According to~\cite{Cul-Mit-Kiy-fund2018}, the cycle mutual information satisfies a nice property that $I^{\circ}_2(Q^{\otimes m})=mI^{\circ}_2(Q)$. Furthermore, we can calculate the $2$-cycle mutual information of $Q$ as $I^{\circ}_2(Q)=-\log(1-2\psia)$, which implies that $\frac12 I^{\circ}_2(Q^{\otimes m})=-\frac m2 \log(1-2\psia)$.
\end{remark}

To compare our results with the best-known database alignment information-theoretic limits, we specialize the attributed \er\ pair model to the bipartite graph pair by removing all of the edges between user vertices, i.e., $p_{00}=1$. Correspondingly, we obtain the following achievable and converse result on bipartite graph alignment from Theorems~\ref{Thm:achievability-general} and~\ref{Thm:converse}.
\begin{theorem}[Specialization from attributed \er{} pair]
\label{thm:sp Bipt}
Consider the attributed \er\ pair $\mathcal{G}(n,\bm{p};m,\bm{q})$ with $p_{00}=1$.\\
\underline{Achievability:}\\
If
\begin{align}
     m\psia \geq \log n + \omega(1),
     \label{eq:bipt_ach}
\end{align}
then the MAP estimator achieves exact alignment w.h.p.\\
\underline{Converse:}\\
If
\begin{align}
    m q_{11} \leq \log n -\omega(1),
    \label{eq:bipt_conv}
\end{align}
% or there exists a constant $\epsilon$ such that 
% \begin{align*}
%     m q_{11}q_{00} \leq (1-\epsilon) \log n,
% \end{align*}
then for any estimator, the probability of error is bounded away from zero.
\end{theorem}
\begin{remark}
    The MAP estimator for bipartite graph alignment can be implemented by solving an assignment problem using the Hungarian method within polynomial time~\cite{hungarian}.
\end{remark}

\begin{remark}[Comparison between achievable regions]
    The achievable region in Theorem~\ref{thm:sp Bipt} is a strict subset of the achievable region in Theorem~\ref{thm:DB_bipt}. This because $-\frac m2 \log(1-2\psia)>  m\psia$. However, in the derivation steps~\eqref{eqLem2_7}--\eqref{eqLem2_9} of Theorem~\ref{Thm:achievability-general}, we were replacing the logarithm term by applying $\log(1+x)\le x$. Without this step of loosening the bound, equation~\eqref{eq:bipt_ach} can be replaced by $-\frac m2 \log(1-2\psia)\ge \log n+\omega(1)$, which is the same as equation~\eqref{eq:db_ach}. Therefore, although the achievability region~\eqref{eq:bipt_ach} in Theorem~\ref{thm:sp Bipt} does not directly recover the achievable region~\eqref{eq:db_ach} in Theorem~\ref{thm:DB_bipt}, it can be improved to the achievable region~\eqref{eq:db_ach} by a slight modification of derivations steps~\eqref{eqLem2_7}--\eqref{eqLem2_9}.
\end{remark}
% \ning{How about the case when $\psia=O(1/\sqrt{m})$ (we have $m\psia^2=O(1)$)? say if $m=(\log n) ^{3/2}$, then it becomes $\psia=O(1/(\log n) ^{3/4})$ which includes $\psia = O(1/\log n)$ }
\begin{remark}[Comparison between converse regions]
The converse result~\eqref{eq:bipt_conv} in Theorem~\ref{thm:sp Bipt} includes certain region that is not covered by the best-known converse region for bipartite alignment~\eqref{eq:db_conv} in Theorem~\ref{thm:DB_bipt}. This new region stems from the difference between the $-\Omega(1)\log n$ term on the right-hand side of~\eqref{eq:db_conv} and the $-\omega(1)$ term on the right-hand side of~\eqref{eq:bipt_conv}.
To illustrate the improved region, consider an instance of the parameters satisfying $q_{11}=\frac{\log n-\log\log n}{m}$, $q_{00}=1-q_{11}$, $q_{01}=q_{10}=0$ and $m=\omega(\log n)$. Because $mq_{11}=\log n-\log\log n\le \log n-\omega(1)$, this instance of parameters is in the converse region given by~\eqref{eq:bipt_conv}. However, this instance is out of the converse region~\eqref{eq:db_conv}, because
\begin{align*}
    -\frac{m}{2}\log(1-2\psia)&=-\frac{m}{2}\log(1-2(\sqrt{q_{11}q_{00}}-\sqrt{q_{01}q_{10}})^2)\\
    &=-\frac{m}{2}\log(1-2q_{00}q_{11})\\
    &\ge \frac{m}{2}2q_{00}q_{11}\\
    &=(\log n-\log\log n)(1-o(1))\\
    &>(1-\Omega(1))\log n.
\end{align*}

To make the comparison fair, we comment that there also exists certain converse region given by~\eqref{eq:db_conv}, which is not covered by~\eqref{eq:bipt_conv}.
\end{remark}

\section{Proof of converse statement} \label{sec:proof-converse}
In this section, we give a detailed proof for Theorem~\ref{Thm:converse}. 
Let $(G_1,G_2)$ be an attributed \er\ pair $\mathcal{G}(n,\bm{p};m,\bm{q})$. In this proof, we will focus on the intersection graph of $G_1$ and $G_2$, denoted as $G_1 \wedge G_2$, which is the graph on the vertex set $\mathcal{V}=\vu\cup\va$ whose edge set is the intersection of the edge sets of $G_1$ and $G_2$.
We say a permutation $\pi$ on the vertex set $\vu$ 
is an automorphism of $G_1\wedge G_2$ if a user-user pair $(i,j)$ is in the edge set of $G_1\wedge G_2$ if and only if $(\pi(i),\pi(j))$ is in the edge set of $G_1\wedge G_2$, and a user-attribute pair $(i,a)$ is in the edge set of $G_1\wedge G_2$ if and only if $(\pi(i),a)$ is in the edge set of $G_1\wedge G_2$ i.e., if $\pi$ is \emph{edge-preserving}.  Note that an identity permutation is always an automorphism.  Let $\mathrm{Aut}(G_1\wedge G_2)$ denote the set of automorphisms of $G_1\wedge G_2$.  
By Lemma~\ref{lemma:Aut-error-prob} below, exact alignment cannot be achieved w.h.p.\ if $\mathrm{Aut}(G_1\wedge G_2)$ contains permutations other than the identity permutation.  This allows us to establish conditions for not achieving exact alignment w.h.p.\ by analyzing automorphisms of $G_1\wedge G_2$.

% \begin{lemma}
% [\cite{Cul-Kiy-improved2016}]
% \label{lemma:Aut-error-prob}
% Let $(G_1,G_2)$ be an attributed \er\ pair $\mathcal{G}(n,\bm{p};m,\bm{q})$.  Given $|\mathrm{Aut}(G_1 \wedge G_2)|$, the probability that MAP estimator succeeds is at most $\frac{1}{|\mathrm{Aut}(G_1 \wedge G_2)|}$.
% \end{lemma}

\begin{restatable}{lemma}{LemmaAUT}
\label{lemma:Aut-error-prob}
Let $(G_1,G_2)$ be an attributed \er\ pair $\mathcal{G}(n,\bm{p};m,\bm{q})$.  Given $|\mathrm{Aut}(G_1 \wedge G_2)|$, the probability that MAP estimator succeeds is at most $\frac{1}{|\mathrm{Aut}(G_1 \wedge G_2)|}$.
\end{restatable}
The proof of Lemma~\ref{lemma:Aut-error-prob} is deferred to Appendix~\ref{appx:PfLem2}.

In the proof of Theorem~\ref{Thm:converse}, we will further focus on the automorphisms given by swapping two user vertices.  To this end, we first define the following equivalence relation between a pair of user vertices.  We say two user vertices $i$ and $j$ ($i\neq j$) are \emph{indistinguishable} in $G_1\wedge G_2$, denoted as $i\equiv j$, if $(G_1\wedge G_2)((i,v))=(G_1\wedge G_2)((j,v))$ for all $v\in\mathcal{V}$.  It is not hard to see that swapping two indistinguishable vertices is an automorphism of $G_1\wedge G_2$, and thus $|\mathrm{Aut}(G_1 \wedge G_2)\setminus\{\textrm{identity permutation}\}|\ge|\{\textrm{indistinguishable vertex pairs}\}|$.  
In the lemma below, we establish the sharp threshold for the existence of indistinguishable vertex pair. 
\begin{lemma}[Sharp threshold for indistinguishable pair]
\label{lem:ind-pair}
    Let $(G_1,G_2)$ be an attributed \er\ pair $\mathcal{G}(n,\bm{p};m,\bm{q})$ and let $G=G_1\wedge G_2$. If 
    \begin{equation}
        \label{eq:ind-pair1}
        -n\log(1-2p_{11}+2p_{11}^2)-m\log(1-2q_{11}+2q_{11}^2)-2\log n\rightarrow -\infty,
    \end{equation}
    then with high probability, there exists at least one pair of indistinguishable vertices. 

    Conversely, if 
    \begin{equation}
        \label{eq:ind-pair2}
        -n\log(1-2p_{11}+2p_{11}^2)-m\log(1-2q_{11}+2q_{11}^2)-2\log n\rightarrow \infty,
    \end{equation}
    then with high probability, there exists no indistinguishable vertex pairs.
\end{lemma}

The proof of Lemma~\ref{lem:ind-pair} is deferred to the end of this section. We complete the proof of Theorem~\ref{Thm:converse} based on Lemma~\ref{lem:ind-pair} in the following.
% Therefore, in the proof below, we show that the number of such indistinguishable vertex pairs is positive with a large probability, which suffices for proving Theorem~\ref{Thm:converse}.

% \begin{reptheorem}{Thm:converse}
% Let $(G_1, G_2')$ be the observable pair generated from the attributed \er\ pair $\mathcal{G}(n, \bm{p},m,\bm{q})$. If 
% \begin{equation}
%     np_{11}+mq_{11}-\log{n} \rightarrow -\infty,
% \end{equation}
% then no algorithm guarantees exact alignment w.h.p.
% \end{reptheorem}
\ThmConverse*
\begin{remark}
    A novelty in the converse proof is that we show the existence of \emph{indistinguishable user pairs} under the attributed \er\ graph pair model, while the converse proof in~\cite{Cul-Kiy-exact2017} shows the existence of \emph{isolated vertices} under the \er\ graph pair model. The reason why we take a different approach from~\cite{Cul-Kiy-exact2017} is because~\eqref{eq:converse} cannot guarantee the existence of isolated user vertices with high probability. To see this, let us consider the example of $p_{11}=\frac{0.5\log n}{n}, m=\log n$ and $q_{11}=0.49$. In this case, \eqref{eq:converse} is satisfied because $-n\log(1-2p_{11}+2p_{11}^2)-m\log(1-2q_{11}+2q_{11}^2)\le 2np_{11}+2mq_{11}=1.98\log n$, where the inequality follows because $-\log(1-2x+2x^2)\le 2x$ for any $x\ge 0$. It is not hard to see that the expected number of isolated vertices can be calculated as $n(1-p_{11})^{n-1}(1-q_{11})^m$. It follows that
        \begin{align*}
    \prob(\exists \text{ an isolated vertex }) &\le \E[\# \text{ isolated vertices}]\\
        &=n(1-p_{11})^{n-1}(1-q_{11})^m\\
        &\le \exp(\log n-(n-1)p_{11}+\log 0.51 \times \log n)\\
        &=\exp(\log n-0.5\log n +\log 0.51 \times \log n)\\
        &\approx \exp(-0.17\log n)\rightarrow 0,
    \end{align*}
     i.e., the expected number of isolated vertices goes to zero. So with high probability, there exist no isolated vertices in the graph. 
\end{remark}
\begin{proof}[Proof of Theorem~\ref{Thm:converse}]

Let $X$ denote the number of indistinguishable user vertex pairs in $G$, i.e.,
 \begin{align*}
     X = \sum_{i<j\colon i,j\in\vu} \mathbb{1} \{i\equiv j\}.
 \end{align*}
 By Lemma~\ref{lem:ind-pair}, we know that $\prob(X=0)=o(1)$.
Now we derive an upper bound on the probability of exact alignment under the MAP estimator, which is also an upper bound for any estimator since MAP minimizes the probability of error.  Note that by Lemma~\ref{lemma:Aut-error-prob}, $\prob(\pi_\mathrm{MAP} = \Pi^* | X=x)\le \frac{1}{x+1}$, which is at most $1/2$ when $x\ge 1$. Therefore,
% \sout{Moreover, we have}
\begin{align}
    \prob(\pi_\mathrm{MAP}=\Pi^*) &= \prob(\pi_\mathrm{MAP} = \Pi^* | X=0) \prob(X=0)\nonumber\\
    &\quad + \prob(\pi_\mathrm{MAP} = \Pi^* | X\ge 1) \prob(X\ge 1)\nonumber\\
    %\label{Thm 2_5}
    &\leq \prob(X=0) + \frac{1}{2} \prob(X\ge 1)\nonumber\\
    %\label{Thm 2_6}
    &= \frac{1}{2} + \frac{1}{2} \prob(X=0),\nonumber
\end{align}
which goes to $1/2$ as $n\to\infty$ and thus is bounded away from $1$.  This completes the proof that no algorithm can guarantee exact alignment w.h.p.
\end{proof}
\begin{proof}[Proof of Lemma~\ref{lem:ind-pair}]
    Let $G_1$ and $G_2$ be an attributed \er\ pair $\mathcal{G}(n,\bm{p};m,\bm{q})$ and let $G=G_1 \wedge G_2$.
Let $X$ denote the number of indistinguishable user vertex pairs in $G$, i.e.,
 \begin{align*}
     X = \sum_{i<j\colon i,j\in\vu} \mathbb{1} \{i\equiv j\}.
 \end{align*}
 We will firstly show that $\prob(X=0)\to 0$ as $n\to \infty$ if the condition \eqref{eq:ind-pair1} in Lemma~\ref{lem:ind-pair} is satisfied.
 
 We start by upper-bounding $\prob(X=0)$ using Chebyshev's inequality
 \begin{align}
     \prob(X=0) %&\leq \prob(|X-\E[X]|\geq \E[X]) \nonumber\\
     \label{eq:X=0-upper-bd}
     \leq \frac{\var(X)}{\E[X]^2} 
     = \frac{\E[X^2] - \E[X]^2}{\E[X]^2}.
 \end{align}
 For the first moment term $\E[X]$, we have
 \begin{align}
    \label{eq:EX}
     \E[X] = \sum_{i<j} \prob (i \equiv j) = {n \choose 2}\prob (i \equiv j).
 \end{align}
 For the second moment term $\E[X^2]$, we expand the sum as
 \begin{align}
 %\label{Eq 2rd moment}
    \E[X^2]
    &= \E \Biggl[\sum_{i<j} \mathbb{1} \{i \equiv j\} \cdot \sum_{k<l} \mathbb{1} \{k \equiv l\} \Biggr] \nonumber\\
    %  &= \E \Biggl[\sum_{i<j} \mathbb{1} \{i \equiv j\} 
    %  + \sum_{\substack{i,j,k,l\colon i<j, k<l\\i=k,j\neq l}} \mathbb{1} \{i \equiv j \equiv l\}\nonumber\\
    %  &\mspace{48mu} + \sum_{\substack{i,j,k,l\colon\\ i<j=k<l}} \mathbb{1} \{i \equiv j \equiv l\}+ \sum_{\substack{i,j,k,l\colon\\ k<l=i<j}} \mathbb{1} \{i \equiv j \equiv k\}\nonumber\\
    %  &\mspace{48mu}+ \sum_{\substack{i,j,k,l\colon i<j, k<l\\i\neq k,j= l}} \mathbb{1} \{i \equiv k \equiv l\}\nonumber\\
    &= \E \Biggl[\sum_{i<j} \mathbb{1} \{i \equiv j\}+\mspace{-12mu} \sum_{\substack{i,j,k,l\colon i<j,k<l\\i,j,k,l\textrm{ are distinct}}} \mathbb{1}\{i \equiv j\} \cdot \mathbb{1} \{k \equiv l\}\nonumber\\[-.5em]
    &\mspace{48mu} + \mspace{-12mu}\sum_{\substack{i,j,k,l\colon i<j,k<l\\\textrm{$\{i,j\}$ and $\{k,l\}$ share one element}}}\mspace{-24mu}\mathbb{1}\{i \equiv j \equiv k \equiv l\} \Biggr]\nonumber\\
    & = {n \choose 2} \prob(i \equiv j) 
    +{n \choose 4}{4\choose 2} \prob(i \equiv j \text{ and } k \equiv l)
      \nonumber \\
    &\mspace{23mu} + 6{n \choose 3} \prob(i \equiv j \equiv k)
    ,\label{eq:ex2}
  \end{align}
 where $i,j,k,l$ are distinct in \eqref{eq:ex2}.
 With \eqref{eq:EX} and \eqref{eq:ex2}, the upper bound given by Chebyshev's inequality 
 in \eqref{eq:X=0-upper-bd} can be written as
 \begin{align}
  \label{eq:Chebshev-bound}
     \frac{\var(X)}{\E[X]^2} &= 
     %\frac{1}{\E[X]}
     \frac{2}{n(n-1)\prob(i \equiv j)}
     +\frac{4(n-2)}{n(n-1)} \frac{\prob(i \equiv j \equiv k)}{\prob(i \equiv j)^2} \nonumber\\
     &\mspace{23mu}+\frac{(n-2)(n-3)}{n(n-1)}\frac{\prob(i \equiv j \text{ and } k\equiv l)} {\prob(i \equiv j)^2} -1.
 \end{align}
 To compute $\prob(i \equiv j)$, we look into the event $\{i \equiv j\}$ which is the intersection of $A_1$ and $A_2$, where 
 $A_1=\{\forall v\in\vu\setminus\{i,j\},\, G((i,v))=G((j,v))\}$, 
 and $A_2=\{\forall u\in \va,\, G((i,u))=G((j,u))\}$.
 Recall that in the intersection graph $G=G_1\wedge G_2$, the edge probability is $p_{11}$ for user-user pairs and $q_{11}$ for user-attribute pairs. Therefore,
 \begin{align*}
     \prob(A_1) &= \sum_{i=0}^{n-2} {n-2 \choose i} p_{11}^{2i} (1-p_{11})^{2(n-2-i)}\\
     &=\left( p_{11}^2 + (1-p_{11})^2 \right)^{n-2},\\
     \prob(A_2) &=  \sum_{i=0}^{m} {m \choose i} p_{11}^{2i} (1-p_{11})^{2(m-i)} \\
     &=  \left(q_{11}^2 + (1-q_{11})^2 \right)^m.
 \end{align*}
 Since $A_1$ and $A_2$ are independent, we have
 \begin{align}
     &\mspace{23mu}\prob(i \equiv j)\\
     &= \prob(A_1)\prob(A_2)\nonumber\\
      &= \left( p_{11}^2 + (1-p_{11})^2 \right)^{n-2}
      \left(q_{11}^2 + (1-q_{11})^2 \right)^m\nonumber\\
     \label{eq:i=j}
     &= \left(1-2p_{11}+2p_{11}^2 \right)^{n-2} \left(1-2q_{11}+2q_{11}^2  \right)^{m}.
 \end{align}
 Similarly, to compute $\prob(i \equiv j \equiv k)$, we look into the event $\{i \equiv j \equiv k\}$ which is the intersection of events $B_0$, $B_1$ and $B_2$, where $B_0 = \{G((i,j))=G((j,k))=G((i,k)) \}$, $B_1 = \{\forall v \in \vu \setminus \{i,j,k\},\, G((i,v))=G((j,v))=G((k,v))\}$, and $B_2 = \{\forall u \in \va,\, G((i,u))=G((j,u))=G((k,u))\}$.
 Then, the probabilities of those three events are
 \begin{align*}
     & \prob(B_0) = p_{11}^3 + (1-p_{11})^3,\\
    %  \prob(B_1) &= \sum_{i=0}^{n-3} {n-3 \choose i} p_{11}^{3i}(1-p_{11})^{3(n-3-i)} = \left(p_{11}^3 + (1-p_{11})^3\right)^{n-3},\\
     & \prob(B_1)  = \left(p_{11}^3 + (1-p_{11})^3\right)^{n-3},\\
    %  \prob(B_2) &= \sum_{i=0}^{m} {m \choose i} q_{11}^{3i}(1-q_{11})^{3(m-i)} = \left(q_{11}^3 + (1-q_{11})^3 \right)^m.
     & \prob(B_2)= \left(q_{11}^3 + (1-q_{11})^3 \right)^m.
 \end{align*}
 Since the events $B_0$, $B_1$ and $B_2$ are independent, we have
 \begin{align*}
     &\mspace{23mu}\prob(i \equiv j \equiv k)\\
     &= \prob(B_0)\prob(B_1)\prob(B_2)\\
     &= \left(p_{11}^3 + (1-p_{11})^3\right)^{n-2}
     \left(q_{11}^3 + (1-q_{11})^3 \right)^m\\
     &= (1-3p_{11}+3p_{11}^2)^{n-2} (1-3q_{11}+3q_{11}^2)^m.
 \end{align*}
 To compute $\prob(i\equiv j \text{ and } k\equiv l)$, we look into the event $\{i\equiv j \text{ and } k\equiv l\}$ which is the intersection of $C_0$, $C_1$, $C_1'$, $C_2$ and $C_2'$, where $C_0 = \{ G(i,k) = G(j,k) = G(i,l) = G(j,l)\}$, $C_1 = \{\forall v \in \vu \setminus \{i,j,k,l\}, G(i,v) = G(j,v)\}$, $C_1' = \{\forall v \in \vu \setminus \{i,j,k,l\}, G(k,v) = G(l,v)\}$, $C_2 = \{\forall u \in \va, G(i,u) = G(j,u)\}$ and $C_2' = \{\forall u \in \va, G(k,u) = G(l,u)\}$. The probabilities of those events are
 \begin{align*}
     \prob(C_0) &= p_{11} ^6 + p_{11}^4(1-p_{11})^2 + p_{11}^2 (1-p_{11})^4 + (1-p_{11})^6,\\
     \prob(C_1) &= \prob(C_1') = (p_{11}^2 + (1-p_{11})^2) ^{n-4},\\
     \prob(C_2) &= \prob(C_2') = (q_{11}^2 + (1-q_{11})^2) ^m.
 \end{align*}
 Since $C_0$, $C_1$, $C_1'$, $C_2$ and $C_2'$ are independent, we have 
 \begin{align*}
     &\mspace{23mu}\prob(i \equiv j \text{ and } k \equiv l)\\
     &= \prob(C_0)\prob(C_1)\prob(C_1')\prob(C_2)\prob(C_2')\\
     &=\prob(C_0) (p_{11}^2 + (1-p_{11})^2) ^{2n-8} (q_{11}^2 + (1-q_{11})^2) ^{2m}.
 \end{align*}
 Now we are ready to analyze the terms in \eqref{eq:Chebshev-bound}. 
 % For the last two terms, note that $\frac{(n-2)(n-3)}{n(n-1)} \to 1$ and $\frac{\prob(i\equiv j \text{ and } k\equiv l)}{\prob(i \equiv j)^2} \to 1$. 
 % This is because condition~\ref{eq:ind-pair1} implies that $\log(p_{11}^2+(1-p_{11})^2)=o(1)$.
 We firstly focus on the fraction $\frac{\prob(i\equiv j \text{ and } k\equiv l)}{\prob(i \equiv j)^2}$, and show that it converges to zero. Note that condition~\eqref{eq:ind-pair1} implies that $-\log(p_{11}^2+(1-p_{11})^2)=o(1)$. This further implies that $p_{11}=o(1)$ or $1-p_{11}=o(1)$. As a result, we have $p_{11}^2+(1-p_{11})^2=1-o(1)$ and $\prob(C_0)=1-o(1)$. It follows that $\frac{\prob(i\equiv j \text{ and } k\equiv l)}{\prob(i \equiv j)^2}=\frac{\prob(C_0)}{(p_{11}^2+(1-p_{11})^2)^4}\rightarrow 1$.
 Therefore, we have $\frac{(n-2)(n-3)}{n(n-1)}\frac{\prob(i\equiv j \text{ and } k\equiv l)}{\prob(i \equiv j)^2} -1 \to 0$ as $n\to \infty$.
 Then we just need to bound the first two terms in \eqref{eq:Chebshev-bound}. For the first term $\frac{2}{n(n-1)\prob(i \equiv j)}$, plugging in the expression %of $\prob(i \equiv j)$ 
 in \eqref{eq:i=j} gives
 \begin{align}
     %\log{\E[X]}
     &\mspace{21mu}-\log{\frac{2}{n(n-1)\prob(i \equiv j)}}\nonumber\\
     &=
     2\log{n} + (n-2) \log{(1-2p_{11}+2p_{11}^2)} \nonumber\\
     & \quad + m\log{(1-2q_{11}+2q_{11}^2)} 
     + O(1) \nonumber\\
     \label{Thm 2_2}
     &= \omega(1).
 \end{align}
Here \eqref{Thm 2_2} follows from
condition~\eqref{eq:ind-pair1}. 
Therefore, the first term in \eqref{eq:Chebshev-bound} $\frac{2}{n(n-1)\prob(i \equiv j)}\to 0$ as $n\to\infty$.

Next, for the second term $\frac{4(n-2)}{n(n-1)} \frac{\prob(i \equiv j \equiv k)}{\prob(i \equiv j)^2}$ in \eqref{eq:Chebshev-bound}, we have
\begin{align}
    &-\log{\left(\frac{4(n-2)}{n(n-1)} \frac{\prob(i \equiv j \equiv k)}{\prob(i \equiv j)^2}\right)} \nonumber\\
    =& \log{n} - (n-2) \log{\left(\frac{1-3p_{11} + 3p_{11}^2} {(1-2p_{11} +2p_{11}^2)^2}\right)} \nonumber\\
    & \quad -m \log{\left(\frac{1-3q_{11} + 3q_{11}^2} {(1-2q_{11} +2q_{11}^2)^2}\right)} + O(1) \nonumber\\
    \label{Thm 2_3}
    {\geq}& \log{n} +\frac12 n \log(1-2p_{11}+2p_{11}^2) + \frac12 m \log(1-2q_{11}+2q_{11}^2) + O(1)\\
    \label{Thm 2_4}
    {=}& \omega(1).
\end{align}
Here \eqref{Thm 2_3} follows from
the inequality $-\log{\left(\frac{1-3x + 3x^2} {(1-2x +2x^2)^2}\right)} \ge \frac12 \log(1-2x+2x^2)$ for any $x \in [0,1]$. Equation \eqref{Thm 2_4} follows from the condition \eqref{eq:ind-pair1}. Hence, the second term in (\ref{eq:Chebshev-bound}) also converges to $0$ as $n\to\infty$, which completes the proof for $\prob(X=0)\to 0$ as $n\to\infty$.

Next, we show that $\prob(X\ge 1)=o(1)$ under condition \eqref{eq:ind-pair2}. From \eqref{eq:EX} and \eqref{eq:i=j}, we know that $$\E[X]=\binom{n}{2}\left(1-2p_{11}+2p_{11}^2 \right)^{n-2} \left(1-2q_{11}+2q_{11}^2  \right)^{m}.$$
By Markov's inequality, we have 
\begin{align}
    \prob(X\ge 1)&\le \E[X]\nonumber\\
    &\le \exp(2\log n+(n-2)\log (1-2p_{11}+2p_{11}^2)+m(1-2q_{11}+2q_{11}^2))\nonumber\\
    &=\exp(-\omega(1))=o(1)\label{eq:last},
\end{align}
where \eqref{eq:last} follows by condition~\eqref{eq:ind-pair2}.
\end{proof}

\section{Proof of the general achievability}
\label{sec:proof-achievability}

In this section, we establish the general achievability result in Theorem~\ref{Thm:achievability-general}. In this proof, we first simplify the the optimal estimator for exact alignment, the MAP estimator, to a minimum weighted distance estimator in Lemma~\ref{Lemma:MAP}. We then analyze the probability analyze the error probability that a wrong permutation $\pi$ has a lower weighted distance than the true underlying permutation in Lemma~\ref{Lemma:P(error of pi)}. The main idea in the proof of Lemma~\ref{Lemma:P(error of pi)} is first bounding the error probability with the probability generating function of the difference of weighted distance. The generating function is then further bounded by applying a cycle decomposition to the edge permutation induced by $\pi$. With these two key lemmas, the proof of Theorem~\ref{Thm:achievability-general} is then completed by a straightforward union bound argument. In the following, we state the two key lemmas and prove Theorem~\ref{Thm:achievability-general} based on them.

To state Lemma~\ref{Lemma:MAP}, we first introduce some basic notation for graph statistics needed in stating the MAP estimator. 
For any attributed graph $g$ on the vertex set $\vu\cup\va$ and any permutation $\pi$ over the user vertex set $\vu$, we use $\pi(g)$ to denoted the graph given by applying $\pi$ to $g$.  For any two attributed graphs $g_1$ and $g_2$ on $\vu\cup\va$, we consider the Hamming distance between their edges restricted to the user-user vertex pairs in $\eu$, denoted as
\begin{align}
\label{Eq:Delta_u}
    \Delta^{\mathrm{u}}(g_1,g_2) &=\mspace{-12mu} \sum_{(i,j)\in \mathcal{E}_{\mathrm{u}}} \mspace{-12mu} \mathbb{1} \{g_1((i,j)) \neq g_2((i,j))\};
\end{align}
and the Hamming distance between their edges restricted to the user-attribute vertex pairs in $\ea$, denoted as
\begin{align}
\label{Eq:Delta_a}
    \Delta^{\mathrm{a}}(g_1,g_2) &= \mspace{-12mu}\sum_{(i,v) \in \mathcal{E}_{\mathrm{a}}} \mspace{-12mu} \mathbb{1} \{g_1((i,v)) \neq g_2((i,v))\}.
\end{align}
\begin{restatable}{lemma}{LemmaMAP}
\label{Lemma:MAP}
Let $(G_1,G_2')$ be an observable pair generated from the attributed \er\ pair $\mathcal{G}(n,\bm{p};m,\bm{q})$. The MAP estimator of the permutation $\Pi^*$ based on $(G_1,G_2')$ simplifies to
\begin{align*}
    &\hat{\pi}_\mathrm{MAP}(G_1,G_2') \\
    &= \argmin_{\pi \in \mathcal{S}_n}\{w_1 \Delta ^{\mathrm{u}}(G_1,\pi^{-1}(G_2')) 
    + w_2 \Delta ^{\mathrm{a}}(G_1,\pi^{-1}(G_2'))\},
    \end{align*}
where $w_1 = \log\left(\frac{p_{11}p_{00}}{p_{10}p_{01}}\right)$, $w_2 = \log \left(\frac{q_{11}q_{00}}{q_{10}q_{01}}\right)$, and
\begin{align*}
    \Delta ^{\mathrm{u}}(G_1,\pi^{-1}(G_2')) &=\mspace{-12mu} \sum_{(i,j)\in \mathcal{E}_{\mathrm{u}}} \mspace{-12mu} \mathbb{1} \{G_1((i,j)) \neq G_2'((\pi(i),\pi(j)))\},\\
    \Delta ^{\mathrm{a}}(G_1,\pi^{-1}(G_2')) &= \mspace{-12mu}\sum_{(i,v) \in \mathcal{E}_{\mathrm{a}}} \mspace{-12mu} \mathbb{1} \{G_1((i,v)) \neq G_2'((\pi(i),v))\}.
\end{align*}
\end{restatable}
In the following lemma, we upper bound the probability that a permutation has a lower weighted distance than the identity permutation.
\begin{restatable}{lemma}{LemmaErrorPi}
\label{Lemma:P(error of pi)}
Let $(G_1,G_2)$ be an attributed \er\ pair $\mathcal{G}(n,\bm{p};m,\bm{q})$. For any permutation $\pi$, let
\begin{equation*}
\delta_\pi(G_1,G_2) \triangleq w_1 (\Delta^\mathrm{u}(G_1,\pi( G_2)) -\Delta^\mathrm{u}(G_1,G_2))
+ w_2 (\Delta^\mathrm{a}(G_1,\pi( G_2)) -\Delta^\mathrm{a}(G_1,G_2)).
\end{equation*}
Then when $\pi$ has $n-\Tilde{n}$ fixed points, we have
\begin{align*}
    \prob \left(\delta_\pi (G_1, G_2)\leq 0\right)
    \leq (1-2\psi_{\mathrm{u}})^{\frac{\Tilde{n}(n-2)}{4}}(1-2\psi_{\mathrm{a}})^{\frac{\Tilde{n}m}{2}}.
\end{align*}
\end{restatable}
We defer the proof of Lemma~\ref{Lemma:MAP} to Appendices~\ref{appd:pfMAP}, and defer the proof of Lemma~\ref{Lemma:P(error of pi)} to the end of this section. With these two lemmas, we are now ready to prove Theorem~\ref{Thm:achievability-general}.
\begin{proof}[Proof of Theorem~\ref{Thm:achievability-general}]
    Given the observable pair $(G_1,G_2')$, the error probability of MAP estimator can be upper-bounded as
\begin{align}
    &\prob(\hat{\pi}_\mathrm{MAP}(G_1,G_2') \neq \Pi^*) \nonumber\\
    &= \sum_{\pi^* \in \Sn}\prob(\hat{\pi}_\mathrm{MAP}(G_1,G_2') \neq \pi^* | \Pi^* =\pi^*) \prob(\Pi^* = \pi^*) \nonumber \\
    \label{eqLem2_1}
    &= \frac{1}{|\Sn|} \sum_{\pi^* \in \Sn} \prob(\hat{\pi}_\mathrm{MAP}(G_1,G_2') \neq \pi^* | \Pi^* =\pi^*) \\
    % \label{eqLem2_2}
    % =& \frac{1}{|\Sn|} \sum_{\pi^* \in \Sn} \prob(\pi_\mathrm{MAP}(G_1,\pi^* \circ G_2) \neq \pi^* | \Pi^* =\pi^*) \\
    \label{eqlem2_2}
    &= \prob\left(\hat{\pi}_\mathrm{MAP}(G_1,G_2) \neq \pi_\mathrm{id} | \Pi^* =\pi_\mathrm{id}\right) \\
    \label{eqlem2_3}
    &= \prob(\hat{\pi}_\mathrm{MAP}(G_1,G_2) \neq \pi_\mathrm{id})\\
    \label{eqLem2_4}
    &\leq  \prob(\exists \pi \in \Sn\setminus\{\pi_\mathrm{id}\}, \delta_{\pi}(G_1,G_2) \leq 0),
\end{align}
where $\pi_\mathrm{id}$ denotes the identity permutation, and
\begin{equation}
\label{Eq:def delta_pi}
\delta_\pi(G_1,G_2) \triangleq w_1 (\Delta^\mathrm{u}(G_1,\pi( G_2)) -\Delta^\mathrm{u}(G_1,G_2))
+ w_2 (\Delta^\mathrm{a}(G_1,\pi(G_2)) -\Delta^\mathrm{a}(G_1,G_2)).
\end{equation}

%Here \eqref{eqLem2_1} follows since $\Pi^*$ is uniformly drawn from $\Sn$, and thus $\prob(\Pi^*=\pi^*) = 1/|\Sn|$ for all $\pi^*$.
Here \eqref{eqLem2_1} follows from the fact that $\Pi^*$ is uniformly drawn from $\Sn$, which implies $\prob(\Pi^*=\pi^*) = 1/|\Sn|$ for all $\pi^*$;
\eqref{eqlem2_2} is due to the symmetry among user vertices in $G_1$ and $G_2$;
\eqref{eqlem2_3} is due to the independence between $\Pi^*$ and $(G_1,G_2)$;
\eqref{eqLem2_4} is true because by Lemma~\ref{Lemma:MAP},  $\pi_{\mathrm{MAP}}(G_1,G_2)$ minimizes the weighted distance, and  $\pi_{\mathrm{MAP}} \neq \pi_{\mathrm{id}}$ only if there exists a permutation $\pi$ such that $\pi \neq \pi_\mathrm{id}$ and $\delta_{\pi}(G_1, G_2) \leq 0$. 

Now to prove that \eqref{Eq:general-feasible-region} implies that the error probability in \eqref{eqLem2_4} converges to $0$ as $n \to \infty$, we
further upper-bound the error probability as follows
\begin{align}
     &\prob\left(\exists \pi \in \Sn \setminus\{\pi_\mathrm{id}\}, \delta_{\pi}(G_1,G_2) \leq 0\right) \nonumber\\
     \label{eqLem2_5}
     & \leq \sum_{\pi \in \Sn \setminus \{\pi_\mathrm{id}\}} \prob(\delta_{\pi}(G_1,G_2) \leq 0)\\
     \label{eqLem2_6}
     & =  \sum_{\Tilde{n}=2}^{n}\sum_{\pi \in \SnTn } \prob( \delta_{\pi} (G_1,G_2) \leq 0)\\
     & \leq \sum_{\Tilde{n}=2}^n |\SnTn| \max_{\pi \in \SnTn} {\{\prob( \delta_{\pi}(G_1,G_2) \leq 0)\}} \nonumber\\
     & \leq \sum_{\Tilde{n}=2}^n n^{\Tilde{n}}
     \max_{\pi \in \SnTn} {\{\prob( \delta_{\pi}(G_1,G_2) \leq 0)\}}. \nonumber
    %  (1-2\psiu)^{\frac{\Tilde{n}(n-2)}{4}}(1-2\psia)^\frac{\Tilde{n}m}{2}.
\end{align}
Here \eqref{eqLem2_5} follows from directly applying the union bound. In \eqref{eqLem2_6}, we use $\SnTn$ to denote the set of permutations on $[n]$ that contains exactly $(n-\Tilde{n})$ fixed points. In the example of Figure~\ref{fig:model}, the given permutation $\Pi^* = (1)(23)$ has 1 fixed point and $(1)(23) \in \mathcal{S}_{3,2}$. Furthermore, we have $|\SnTn|= \binom{n}{\Tilde{n}}%{n \choose \Tilde{n}} 
\left(!\Tilde{n}\right) \leq n^{\Tilde{n}}$, where $!\Tilde{n}$, known as the number of derangements, represents the number of permutations on a set of size $\Tilde{n}$ such that no element appears in its original position.

% Now, we state a lemma to obtain a closed-form upper bound for $\max_{\pi \in \SnTn} {\{\prob( \delta_{\pi}(G_1,G_2) \leq 0)\}}$ and defer its proof to Section~\ref{Proof:Lemma4}.

% \begin{lemma}
% \label{Lemma:P(error of pi)}
% Let $(G_1,G_2)$ be an attributed \er\ pair $\mathcal{G}(n,\bm{p};m,\bm{q})$. For any permutation $\pi$, let
% \begin{equation*}
% \delta_\pi(G_1,G_2) \triangleq w_1 (\Delta^\mathrm{u}(G_1,\pi( G_2)) -\Delta^\mathrm{u}(G_1,G_2))
% + w_2 (\Delta^\mathrm{a}(G_1,\pi( G_2)) -\Delta^\mathrm{a}(G_1,G_2)).
% \end{equation*}
% Then when $\pi$ has $n-\Tilde{n}$ fixed points, we have
% \begin{align*}
%     \prob \left(\delta_\pi (G_1, G_2)\leq 0\right)
%     \leq (1-2\psi_{\mathrm{u}})^{\frac{\Tilde{n}(n-2)}{4}}(1-2\psi_{\mathrm{a}})^{\frac{\Tilde{n}m}{2}}.
% \end{align*}
% \end{lemma}

With the upper bound in Lemma~\ref{Lemma:P(error of pi)}, we have
\begin{align}
    &\prob\left(\exists \pi \in \Sn \setminus\{\pi_\mathrm{id}\},
    \delta_{\pi}(G_1,G_2) \leq 0\right)\nonumber\\
    &\leq \sum_{\Tilde{n}=2}^n n^{\Tilde{n}}
     (1-2\psiu)^{\frac{\Tilde{n}(n-2)}{4}}(1-2\psia)^\frac{\Tilde{n}m}{2}\nonumber\\
    &= \sum_{\Tilde{n}=2}^n \left(n
     (1-2\psiu)^{\frac{n-2}{4}}(1-2\psia)^\frac{m}{2}\right)^{\Tilde{n}}.\label{eqLem2_7}
\end{align}
For this geometry series, the negative logarithm of its common ratio is 
\begin{align}
    &-\log{n} - \frac{n-2}{4} \log{(1-2\psiu)} -\frac{m}{2} \log{(1-2\psia)}\nonumber\\
    \label{eqLem2_8}
    &\geq -\log{n} + \frac{n-2}{2} \psi_u + m\psi_a\\
    \label{eqLem2_9}
    &= \omega(1).
\end{align}
Here we have $\psiu =(\sqrt{p_{11}p_{00}}-\sqrt{p_{10}p_{01}})^2 \leq 1/4$ and $\psia =(\sqrt{q_{11}q_{00}}-\sqrt{q_{10}q_{01}})^2 \leq 1/4$. Therefore, \eqref{eqLem2_8} follows from the inequality $\log{(1+x)} \leq x$ for $x>-1$.
Equation \eqref{eqLem2_9} follows from condition \eqref{Eq:general-feasible-region} by noting that $\psi_u$ is no larger than $1$.
Therefore, the geometry series in \eqref{eqLem2_7} converge to $0$ as $n\to \infty$. This completes the proof that MAP estimator achieves exact alignment w.h.p.\ under condition \eqref{Eq:general-feasible-region}. 
\end{proof}

\begin{proof}[Proof of Lemma~\ref{Lemma:P(error of pi)}]
To prove the upper bound on $\prob (\delta_\pi (G_1, G_2)\leq 0)$ in Lemma~\ref{Lemma:P(error of pi)}, we will use the method of \emph{generating functions}. We first introduce our construction of a generating function and how it can be used to bound $\prob (\delta_\pi (G_1, G_2)\leq 0)$. We then present several properties of generating functions (Facts~\ref{Fact:ogf cycle decomposition}, \ref{Fact:l and 2 cycle}, and \ref{Fact:ogf coefficient}), which will be needed to finish the proof of Lemma~\ref{Lemma:P(error of pi)}.

\textbf{\emph{A generating function for the attributed \er\ pair.}}
% A generating function encodes an sequence of numbers by treating them as the coefficients of a formal power series, and thus turns a counting problem into analysis on functions. We define generating function on a pair of graphs in a way that the powers of formal variables represent the graph statics $\bm{\muu}, \bm{\mua}$ and $\delta_\pi$. 
For any graph pair $(g,h)$ that is a realization of an attributed \er\ pair,
we define a $2 \times 2$ matrix $\bm{\muu}(g,h)$ as follows for user-user edges:
\begin{equation*}
    \bm{\muu}(g,h) = 
    \begin{pmatrix}
    \muu_{11}  & \muu_{10} \\
    \muu_{01}  & \muu_{00} \\
    \end{pmatrix},
\end{equation*}
where $\muu_{ij} = \sum_{e \in \mathcal{E}_{\mathrm{u}}} \mathbb{1}\{g(e) =i, h(e) =j\}$. Similarly, we define $\bm{\mua}(g,h)$ as follows for user-attribute edges:
\begin{equation*}
    \bm{\mua}(g,h) = 
    \begin{pmatrix}
    \mua_{11}  & \mua_{10} \\
    \mua_{01}  & \mua_{00} \\
    \end{pmatrix},
\end{equation*}
where $\mua_{ij} = \sum_{e \in \mathcal{E}_{\mathrm{a}}} \mathbb{1}\{g(e) =i, h(e) =j\}$.

Now we define a generating function for attributed graph pairs, which encodes information in a formal power series. 
Let $z$ be a single formal variable and $\bm{x}$ and $\bm{y}$ be $2\times2$ matrices of formal variables where
$$\bm{x} = \begin{pmatrix}
x_{00}& x_{01}  \\
x_{10}& x_{11} 
\end{pmatrix} \text{ and }
\bm{y} =\begin{pmatrix}
y_{00}& y_{01}  \\
y_{10}& y_{11} 
\end{pmatrix}. $$
Then for each permutation $\pi$, we construct the following generating function:
\begin{equation}
\label{Eq:generating func def}
    \mathcal{A}(\bm{x},\bm{y},z) = \sum_{g\in \{0,1\}^\mathcal{E}}\sum_{h \in \{0,1\}^\mathcal{E}} z^{\delta_{\pi}(g,h)}\bm{x}^{\bm{\muu}(g,h)} \bm{y}^{\bm{\mua}(g,h)},
\end{equation}
where
$$\bm{x}^{\bm{\muu}(g,h)} \triangleq 
x_{00}^{\mu_{00}}\cdot
x_{01}^{\muu_{01}}\cdot
x_{10}^{\muu_{10}}\cdot
x_{11}^{\muu_{11}},$$
$$\bm{y}^{\bm{\mua}(g,h)} \triangleq 
y_{00}^{\mua_{00}}\cdot
y_{01}^{\mua_{01}}\cdot
y_{10}^{\mua_{10}}\cdot
y_{11}^{\mua_{11}}.$$
Note that in the above expression of $\mathcal{A}(\bm{x},\bm{y},z)$, we enumerate all possible attributed graph pairs $(g,h)$ as realizations of the random graph pair $(G_1,G_2)$. For each realization, we encode the corresponding  $\bm{\muu}(g,h), \bm{\mua}(g,h)$ and $\delta_\pi(g,h)$ in the powers of formal variables $\bm{x},\bm{y}$ and $z$. By summing over all possible realizations $(g,h)$, the terms having the same powers are merged as one term. Therefore, the coefficient of a term $z^{\delta_{\pi}}\bm{x}^{\bm{\muu}}\bm{y}^{\bm{\mua}}$ represents the number of graph pairs that have the same graph statics represented in the powers of formal variables. 

\textbf{\emph{Bounding ${\prob(\delta_\pi (G_1, G_2)\leq 0})$ in terms of the generating function}.} 
%In the above definition of generating function, we have specially designed the power of formal variables as some graph statics.
%The reason we care about $\bm{\muu}, \bm{\mua}$ is that they are sufficient graph statistics for characterizing the joint probability. 
% More specifically, recall our attributed \er\ pair model $\mathcal{G}(n, \bm{p}; m,\bm{q})$ where $\prob((G_1,G_2) = (g,h)) = \bm{p}^{\bm{\muu}(g,h)}\bm{q}^{\bm{\mua}(g,h)}$.
% Therefore, by setting $\bm{x}=\bm{p}$ and $\bm{y}=\bm{q}$, we only have $z$ as formal variable and its coefficient is the probability that the discrete random variable $\delta_\pi(G_1,G_2)$ takes value of $d$, i.e., $\prob(\delta_\pi(G_1,G_2) = d) = [z^d]\mathcal{A}(\bm{p},\bm{q},z)$, where $[z^d]$ is the coefficient extraction operator and $[z^d]\mathcal{A}(\bm{p},\bm{q},z)$ represents the coefficient of $z^d$.
% In other words, $\mathcal{A}(\bm{p},\bm{q},z)$ is the \emph{probability generating function} of $\delta_\pi(G_1,G_2)$ with formal variable $z$, i.e., $\mathcal{A}(\bm{p},\bm{q},z) = \E[z^{\delta_{\pi}(G_1,G_2)}]$. We have
% \begin{equation}
% \label{Eq:P(delta<=0)}
%     \prob (\delta_\pi (G_1, G_2)\leq 0) = \sum_{d \leq 0}[z^d] \mathcal{A}(\bm{p},\bm{q},z).
% \end{equation}
We first argue that when we set $\bm{x}=\bm{p}$ and $\bm{y}=\bm{q}$, the generating function $\mathcal{A}(\bm{p},\bm{q},z)$ becomes the probability generating function of $\delta_\pi(G_1,G_2)$ for the attributed \er\ pair $(G_1,G_2)\sim \mathcal{G}(n, \bm{p}; m,\bm{q})$.  To see this, note that the joint distribution of $G_1$ and $G_2$ can be written as $\prob((G_1,G_2) = (g,h)) = \bm{p}^{\bm{\muu}(g,h)}\bm{q}^{\bm{\mua}(g,h)}$. Then by combining terms in $\mathcal{A}(\bm{p},\bm{q},z)$, we have $\prob(\delta_\pi(G_1,G_2) = d) = [z^d]\mathcal{A}(\bm{p},\bm{q},z)$, where $[z^d]\mathcal{A}(\bm{p},\bm{q},z)$ denotes the coefficient of $z^d$ with $[z^d]$ being the \emph{coefficient extraction operator}. We comment that the probability generating function here is defined in the sense that $\mathcal{A}(\bm{p},\bm{q},z) = \E\bigl[z^{\delta_{\pi}(G_1,G_2)}\bigr]$. Since $\delta_{\pi}(G_1,G_2)$ takes real values, this is slightly different from the standard probability generating function for random variables with nonnegative integer values.  But this distinction does not affect our analysis in a significant way since $\delta_{\pi}(G_1,G_2)$ takes values from a finite set.

Now it is easy to see that
\begin{equation}
\label{Eq:P(delta<=0)}
    \prob (\delta_\pi (G_1, G_2)\leq 0) = \sum_{d \leq 0}[z^d] \mathcal{A}(\bm{p},\bm{q},z).
\end{equation}

\textbf{\emph{Cycle decomposition.}}
% \sout{\ning{The goal of introducing generating function $\mathcal{A}(\bm{x},\bm{y},z)$ is to obtain an upper bound for $\prob(\delta_\pi(G_1,G_2))\leq 0)$.
% }%
% However, due to the complicated form of $\mathcal{A}(\bm{x},\bm{y},z)$, we can not directly derive a closed-form expression for $\prob(\delta_\pi (G_1, G_2)\leq 0)$ through equation \eqref{Eq:P(delta<=0)}.
% To simplify $\mathcal{A}(\bm{x},\bm{y},z)$, we apply the method called \emph{cycle decomposition} \ning{and stated the simplified generating function in Fact~\ref{Fact:ogf cycle decomposition}}.}
We will use the cycle decomposition of permutations to simply the form of the generating function $\mathcal{A}(\bm{x},\bm{y},z)$.

Each permutation $\pi$ induces a permutation on the vertex-pair set. We denote this induced permutation as $\pi^\mathcal{E}$, where $\pi^\mathcal{E}: \mathcal{E} \to \mathcal{E}$ and $\pi^\mathcal{E}((u,v)) = (\pi(u),\pi(v))$ for $u,v \in \mathcal{V}$.
A \emph{cycle} of the induced permutation $\pi^{\mathcal{E}}$ is a list of vertex pairs such that each vertex pair is mapped to the vertex pair next to it in the list (with the last  mapped to the first one). 
The cycles naturally partition the set of vertex pairs, $\mathcal{E}$, into disjoint subsets where each subset consists of the vertex pairs from a cycle. We refer to each of these subsets as an \emph{orbit}.
For the example given in Figure~\ref{fig:model}, the induced permutation on $\mathcal{E}$ can divide it into $4$ orbits of size $1$ ($1$-orbit): $\{(2,3)\}$, $\{(1,a)\}$, $ \{(1,b) \}$, $\{ (1,c)\}$, and $4$ orbits of length $2$ ($2$-orbit): $\{(1,2),(1,3)\}$, $\{ (2,a),(3,a) \}$, $\{(2,b),(3,b)\}$, $\{ (2,c),(3,c)\}$. 

We write this partition of $\mathcal{E}$ based on the cycle decomposition as $\mathcal{E} = \cup_{k\geq1} \orbit _k$, where $\orbit_k$ denotes the $k$th orbit. Note that each cycle consists of either only user-user vertex pairs or only user-attribute vertex pairs. If a single orbit $\orbit_k$ contains only user-user vertex pairs, we define its generating function on formal variables $z$ and $\bm{x}$ as
\begin{align*}
    \mathcal{A}_{\orbit_k}(\bm{x},z) = \sum_{g\in \{0,1\}^{\orbit_k}}\sum_{h \in \{0,1\}^{\orbit_k}} z^{\delta_{\pi}(g,h)}\bm{x}^{\bm{\muu}(g,h)}.
    % \bm{x}^{\bm{\mua}(g,h)}.
\end{align*}
If $\orbit_k$ contains only user-attribute vertex pairs, we define its generating function on formal variables $z$ and $\bm{y}$ as
\begin{align*}
    \mathcal{A}_{\orbit_k}(\bm{y},z) = \sum_{g\in \{0,1\}^{\orbit_k}}\sum_{h \in \{0,1\}^{\orbit_k}} z^{\delta_{\pi}(g,h)}\bm{y}^{\bm{\mua} (g,h)}.
    % \bm{y}^{\bm{\muu}(g,h)}.
\end{align*}
Here, we extend the previous definitions of $\delta_\pi$, $\bm{\muu}$ and $\bm{\mua}$ on attributed graphs to any set of vertex pairs. 
Let $\mathcal{E'}$ be an arbitrary set of vertex pairs. Then
we define $\delta_\pi$ for any $g,h \in \{0,1\}^\mathcal{E'}$ as
\begin{align}
    \delta_\pi(g,h) = & w_1 \mspace{-20mu} \sum_{e \in \mathcal{E'} \cap  \mathcal{E}_\mathrm{u}} \left(\mathbb{1}\{g(e) \neq h(\pi^{\mathcal{E}}(e))\} -\mathbb{1}\{g(e) \neq h(e)\} \right)\nonumber\\
    \label{Eq:extend delta_pi}
    & \mspace{-50mu} +w_2\sum_{{e \in \mathcal{E'} \cap \mathcal{E}_\mathrm{a}}} \left(\mathbb{1}\{g(e) \neq h(\pi^{\mathcal{E}}(e))\} -\mathbb{1}\{g(e) \neq h(e)\} \right).
\end{align}
For $g, h \in \{0,1\}^{\mathcal{E'}}$, we keep $\bm{\muu}(g,h)$ and $\bm{\mua}(g,h)$ as $2\times2$ matrices as follows:
\begin{align*}
    \bm{\muu}(g,h) = 
    \begin{pmatrix}
    \muu_{11}, \muu_{10}\\
    \muu_{01}, \muu_{00}
    \end{pmatrix}
    \text{ \;\ and \;\ }
    \bm{\mua}(g,h) = 
    \begin{pmatrix}
    \mua_{11}, \mua_{10}\\
    \mua_{01}, \mua_{00}
    \end{pmatrix},
\end{align*}
where
\begin{align}
    \label{Eq:extend muu}
    \muu_{ij} = \muu_{ij}(g,h) \triangleq\sum_{{e \in \mathcal{E'} \cap \mathcal{E}_\mathrm{u}}} \mathbb{1}\{ g(e)=i, h(e)=j\},\\
    \label{Eq:extend mua}
    \mua_{ij} = \mua_{ij}(g,h) \triangleq \sum_{{e \in \mathcal{E'} \cap \mathcal{E}_\mathrm{a}}} \mathbb{1}\{ g(e)=i, h(e)=j\}.
\end{align}
We remind the reader that by setting the set of vertex pairs $\mathcal{E'}$ to be $\mathcal{E}$ these extended definitions on $\delta_\pi$, $\bm{\muu}$ and $\bm{\mua}$ agree with the previous definition where $g, h$ are attributed graphs.

Now, we consider the generating functions for two orbits $\orbit_k$ and $\orbit_{k'}$. If the size of $\orbit_k$ equals to the size of $\orbit_{k'}$ and both orbits consist of user-user vertex pairs, then we claim that $\mathcal{A}_{\orbit_k}(\bm{x},z) = \mathcal{A}_{\orbit_{k'}}(\bm{x},z)$. This is because to obtain $\mathcal{A}_{\orbit_k}(\bm{x},z)$, we sum over all realizations $g,h \in \{0,1\}^{\orbit_k}$, which is equivalent to summing over $g,h \in \{0,1\}^{\orbit_{k'}}$. Similarly, if the size of $\orbit_k$ equals to the size of $\orbit_{k'}$ and both orbits consist of user-attribute vertex pairs, we have $\mathcal{A}_{\orbit_k}(\bm{y},z) = \mathcal{A}_{\orbit_{k'}}(\bm{y},z)$.
To make the notation compact, we define a generating function $\mathcal{A}_l(\bm{x},z)$ for size $l$ user-user orbits and a generating function $\mathcal{A}_l(\bm{y},z)$ for size $l$ user-attribute orbits. Let $\mathcal{E}_l^\mathrm{u}$ denote a general user-user orbit of size $l$ and $\mathcal{E}_l^\mathrm{a}$ denote a general user-attribute orbit of size $l$. Then
\begin{align}
\label{Eq:l-cycle OGFu}
    \mathcal{A}_l(\bm{x},z) \triangleq \sum_{g\in \{0,1\}^{\mathcal{E}_l^{\mathrm{u}}}}\sum_{h \in \{0,1\}^{\mathcal{E}_l^{\mathrm{u}}}} z^{\delta_{\pi}(g,h)}\bm{x}^{\bm{\muu} (g,h)},\\
    % \bm{x}^{\bm{\mua} (g,h)},\\
    \label{Eq:l-cycle OGFa}
    \mathcal{A}_l(\bm{y},z) \triangleq \sum_{g\in \{0,1\}^{\mathcal{E}_l^{\mathrm{a}}}}\sum_{h \in \{0,1\}^{\mathcal{E}_l^{\mathrm{a}}}} z^{\delta_{\pi}(g,h)}\bm{y}^{\bm{\mua} (g,h)}.
    % \bm{y}^{\bm{\mua} (g,h)}.
\end{align}

\textbf{\emph{Properties of generating functions.}}
\begin{restatable}{fact}{ogfcycle}\label{Fact:ogf cycle decomposition}
The generating function $\mathcal{A}(\bm{x},\bm{y},z)$ of permutation $\pi$ can be decomposed into
\begin{align*}
    \mathcal{A}(\bm{x},\bm{y},z)= & \prod_{l\geq1} \mathcal{A}_{l}(\bm{x},z)^{t_l^{\mathrm{u}}} \mathcal{A}_{l} (\bm{y},z)^{t_l ^{\mathrm{a}}},
\end{align*}
where $t_l^\mathrm{u}$ is the number of user-user orbits of  size $l$, $t_l^\mathrm{a}$ is the number of user-attribute orbits of size $l$. 
\end{restatable}
\begin{fact}
\label{Fact:l and 2 cycle}
Let $\bm{x} \in \R ^{2 \times 2}$ and $z \neq 0$. Then for all $l \geq 2$, we have $\mathcal{A}_{l}(\bm{x},z) \leq \mathcal{A}_{2}(\bm{x},z)^{\frac{l}{2}}$ and  $\mathcal{A}_{l}(\bm{x},z) \leq \mathcal{A}_{2}(\bm{x},z)^{\frac{l}{2}}$.
\end{fact}
We refer the readers to Appendix~\ref{Appendix:fact1} for the proof of Fact~\ref{Fact:ogf cycle decomposition}, and Theorem~4 in \cite{Cul-Kiy-exact2017} for the proof of Fact~\ref{Fact:l and 2 cycle}.
Combining these two facts, we get 
\begin{align}
    \mathcal{A}(\bm{x},\bm{y},z) \leq  &\mathcal{A}_{1}(\bm{x},z)^{t_1^\mathrm{u}}
    \mathcal{A}_{1}(\bm{y},z)^{t_1^\mathrm{a}}\nonumber\\
    \label{Eq:ogf bound}
    &\mathcal{A}_{2}(\bm{x},z)^{\frac{\vpin -t_1^\mathrm{u}}{2}}
    \mathcal{A}_{2}(\bm{x},z)^{\frac{nm-t_1^\mathrm{a}}{2}}.
\end{align}
Here, in \eqref{Eq:ogf bound}, we use $\vpin$ to denote the total number of user-user pairs and $\vpin = \sum_{l \geq 1} t_l^\mathrm{u} l = \binom{n}{2}$.
We have the closed-form expressions for $\mathcal{A}_1$ and $\mathcal{A}_2$ following from their definition in \eqref{Eq:l-cycle OGFu} and \eqref{Eq:l-cycle OGFa}
\begin{align}
    \label{Eq:A1(x)}
    \mathcal{A}_{1}(\bm{x},z)&=x_{00}+x_{10}+x_{01}+x_{11},\\
    \label{Eq:A1(y)}
    \mathcal{A}_{1}(\bm{y},z)&=y_{00}+y_{10}+y_{01}+y_{11},\\
    \label{Eq:A2(x)}
    \mathcal{A}_{2}(\bm{x},z)&=(x_{00}+x_{10}+x_{01}+x_{11})^2\nonumber \\ 
    &+2x_{00}x_{11}(z^{2w_1}-1)+2x_{10}x_{01}(z^{-2w_1}-1), \\
    \label{Eq:A2(y)}
    \mathcal{A}_{2}(\bm{y},z)&=(y_{00}+y_{10}+y_{01}+y_{11})^2\nonumber \\
    &+2y_{00}y_{11}(z^{2w_2}-1)+2y_{10}y_{01}(z^{-2w_2}-1).
\end{align}
Moreover, we have Fact~\ref{Fact:ogf coefficient} which gives explicit upper bounds on the coefficients of a generating function
\begin{fact}
\label{Fact:ogf coefficient} 
For a discrete random variable $X$ defined over a finite set $\mathcal{X}$, let
\begin{equation}
\label{Eq:prob-gen-func}
\Phi(z)  \triangleq \E[z^X] = \sum_{i \in \mathcal{X}}\prob(X=i)z^i
\end{equation}
be the probability generating function of $X$.
Then, for any $j \in \mathcal{X}$ and $z > 0$,
\begin{equation}
\label{eq:ogf coefficient}
     [z^j]\Phi(z) \leq z^{-j} \Phi(z).
\end{equation}
For any $j \in \mathcal{X}$ and $z \in (0,1]$, 
\begin{equation}
\label{eq:ogf sum-coefficient}
    \sum_{\substack{i \leq j\\ i\in \mathcal{X}}} [z^i]\Phi(z) \leq  z^{-j}\Phi(z).
\end{equation}
For any $j \in \mathcal{X}$ and $z \geq 1$, 
\begin{equation}
\label{eq:ogf sum-coefficient2}
    \sum_{\substack{i \geq j\\ i\in \mathcal{X}}} [z^i]\Phi(z) \leq  z^{-j}\Phi(z).
\end{equation}

\end{fact}
\begin{proof}[Proof of Fact~\ref{Fact:ogf coefficient}]
We write $p_i \triangleq \prob(X=i)$ in this proof. For any $j \in \mathcal{X}$ and $z>0$, we have 
\begin{align*}
    z^{-j}\Phi(z) - [z^j]\Phi(z) &= \sum_{i\in\mathcal{X}}p_i z^{i-j}-p_j
    = \sum_{\substack{i \neq j\\i\in\mathcal{X}}}p_i z^{i-j}
    \ge 0,
\end{align*}
which establishes~\eqref{eq:ogf coefficient}.

For any $j \in \mathcal{X}$ and $z \in (0,1)$, we have $\sum_{i \leq j}p_i \leq \sum_{i \leq j}p_i z^{i-j}$. Therefore, we have
\begin{align*}
    \sum_{\substack{i \leq j\\  i\in \mathcal{X}} } [z^i] \Phi (z) = \sum_{\substack{i \leq j\\ i \in \mathcal{X}}}p_i \leq \sum_{\substack{i \leq j\\ i \in \mathcal{X} }}p_i z^{i-j} \leq \sum_{i \in \mathcal{X}}p_i z^{i-j} = z^{-j} \Phi (z),
\end{align*}
which establishes \eqref{eq:ogf sum-coefficient}. 

For any $z >1$ and $j \in \mathcal{X}$, we have $\sum_{i \geq j}p_i \leq \sum_{i \geq j}p_i z^{i-j}$. Therefore, we have
\begin{align*}
    \sum_{\substack{i \geq j\\  i\in \mathcal{X}} } [z^i]\Phi (z) = \sum_{\substack{i \geq j\\ i \in \mathcal{X} }}p_i \leq \sum_{\substack{i \geq j\\ i \in \mathcal{X}}}p_i z^{i-j} \leq \sum_{i \in \mathcal{X}}p_i z^{i-j} = z^{-j} \Phi(z),
\end{align*}
which establishes \eqref{eq:ogf sum-coefficient2}.
\end{proof}
% \subsection{Proof of Lemma~\ref{Lemma:P(error of pi)}}\label{Proof:Lemma4}
% \begin{proof}%[Proof of Lemma~\ref{Lemma:P(error of pi)}]
With the three facts of generating functions stated above, we are now ready to finish the proof of Lemma~\ref{Lemma:P(error of pi)}.
For any $\pi \in \SnTn$ and any $z_1\in (0,1)$, we have
\begin{align}
    & \prob\left(\delta_\pi (G_1, G_2)\leq 0\right) \nonumber\\
    &= \sum_{d\leq 0}[z^d] \mathcal{A}(\bm{p},\bm{q} ,z) \nonumber\\
    \label{eqLem4_0}
    &\leq \mathcal{A}(\bm{p},\bm{q} ,z_1)\\
    \label{eqLem4_-1}
    &\leq \mathcal{A}_{1}(\bm{p},z)^{t_1^\mathrm{u}}
    \mathcal{A}_{1}(\bm{q},z)^{t_1^\mathrm{a}}\nonumber\\
    & \;\ \mathcal{A}_{2}(\bm{p},z)^{\frac{\vpin -t_1^\mathrm{u}}{2}}
    \mathcal{A}_{2}(\bm{q},z)^{\frac{nm-t_1^\mathrm{a}}{2}}\\
    \label{eqLem4_1}
    &\leq \mathcal{A}_2(\bm{p},z)^{\frac{\vpin - t_1 ^{\mathrm{u}}}{2}} \mathcal{A}_2(\bm{q},z)^{\frac{nm- t_1 ^{\mathrm{a}}}{2}}.
\end{align}
In \eqref{eqLem4_0}, we set $z \in (0,1)$, and this upper bound follows from Fact~\ref{Fact:ogf coefficient}. \eqref{eqLem4_-1} follows from the decomposition on $\mathcal{A}(\bm{p},\bm{q} ,z)$ stated in Fact~\ref{Fact:ogf cycle decomposition}. Equation \eqref{eqLem4_1} follows since $\mathcal{A}_1(\bm{p},z) =\mathcal{A}_1(\bm{q},z) =1$ according to their expression in \eqref{Eq:A1(x)} and \eqref{Eq:A1(y)}. To obtain a tight bound, we then search for $z \in (0,1)$ that achieves the minimum of \eqref{eqLem4_1}. Following the definition of $\mathcal{A}_2(\bm{p},z)$ in \eqref{Eq:A2(x)} and using the inequality $a/x+bx \geq 2\sqrt{ab}$, we have
\begin{align}
    \mathcal{A}_2(\bm{p},z) &= 1+ 2p_{00}p_{11}( z^{2w_1}-1)+2p_{10}p_{01}( z^{-2w_1}-1)\nonumber\\
    &\geq 1-2p_{00}p_{11}-2p_{10}p_{01} + 4\sqrt{p_{00}p_{11}p_{10}p_{01}}\nonumber\\
    \label{Eq:def-psiu}
    & = 1-2(\sqrt{p_{00}p_{11}}-\sqrt{p_{10}p_{01}})^2 \triangleq 1-2\psiu.
\end{align}
Here the equality holds if and only if $z^{2w_1} = \sqrt{\frac{p_{10}p_{01}}{p_{00}p_{11}}}$. Recall that $w_1 = \log\left(\frac{p_{11}p_{00}}{p_{10}p_{01}}\right)$. Therefore, $\mathcal{A}_1(\bm{p},z)$ achieves the minimum when  {$z = e^{-1/4}$}. Similarly, we have 
\begin{align}
    \mathcal{A}_2(\bm{q},z) &= 1+ 2q_{00}q_{11}( {z^{2w_2}}-1)+2q_{10}q_{01}( {z^{-2w_2}}-1)\nonumber\\
    &\geq 1-2q_{00}q_{11}-2q_{10}q_{01} + 4\sqrt{q_{00}q_{11}q_{10}q_{01}}\nonumber\\
    \label{Eq:def-psia}
    & = 1-2(\sqrt{q_{00}q_{11}}-\sqrt{q_{10}q_{01}})^2 \triangleq 1-2\psia.
\end{align}
Here the equality holds if and only if  $ {z^{2w_2}} = \sqrt{\frac{q_{10}q_{01}}{q_{00}q_{11}}}$. With $w_2 = \log\left(\frac{q_{11}q_{00}}{q_{10}q_{01}}\right)$, we have that  $\mathcal{A}_2(\bm{q},z)$ achieves the minimum when  {$z = e^{-1/4}$}. Therefore,  {$z =e^{-1/4}$} minimizes \eqref{eqLem4_1} and we have
\begin{align}
    &\prob\left(\delta_\pi (G_1, G_2)\leq 0\right) \nonumber \\
    &\leq (1-2\psiu)^{\frac{\vpin - t_1^{\mathrm{u}}}{2}}
    (1-2\psia)^{\frac{mn- t_1 ^{\mathrm{a}}}{2}}\nonumber\\
    \label{eqLem4_2}
    &\leq (1-2\psi_{\mathrm{u}})^{\frac{\Tilde{n}(2n-\Tilde{n}- 2)}{4}}
    (1-2\psi_{\mathrm{a}})^{\frac{\Tilde{n}m}{2}}\\
    \label{eqLem4_3}
    &\leq (1-2\psi_{\mathrm{u}})^{\frac{\Tilde{n}(n-2)}{4}}(1-2\psi_{\mathrm{a}})^{\frac{\Tilde{n}m}{2}}.
\end{align}
In \eqref{eqLem4_2}, we use the following relations between the number of fixed vertex pairs $t_1^\mathrm{u}$, $t_1^\mathrm{a}$ and number of fixed vertices $\Tilde{n}$
\begin{align}
    &\binom{n-\Tilde{n}}{2} \leq t_1^\mathrm{u} \leq \binom{n-\Tilde{n}}{2}+\frac{\Tilde{n}}{2},\\
    & t_1^\mathrm{a} = (n-\Tilde{n})m \nonumber.
\end{align}
In the given upper bound of $t_1^\mathrm{u}$, $\binom{n-\Tilde{n}}{2}$ corresponds to the number of user-user vertex pairs whose two vertices are both fixed under $\pi$, and $\frac{\Tilde{n}}{2}$ is the upper bound of user-user vertex pairs whose two vertices are swapped under $\pi$. 
In $\eqref{eqLem4_3}$, we use the fact that $\Tilde{n} \leq n$. 
\end{proof}

\section{Proof of achievability in the sparse regime}

% \emph{\textbf{Proof Outline.}}
In this section, we prove Theorem~\ref{Thm:achievability-sparse}, which characterizes the achievable region when the user-user connection is  \emph{sparse} in the sense that $p_{11} = O(\frac{\log n}{n})$. We use $\Density$ to denote the number of user-user edges in the intersection graph and it follows a binomial distribution $\Bin(\vpin, p_{11})$. 
%%ToDo
% From the last section, $p_{11} = \omega(\frac{\log n}{n})$ is proved to be achievable (Corollary~\ref{}). 
In the sparse regime where $p_{11} = O(\frac{\log n}{n})$, the achievability proof here is different from what we did in  Section~\ref{sec:proof-achievability}. The reason for applying a different proof technique is that, in this sparse regime, the union bound we applied in Section~\ref{sec:proof-achievability} on $\prob\left(\exists \pi \in \Sn \setminus\{\pi_\mathrm{id}\}, \delta_{\pi}(G_1,G_2) \leq 0\right)$ becomes very loose. To elaborate on this point, notice that the error of union bound comes from counting the intersection events multiple times. Therefore, if the probability of such intersection events becomes larger, then the union bound will be looser. In our problem, our event space contains sets of possible realizations on $(G_1, G_2)$ and an example of the aforementioned intersection events is  $\{\Density=0\}$ which lays in the intersection of $\{\delta_{\pi}(G_1,G_2) \leq 0\}$ for all $\pi \in \Sn$.  Moreover, other events where $\Density$ is small are also in the intersection of $\{\delta_{\pi}(G_1,G_2) \leq 0\}$ for some $\pi \in \Sn$ and the number of such permutations (equivalently the times of repenting when apply union bound) increases as $R$ gets smaller. As a result,  if $p_{11}$ becomes relatively small, then the probability that $\Density$ is small will be large and thus union bound will be loose.

To overcome the problem of the loose union bound in the sparse regime, we apply a \emph{truncated union bound}. We first expand the probability we want to bound as follows
\begin{align*}
    &\prob\left(\exists \pi \in \Sn \setminus\{\pi_\mathrm{id}\}, \delta_{\pi}(G_1,G_2) \leq 0\right) \\
    &\mspace{-5mu}=\mspace{-5mu} \sum_{r\geq 0} \prob\left(\exists \pi \in \Sn \setminus\{\pi_\mathrm{id}\}, \delta_{\pi}(G_1,G_2) \leq 0 | \Density=\density \right) \prob(\Density=\density).
\end{align*}
We then apply the union bound on the conditional probability $\prob\left(\exists \pi \in \Sn \setminus\{\pi_\mathrm{id}\}, \delta_{\pi}(G_1,G_2) \leq 0 | \Density=\density \right).$ As we discussed before, the error of applying union bound directly should be a function on $\density$. Therefore, for some small $\density$, the union bound on $\prob\left(\exists \pi \in \Sn \setminus\{\pi_\mathrm{id}\}, \delta_{\pi}(G_1,G_2) \leq 0 | \Density=\density \right)$ is very loose while for the other $\density$, the union bound is relatively tight.
Therefore, we truncate the union bound on the conditional probability by taking the minimum with 1, which is an upper bound for any probability
\begin{align*}
    &\prob\left(\exists \pi \in \Sn \setminus\{\pi_\mathrm{id}\}, \delta_{\pi}(G_1,G_2) \leq 0 | \Density=\density \right)\\
    &= \min \{1, \sum_{\pi \in \Sn \setminus\{\pi_{\mathrm{id}}\}}\prob\left(\delta_{\pi}(G_1,G_2) \leq 0 | \Density=\density \right)\}.
\end{align*}
Through such truncating, we avoid adopting the union bound when it is too loose and obtain a tighter bound. For example, given that $\Density=0$,  we have $\prob\left(\delta_{\pi}(G_1,G_2) \leq 0 | \Density=0 \right) =1$ for all $\pi \in \Sn$. Thus, by using the truncated union bound, we obtain $1$ as a the upper bound instead of $(n!-1)$.
Overall, the key idea of our proof is first derive $\prob\left(\delta_{\pi}(G_1,G_2) \leq 0 | \Density=\density \right)$ as a function of $\density$ and then apply the truncated union bound according to how large this conditional probability  is.
This idea is inspired by \cite{Cul-Kiy-exact2017} and is extended to the attributed \er\ pair model. 
We restate the theorem to prove as follows.

% \begin{replemma}{Lemma:achievability-sparse}
% Consider the attributed \er\ pair $\mathcal{G}(n,\bm{p};m,\bm{q})$.
% If
% \begin{align}
%     % \ning{m \psia} & \ning{= \Theta(\log n)},\label{eq:lem30}\\
%     p_{11} &=\ning{O}\left(\tfrac{\log{n}}{n}\right),\label{eq:lem31}\\ 
%     p_{10}+p_{01} &= O\left(\tfrac{1}{\log{n}}\right),\label{eq:lem32}\\
%     \frac{p_{10}p_{01}}{p_{11}p_{00}} &= O\left(\tfrac{1}{\ning{(\log{n})^3}}\right),\label{eq:lem33}\\
%     np_{11}+m \psia -\log{n} &= \omega(1),\label{eq:lem34}
% \end{align}
% then the MAP estimator achieves exact alignment w.h.p.
% \end{replemma}
\LemmaSparse*

\begin{proof}[Proof of Theorem~\ref{Thm:achievability-sparse}]
% \ning{We first restrict the regime to discuss using Lemma~\ref{Lemma:achievability-general}.  
% If $m \psia = \omega (n \log n)$, then exact alignment is achievable according to \eqref{eq:achievability}, which agrees with \eqref{eq:lem34}. 
% Therefore we only need to consider $m \psia = O(n \log n)$ in the later part of the proof.}
We discuss two regimes $p_{11}= O(\frac{1}{n})$ and  $\omega(\frac{1}{n}) \leq p_{11} \leq \Theta(\frac{\log{n}}{n})$ .

When $p_{11} = O(\frac{1}{n})$, we have $n\psiu \leq np_{11} = O(1)$. Thus, the sufficient condition~\eqref{Eq:general-feasible-region} for exact alignment in Theorem~\ref{Thm:achievability-general} 
\[
\tfrac{n\psiu}{2} + m\psia -\log{n} = \omega(1)
\]
is satisfied when condition~\eqref{eq:lem34} 
\[
np_{11}+m \psia -\log{n} = \omega(1)
\]
is satisfied. By Theorem~\ref{Thm:achievability-general} , exact alignment is achievable w.h.p. 

Now suppose $\omega(\frac{1}{n}) \le p_{11} \le \Theta(\frac{\log n}{n})$. 
Note that the number of edges in the intersection graph of $G_1$ and $G_2$ has the following equivalent representation
\[
\Density = \muu_{11}(G_1,G_2)=\sum_{e \in \mathcal{E}_{\mathrm{u}}} \mathbb{1}\{G_1(e) =1, G_2(e) =1\}.
\]
Then, $\Density \sim \Bin(\vpin,p_{11})$ and $\E[\Density] = \vpin p_{11} = \binom{n}{2} p_{11} = \omega(n)$.
By the Chebyshev's inequality, for any constant $\epsilon > 0$,
\begin{align*}
    \prob(|\Density -\E[\Density]|\geq \epsilon \E[\Density]) \leq \mspace{-5mu} \frac{\var(\Density)}{\epsilon^2 \E[\Density]^2}\mspace{-5mu}=\mspace{-5mu}\frac{1-p_{11}}{\epsilon^2}\frac{1}{\E[\Density]}= o\left(\frac{1}{n}\right).
\end{align*}
In the following, we upper bound the probability of error by discussing two cases: when $\Density\leq (1+\epsilon) \E[\Density]$ and when $\Density>(1+\epsilon) \E[\Density]$. We have
% \sout{Therefore, using law of total probability, we have}
\begin{align}
    & \prob\left(\exists \pi \in \Sn  \setminus\{\pi_\mathrm{id}\}, \delta_{\pi}(G_1,G_2) \leq 0\right)\nonumber\\
   &= \sum_{\density=0}^{\vpin} \prob \left(\exists \pi \in \Sn \mspace{-5mu} \setminus \mspace{-5mu} \{\pi_\mathrm{id}\}, \delta_{\pi}(G_1,G_2) \leq 0 | \Density = \density \right) \prob(\Density =\density) \nonumber\\
   &\le \mspace{-25mu}\sum_{\substack{ {\density}\leq (1+\epsilon) \E[\Density] }}\mspace{-30mu}\prob \left(\exists \pi \in \Sn \mspace{-5mu} \setminus \mspace{-5mu} \{\pi_\mathrm{id}\}, \delta_{\pi}(G_1,G_2) \leq 0 | \Density =\density \right) \prob(\Density =\density)
   \nonumber\\
   &\qquad \quad + \prob\left(|\Density -\E[\Density]| > \epsilon \E[\Density]\right)\nonumber\\
   & = \mspace{-25mu}\sum_{{\density}\leq (1+\epsilon) \E[\Density]}\mspace{-30mu}\prob \left(\exists \pi \in \Sn \mspace{-5mu} \setminus \mspace{-5mu} \{\pi_\mathrm{id}\}, \delta_{\pi}(G_1,G_2) \leq 0 | \Density =\density \right) \prob(\Density =\density)
   \nonumber\\
   \label{eqLem3_1}
   &\qquad \quad + o(1)\\
   \label{eqLem3_2}
   & \stackrel{\text{Lemma}~\ref{Lemma:approximated union bound}}{\le} \sum_{{\density}\leq (1+\epsilon) \E[\Density]} 3n^2   z_6^2 \prob(\Density =\density) +o(1)\\
   & = \sum_{\density \leq (1+\epsilon) \E[\Density]} 3n^2   z_6^2 {\vpin \choose \density} p_{11}^{\density} (1-p_{11})^{\vpin - \density} +o(1) \nonumber\\
   &= 3n^2 (1-2\psi_{\mathrm{a}})^m  \sum_{\density=0}^{\vpin} {\vpin \choose \density} p_{11}^{\density} e^{-\frac{4\density}{n}} (1-p_{11})^{\vpin - \density}  +o(1) \nonumber \\
   \label{eqLem3_3}
   &= 3n^2 (1-2\psi_{\mathrm{a}})^m \left(p_{11}e^{-\frac{4}{n}}+1 - p_{11}\right)^{\vpin} + o(1)\\
   \label{eqLem3_4}
   &\leq  3n^2 (1-2\psi_{\mathrm{a}})^m \left(1-\tfrac{4}{n}p_{11}\right)^{\vpin}+o(1).
\end{align}
Here \eqref{eqLem3_1} follows from the Chebyshev's inequality above. In~\eqref{eqLem3_2}, $z_6 = \exp \{ -\tfrac{2r}{n} + \tfrac{m}{2} \log{(1-2\psia)} +O(1)\}$. This step will be justified by Lemma~\ref{Lemma:approximated union bound}, which is the major technical step in establishing the error bound. 
To apply Lemma~\ref{Lemma:approximated union bound}, we need the conditions~\eqref{eq:lem31}~\eqref{eq:lem32}~\eqref{eq:lem33} and $\density = O(\E[\Density]) = O(n \log n)$  to hold and we will explain the reason in the proof of Lemma~\ref{Lemma:approximated union bound}. Equation \eqref{eqLem3_3} follows from the binomial formula and \eqref{eqLem3_4} follows from the inequality $e^x-1 \leq x$. Taking the negative logarithm of the first term in \eqref{eqLem3_4}, we have
\begin{align}
    &-\log{\left(3n^2 (1-2\psi_{\mathrm{a}})^m \left(1-\tfrac{4}{n}p_{11}\right)^{\vpin} \right)} \nonumber\\
    &= -2\log{n} -m\log{(1-2\psia)} - \vpin \log{\left(1-\tfrac{4p_{11}}{n}\right)}+O(1)\nonumber\\
    \label{eqLem3_5}
    &\geq -2\log{n} +2m\psia +\vpin \frac{4p_{11}}{n}+O(1)\\
    \label{eqLem3_6}
    &= -2\log{n} +2m\psia +2np_{11}+O(1)\\
    \label{eqLem3_7}
    &= \omega(1).
\end{align}
Here, we have \eqref{eqLem3_5} follows from the inequality $\log{(1+x)} \leq x $ for $x>-1$. 
We get equation \eqref{eqLem3_6} by plugging in $\vpin = {n \choose 2}$. Equation \eqref{eqLem3_7} follows from the assumption \eqref{eq:lem34} in Theorem~\ref{Thm:achievability-sparse}. Therefore, we have \eqref{eqLem3_4} converges to 0 and so does the error probability. 
\end{proof}

\begin{lemma}
\label{Lemma:approximated union bound} 
Let $(G_1,G_2) \sim \mathcal{G}(n,\bm{p};m,\bm{q})$ and $\Density=\sum_{e \in \mathcal{E}_{\mathrm{u}}} \mathbb{1}\{G_1(e) =1, G_2(e) =1\}$.
If $\bm{p}$ satisfies constraints \eqref{eq:lem31} \eqref{eq:lem32} \eqref{eq:lem33}, and $\density = O(n \log n)$, then
\begin{equation*}
    \prob(\exists \pi \in \Sn \setminus \{\pi_\mathrm{id}\}, \delta_\pi(G_1,G_2)  \leq 0\mid \Density = \density) \le 3n^2   z_6^2,
\end{equation*}
where $z_6 = \exp \{ -\frac{2r}{n} +\tfrac{m}{2} \log{(1-2\psia)} +O(1)\}$.
\end{lemma}
\begin{proof} 

We will establish the above upper bound in three steps.
{We denote the set of vertex pairs that are \emph{moving} under permutation $\pi^{\mathcal{E}}$ as
$\mathcal{E}_\mathrm{m}= \{e \in \mathcal{E}: \pi^\mathcal{E}(e) \neq e$\}. Let 
\[
\nonDensity = \sum_{e \in \mathcal{E}_\mathrm{m} \cap \mathcal{E}_\mathrm{u}}\mathbb{1}\{G_1(e) =1, G_2(e) =1\}
\]
represent the number of co-occurred user-user edges in $\mathcal{E}_\mathrm{m}$ of $G_1 \wedge G_2$.} In Step 1, we apply the method of generating functions to get an upper bound on $\prob (\delta_\pi(G_1,G_2)  \leq 0\mid \nonDensity = \nondensity)$.
The reason for conditioning on $\nonDensity$ first is that the corresponding generating function only involves cycles of length $l\geq 2$ and its upper bound is easier to derive compared with the probability conditioned on $\Density$.
In Step 2, we upper bound $\prob (\delta_\pi(G_1,G_2)  \leq 0\mid \Density = \density)$ using result from Step 1 and properties of the Hypergeometric distribution. 
In Step 3, we upper bound $\prob(\exists \pi \in \Sn \setminus \{\pi_\mathrm{id}\}, \delta_\pi(G_1,G_2)  \leq 0\mid \Density = \density)$ using the truncated union bound.

\vspace{.5em}
\textbf{\emph{Step 1.}} We prove that for any $\pi \in \SnTn$, $\nondensity = O(\frac{\Tilde{t ^{\mathrm{u}}}\log{n}}{n})$, and $z_3 = (1-2\psi_{\mathrm{a}})^\frac{m}{2}$, 
\begin{equation}
\label{Eq: P(delta|nonDensity)}
   \prob\left(\delta_{\pi}(G_1,G_2) \leq 0\mid \nonDensity=\nondensity \right) \leq z_3 ^{\Tilde{n}} z_4^{\nondensity} z_5^{\Tilde{n}}
\end{equation}
for some $z_4 = O(\frac{1}{\log{n}})$ and some $z_5=O(1)$.

\vspace{0.5em}
{For the induced subgraph pair on $\mathcal{E}_\mathrm{m} \times \mathcal{E}_\mathrm{m}$, define the generating function as}%
\begin{align}
\label{Eq:def ogf Tilde-A}
    \Tilde{\mathcal{A}} (\bm{x},\bm{y},z) 
    &= \sum_{g\in \{0,1\}^{\mathcal{E} _{\mathrm{m}}}}\sum_{h \in \{0,1\}^{\mathcal{E} _{\mathrm{m}}}} z^{\delta_{\pi}(g,h)}\bm{x}^{\bm{{\muu}}(g,h)}\bm{y}^{\bm{{\mua}}(g,h)}.
\end{align}
Recall for $g,h \in \{0,1\}^{\mathcal{E}_\mathrm{m}}$, the expression for the extended  $\delta_\pi(g,h)$, $\bm{{\muu}}(g,h)$ and $\bm{{\mua}}(g,h)$ in \eqref{Eq:extend delta_pi}, \eqref{Eq:extend muu} and \eqref{Eq:extend mua}. We have
\begin{align*}
    \delta_\pi(g,h) &= w_1 \mspace{-20mu} \sum_{e \in \mathcal{E}_\mathrm{m} \cap \mathcal{E}_\mathrm{u}} \mspace{-20mu} \left(\mathbb{1}\{g(e) \neq h(\pi^{\mathcal{E}}(e))\} -\mathbb{1}\{g(e) \neq h(e)\} \right)\nonumber\\
    & \mspace{-50mu} +w_2 \sum_{e \in \mathcal{E}_\mathrm{m} \cap \mathcal{E}_\mathrm{a}} \left(\mathbb{1}\{g(e) \neq h(\pi^{\mathcal{E}}(e))\} -\mathbb{1}\{g(e) \neq h(e)\} \right).
\end{align*}
For the $2 \times 2$ matrices $\bm{{\muu}}(g,h)$ and $\bm{{\mua}}(g,h)$, their entries $\muu_{i,j}$ and $\mua_{i,j}$ are 
\begin{align*}
    \muu_{ij} &= \muu_{ij}(g,h) = \sum_{e \in \mathcal{E}_\mathrm{m} \cap \mathcal{E}_\mathrm{u}} \mathbb{1}\{ g(e)=i, h(e)=j\},\\
    \mua_{ij} &= \mua_{ij}(g,h) = \sum_{e \in \mathcal{E}_\mathrm{m}\cap \mathcal{E}_\mathrm{a}} \mathbb{1}\{ g(e)=i, h(e)=j\}.
\end{align*}
Moreover, according to the decomposition of generating function in Fact~\ref{Fact:ogf cycle decomposition} and using the fact that $\mathcal{E}_\mathrm{m}$ only contains orbits of size larger than 1, we obtain
\begin{align*}
    \Tilde{\mathcal{A}} (\bm{x},\bm{y},z)  &= \prod_{l \geq2} \mathcal{A}_{l}(\bm{x},z)^{t_l ^{\mathrm{u}}} \prod_{l \geq 2} \mathcal{A}_{l}(\bm{y},z)^{t_l ^{\mathrm{a}}}.
\end{align*}
where $t_l ^{\mathrm{u}}$ is the number of user-user orbits of size $l$ and $t_l ^{\mathrm{a}}$ is the number of user-attribute orbits of size $l$.

Now, by setting 
\begin{align*}
\bm{x} =\bm{x_{11}} \triangleq \begin{pmatrix}
    p_{00}& p_{01}  \\
    p_{10}& x_{11}p_{11}
\end{pmatrix}
\end{align*}
and $\bm{y} = \bm{q}$, the generating function $\Tilde{\mathcal{A}}(\bm{x_{11}},\bm{q},z)$ contains only two formal variables $x_{11}$ and $z$. Recall the expression of $\Tilde{\mathcal{A}}$ in \eqref{Eq:def ogf Tilde-A}.
For each $g,h \in \{0,1\}^{\mathcal{E}_\mathrm{m}}$, the term in the summation of  $\Tilde{\mathcal{A}} (\bm{x_{11}},\bm{q},z)$ can be written as 
\begin{align*}
    &z^{\delta_\pi(g,h)}\bm{x_{11}}^{\bm{\muu}(g,h)} \bm{q}^{\bm{\mua}(g,h)}\\
    &= z^{\delta_\pi(g,h)} x_{11} ^{\muu_{11}(g,h)} \bm{p}^{\bm{\muu}(g,h)} \bm{q}^{\bm{\mua}(g,h)} \\
    &= \prob((G_1^{\mathcal{E_\mathrm{m}}},G_2^{\mathcal{E_\mathrm{m}}}) = (g,h)) \,\ z^{\delta_\pi(g,h)}  x_{11} ^{\muu_{11}(g,h)},
\end{align*}
where we use $G_1 ^{\mathcal{E_\mathrm{m}}}$ to denote the component of $G_1$ that only concerns the vertex pair set  $\mathcal{E}_\mathrm{m}$ and thus the support of $G_1 ^{\mathcal{E_\mathrm{m}}}$ is $\{0,1\}^{\mathcal{E_\mathrm{m}}}$.
The event $\{(G_1 ^{\mathcal{E_\mathrm{m}}}, G_2 ^{\mathcal{E_\mathrm{m}}}) = (g,h)\}$ is a collection of attributed graph pairs $(g_1,g_2)$ each of which have exactly the same edges in the vertex pair set $\mathcal{E}_\mathrm{m}$ as $(g,h)$.

Notice that the fixed vertex pairs $\mathcal{E} \setminus \mathcal{E}^\mathrm{m}$ do not have a influence on $\delta(G_1,G_2)$. The event $\{\nonDensity = \nondensity, \delta_\pi(G_1,G_2) =d \}$ is a collection of attributed graph pairs $(g_1,g_2)$ such that  $\muu_{11}(g_1^{\mathcal{E_\mathrm{m}}},g_2^{\mathcal{E_\mathrm{m}}}) = \nondensity$ and $\delta_\pi(g_1^{\mathcal{E_\mathrm{m}}},g_2^{\mathcal{E_\mathrm{m}}}) =d$. 
Then by summing over all possible $g,h \in \{0,1\}^{\mathcal{E}_\mathrm{m}}$ , we have
\[
\prob(\delta_\pi(G_1,G_2)=d,\nonDensity=\nondensity) =[z^d x_{11}^{\nondensity}] \Tilde{\mathcal{A}} (\bm{x_{11}},\bm{y},z).
\]
Thus, we can write
\begin{align}
    &\prob \left(\delta_\pi (G_1, G_2)\leq 0, \nonDensity =\nondensity\right) \nonumber\\
    &= \sum_{d\leq 0}[z^d x_{11}^{\nondensity}] \Tilde{\mathcal{A}}(\bm{x_{11}},\bm{q},z) \nonumber\\
    &= \sum_{d\leq 0}[z^d x_{11}^{\nondensity}] \prod_{l\geq2} \mathcal{A}_l(\bm{x_{11}},z )^{t_l ^{\mathrm{u}}} \mathcal{A}_l(\bm{q},z  )^{t_l ^{\mathrm{a}}}\nonumber\\
    \label{Lem7_0}
    & \leq (x_{11})^{-\nondensity} \sum_{d\leq 0}[z^d] \prod_{l\geq2} \mathcal{A}_l(\bm{x_{11}},z )^{t_l ^{\mathrm{u}}} \mathcal{A}_l(\bm{q},z)^{t_l ^{\mathrm{a}}}\\
    \label{Lem7_1}
    &\leq  (x_{11})^{-\nondensity}\prod_{l\geq2} \mathcal{A}_l(\bm{x_{11}},z )^{t_l^{\mathrm{u}}} \mathcal{A}_l(\bm{q},z  )^{t_l^{\mathrm{a}}}\\
    \label{Lem7_2}
    &\leq  (x_{11})^{-\nondensity} \mathcal{A}_2(\bm{x_{11}},z )^{\frac{\nontrivpin}{2}} \mathcal{A}_2(\bm{q},z  )^{\frac{m\Tilde{n}}{2}}.
\end{align}
In \eqref{Lem7_0}, we set $x_{11} >0$ and the inequality follows from \eqref{eq:ogf coefficient} in Fact~\ref{Fact:ogf coefficient}.
In \eqref{Lem7_1}, we set $z\in (0,1)$ and this inequality follows from \eqref{eq:ogf sum-coefficient} Fact~\ref{Fact:ogf coefficient}.
Inequality in~\eqref{Lem7_2} follows from Fact~\ref{Fact:l and 2 cycle}, where 
\[
\nontrivpin = \sum_{l\geq 2} \vpin_l l = |\mathcal{E}_\mathrm{m} \cap \mathcal{E}_\mathrm{u}|
\]
is the number of moving user-user pairs  and 
\[
\nontrivpout =\sum_{l\geq 2} \vpout_l l = |\mathcal{E}_\mathrm{m} \cap \mathcal{E}_\mathrm{a}| = \Tilde{n}m
\]
is the number of moving user-attribute pairs. 

Next, let us lower bound $\prob(\nonDensity = \nondensity)$. Note that $\nonDensity \sim \mathrm{Bin}(\nontrivpin, p_{11})$. We have
\begin{align}
    &\prob(\nonDensity = \nondensity) = {\nontrivpin \choose \nondensity} p_{11} ^{\nondensity} (1-p_{11})^{\nontrivpin -\nondensity}\nonumber\\
    &\geq \left(\frac{\nontrivpin p_{11}}{\nondensity(1-p_{11})}\right)^{\nondensity}(1-p_{11})^{\nontrivpin},\label{eq:rtilde}
\end{align}
where equation~\eqref{eq:rtilde} follows since $\binom{n}{k} \ge (n/k)^k$ for any nonnegative integers $k \le n$.

Now we combine the bounds in~\eqref{Lem7_2}~and~\eqref{eq:rtilde} to upper bound $\prob \left(\delta_{\pi}(G_1,G_2) \leq 0\mid \nonDensity=\nondensity \right)$. Define $p_{ij}'\triangleq\frac{p_{ij}}{1-p_{11}}$ for $i,j \in \{0,1\}$. We have 
\begin{align}
    % \label{Eq:3 terms}
    &\prob \left(\delta_{\pi}(G_1,G_2) \leq 0\mid \nonDensity=\nondensity \right) =\frac{\prob (\delta_\pi (G_1, G_2)\leq 0, \nonDensity=\Tilde{r})}{\prob(\nonDensity = \nondensity)} \nonumber\\
    \label{eq:step1}
    & \leq \mathcal{A}_2(\bm{q},z  )^{\frac{m\Tilde{n}}{2}}
    \left(\frac{\nondensity}{x_{11} p'_{11} \nontrivpin}\right)^{\nondensity} 
    \left(\frac{\mathcal{A}_2(\bm{x_{11}} , z )}{(1-p_{11})^2}\right)^{\nontrivpin /2}.
\end{align}
For the first term, similar to what we did in \eqref{Eq:def-psia}, we set $z = e^{-1/4}$, which satisfies the condition $z \in (0,1)$ in Fact~\ref{Fact:ogf coefficient}. Recall the expression of $\mathcal{A}_2(\bm{y},z)$ in \eqref{Eq:A2(y)}, we have
\begin{align}
    &\mathcal{A}_2(\bm{q},z  )^{\frac{m\Tilde{n}}{2}}\nonumber \\
    \label{Eq:Step1-add-1}
    &=\left(1+ 2q_{00}q_{11}(z^{2w_2}-1)+2q_{10}q_{01}(z^{-2w_2}+1)\right)^{\frac{m\Tilde{n}}{2}} \\
    \label{Eq:Step1-add-2}
    &= (1-2q_{00}q_{11}-2q_{10}q_{01} + 4\sqrt{q_{00}q_{11}q_{10}q_{01}})^{\frac{m\Tilde{n}}{2}}\\
    &= \left(1-2(\sqrt{q_{11}q_{00}}-\sqrt{q_{10}q_{01}})^2\right)^{\frac{m\Tilde{n}}{2}}\nonumber\\
    &= (1-2\psi_{\mathrm{a}})^{\frac{m\Tilde{n}}{2}}\nonumber\\
    \label{Eq:3terms}
    &\triangleq z_3 ^{\Tilde{n}}.
\end{align}
where \eqref{Eq:Step1-add-1} follows since $q_{00}+q_{01}+q_{10}+q_{11} = 1$ and \eqref{Eq:Step1-add-2} follows by plugging in $z = e^{-1/4}$ and $w_2 = \log{\left(\frac{q_{11}q_{00}}{q_{10}q_{01}}\right)}$.
For the second term in~\eqref{eq:step1}, we set
% \begin{align}
% \label{x_11}
%     x_{11} = \frac{1}{p_{11}'} \frac{\nondensity \log{n} + \E[\nonDensity]}{\nontrivpin} = \frac{\nondensity \log{n} + p_{11} \nontrivpin}{p_{11}'\nontrivpin},
% \end{align}
\begin{align}
\label{x_11}
    x_{11} =  \frac{\nondensity \log{n} + p_{11} \nontrivpin}{p_{11}'\nontrivpin},
\end{align}
which is positive. Then, we have
\begin{align}
\label{Eq:3terms-1}
    \left(\frac{\nondensity}{x_{11} p'_{11} \nontrivpin}\right)^{\nondensity} =\left(\frac{\nondensity}{\nondensity\log{n} + p_{11} \nontrivpin}\right)^{\nondensity} \leq \left( \frac{1}{\log{n}}\right)^{\nondensity}.
\end{align}
For the third term in~\eqref{eq:step1}, using equation \eqref{Eq:A2(x)} with $z = e^{-1/4}$, we have 
\begin{align*}
    &\frac{\mathcal{A}_2(\bm{x_{11}} , z )}{(1-p_{11})^2} \\
    &= \frac{(1-p_{11}+x_{11} p_{11})^2}{(1-p_{11})^2} +\frac{ 2x_{11} p_{11}p_{00} (\sqrt{\frac{p_{10}p_{01}}{p_{11}p_{00}}}-1)}{(1-p_{11})^2} \nonumber \\
    &\quad +\frac{2p_{10}p_{01} (\sqrt{\frac{p_{00}p_{11}}{p_{10}p_{01}}}-1)}{(1-p_{11})^2}\\
    &= (1 +p_{11}'x_{11})^2 -2x_{11}p_{11}' p_{00}' -2p_{10}' p_{01}'\\
    &\quad + 2(x_{11}+1) \sqrt{p_{11}' p_{00}' p_{10}' p_{01}'}\\
    &\leq  1 +(p_{11}'x_{11})^2 + 2p_{11}'x_{11}(p_{10}'+p_{01}') \\
    &\quad +2(x_{11}+1) \sqrt{p_{11}' p_{00}' p_{10}' p_{01}'},
\end{align*}
where the last inequality follows since $1- p_{00}' = p_{10}' +p_{01}'$ and $-2p_{10}' p_{01}' \le 0$. Taking logarithm of $\left(\frac{\mathcal{A}_2(\bm{x_{11}}, z )}{(1-p_{11})^2}\right)^{\nontrivpin/2}$, we get
\begin{align}
% \label{Lem 8_1}
     &\frac{\nontrivpin}{2} \log{\left(\frac{\mathcal{A}_2(\bm{x_{11}}, z )}{(1-p_{11})^2}\right)} \\
    &\leq \frac{1}{2} \nontrivpin (p_{11}'x_{11})^2 + \nontrivpin p_{11}'x_{11}(p_{10}'+p_{01}') \nonumber\\
    \label{2nd term in 3}
    &\quad +\nontrivpin(x_{11}+1) \sqrt{p_{11}' p_{00}' p_{10}' p_{01}'},
\end{align}
where~\eqref{2nd term in 3} follows from the inequality $\log(1+x) \leq x$.
Let us now bound the three terms in \eqref{2nd term in 3}.
\begin{itemize} \itemsep .5em
    \item For the first term, we have
% under the conditions on number of non-trivial user-user pairs $\nontrivpin \leq n\Tilde{n}$, assumption $\nondensity \leq O(\frac{\nontrivpin\log{n}}{n})$ and sparsity constrain on $p_{11}$ (\ref{Eq sparsity constrains p11}), by taking $x_{11}$ to be (\ref{x_11}) 
\begin{align}
    &\nontrivpin(p_{11}'x_{11})^2  \nonumber\\
    &=\nontrivpin \left(\frac{\nondensity \log{n}}{\nontrivpin} +p_{11}\right)^2 \nonumber\\
    % \label{Lem7_3}
    &= \frac{\nondensity^2 (\log{n})^2}{\nontrivpin} + 2\nondensity (\log{n})p_{11} + \nontrivpin p_{11}^2\nonumber\\ 
    \label{Lem7_4}
    &=  O\left(\frac{\nondensity(\log{n})^3}{n} \right) + 2\nondensity (\log{n})p_{11} + \nontrivpin p_{11}^2\\
    \label{Lem7_5}
    &=  O\left(\frac{\nondensity(\log{n})^3}{n}+ \frac{\nondensity(\log{n})^2}{n} +\frac{\Tilde{n}(\log{n})^2}{n} \right) \\
    &=o(\nondensity+\Tilde{n}), \nonumber
\end{align}
where \eqref{Lem7_4} follows from the assumption  $\nondensity = O(\frac{\nontrivpin\log{n}}{n})$ in~\eqref{Eq: P(delta|nonDensity)} and~\eqref{Lem7_5} follows since  $p_{11} = O(\frac{\log{n}}{n})$ and $\nontrivpin \leq \Tilde{n}n$. 
\item For the second term in (\ref{2nd term in 3}), we have
\begin{align}
    &\nontrivpin p_{11}'x_{11}(p_{10}'+p_{01}')\nonumber\\
    &= (\nondensity \log{n}+ \nontrivpin p_{11})(p_{10}'+p_{01}') \nonumber\\
    \label{Lem7_6}
    &\leq \frac{\nondensity (p_{10}+p_{01})\log{n} + \Tilde{n}np_{11}(p_{10}+p_{01})} {1-p_{11}}\\
    \label{Lem7_8}
    &= O(\nondensity +\Tilde{n}),
\end{align}
where \eqref{Lem7_6} follows 
from $\nontrivpin \leq \Tilde{n}n$ and
\eqref{Lem7_8} follows since $p_{01}+p_{10} = O(\frac{1}{\log{n}})$, $p_{11} = O(\frac{\log{n}}{n})$, and $1-p_{11} = \Theta(1)$.

\item For the third term in (\ref{2nd term in 3}), we have
\begin{align}
    &\nontrivpin(x_{11}+1) \sqrt{p_{11}' p_{00}' p_{10}' p_{01}'} \nonumber\\
    &= \nontrivpin \left(\frac{\nondensity \log{n} + p_{11}\nontrivpin}{p_{11}'\nontrivpin} +1\right)\sqrt{p_{11}' p_{00}' p_{10}' p_{01}'} \nonumber\\
    &= (\nondensity\log{n}+p_{11} \nontrivpin +p_{11}'\nontrivpin)p_{00}'\sqrt{\frac{p_{10}p_{01}}{p_{11}p_{00}}} \nonumber\\
    \label{Lem 7_9}
    &\leq (\nondensity\log{n}+p_{11}  n \tilde{n} +p_{11}' n \tilde{n})p_{00}'\sqrt{\frac{p_{10}p_{01}}{p_{11}p_{00}}} \\
    \label{Lem7_10}
    &=  o(\nondensity+ \Tilde{n}).
\end{align}
Here (\ref{Lem 7_9}) follows since $\nontrivpin \leq \Tilde{n}n$.
\eqref{Lem7_10} follows since $p_{11}' = O(p_{11}) = O \left(\frac{\log n}{n}\right)$,
$p_{00}' = O(1)$, and $\frac{p_{10}p_{01}}{p_{11}p_{00}} = O\left(\frac{1}{(\log{n})^3}\right)$. 
\end{itemize}
\vspace{.5em}
In summary, the third term of \eqref{eq:step1} is upper bounded as
\begin{equation}
\label{Eq:3terms-2}
\left(\frac{\mathcal{A}_2(\bm{x_{11}}, z )}{(1-p_{11})^2}\right)^{\nontrivpin/2} \le \exp\{O(\Tilde{r}+\Tilde{n})\}.
\end{equation}
Finally, combining~\eqref{Eq:3terms}~\eqref{Eq:3terms-1}~\eqref{Eq:3terms-2}, we have
\begin{align*}
     &\prob \left(\delta_\pi(G_1,G_2) \leq 0\mid \nonDensity = \nondensity\right) \\
     &\leq (1-2\psi_{\mathrm{a}})^\frac{m\Tilde{n}}{2}\left( \frac{1}{\log{n}}\right)^{\nondensity} \exp\{O(\nondensity+\Tilde{n})\}  \\
     &\leq (1-2\psi_{\mathrm{a}})^\frac{m\Tilde{n}}{2} \left(\frac{e^{O(1)}}{\log{n}}\right)^{\nondensity} \left(e^{O(1)}\right)^{\Tilde{n}} \\
     &= z_3^{\Tilde{n}} z_4^{\nondensity} z_5 ^{\Tilde{n}}
\end{align*}
for some $z_4 = O (\frac{1}{\log{n}})$ and $z_5 = O(1)$.
% \sout{where we define $ z_4 \triangleq O (\frac{1}{\log{n}})$ and $z_5 \triangleq O(1)$.}

\vspace{.5em}
\textbf{\emph{Step 2.}}
We will prove that for any $\pi \in \SnTn$ and $\density =O ( n \log{n})$, 
\begin{equation}
\label{Eq:step2}
\prob(\delta_\pi(G_1,G_2)\leq 0 \mid \Density =\density) \leq z_6^{\Tilde{n}}
\end{equation}
for some $z_6 = \exp\{-\frac{2\density}{n} + \frac{m}{2}\log(1-2\psi_{\mathrm{a}})+O(1)\}$.

\vspace{.5em}
In this step, we will compute $\prob(\delta_\pi(G_1,G_2) \leq 0 | \Density=\density)$ through $\prob(\delta_\pi(G_1,G_2) \leq 0 | \nonDensity=\nondensity)$, which involves using properties of a Hypergeometric distribution. 

Recall a Hypergeometric distribution, denoted as $\Hyp(n,N,K)$, is the probability distribution of the number of marked elements out of the $n$ elements we draw without replacement from a set of size $N$ with $K$ marked elements. Let $\Phi_\mathrm{Hyp}(z)$ be the probability generating function for  $\Hyp(n,N,K)$ and  $\Phi_\mathrm{Bin}(z)$ be the probability generating function for a binomial distribution $\Bin(n, \frac{K}{N})$. A few useful properties of the two distributions are as follows.
\begin{itemize}
    \item The mean of $\mathrm{Hyp}(n,N,K)$ is $nK/N$.
    \item  For all $n,N,K \in \N$ and $z >0$, we have 
        $\Phi_\mathrm{Hyp}(z) \leq \Phi_\mathrm{Bin}(z)$ \cite{V.Ch-tail1979}.
    \item 
        $\Phi_\mathrm{Bin}(z) = \left(1+\tfrac{K}{N}(z-1)\right)^n$.
\end{itemize}

In our problem, we are interested in the random variable $\nonDensity|\Density=\density$. We treat the set of moving user-user vertex pairs $\mathcal{E}_\mathrm{u} \cap \mathcal{E}_\mathrm{m}$ as a group of marked elements in $\mathcal{E}_\mathrm{u}$. 
From $\mathcal{E}_\mathrm{u}$, we consider drawing $\density$ vertex pairs and creating co-occurred edges for each chosen vertex pair. Along this line, the random variable  $\nonDensity| \Density=\density$, which is the number of co-occurred edges in $\mathcal{E}_\mathrm{m} \cap \mathcal{E}_\mathrm{u}$, represents the number of marked elements out of the $\density$ chosen elements and it follows a Hypergeometric distribution $\Hyp(\density, \vpin, \nontrivpin)$.
From this point and on, we always consider generating functions $\Phi_\mathrm{Hyp}(z)$ and $\Phi_\mathrm{Bin}(z)$ with parameters $n=\density$, $N= \vpin$, $K=\nontrivpin$.
Moreover, from \cite[Lemma~\rom{4}.5] {Cul-Kiy-exact2017}, we have the following upper bound on $\Phi_\mathrm{Hyp}(z)$ for any $z \in (0,1)$
    \begin{align} 
    \label{Eq:Hye gen upper bound}
        \Phi_\mathrm{Hyp}(z) \leq \exp\left \{ \tfrac{\density \Tilde{n}}{n}(-2+\tfrac{e}{n-1}+2e z)\right \}.
    \end{align}
Now, we are ready for proving \eqref{Eq:step2}. We first write 
\begin{align}
    &\prob(\delta_\pi(G_1,G_2) \leq 0 \mid \Density = \density)\nonumber\\
    &= \prob(\delta_\pi(G_1,G_2) \leq 0 ,\nonDensity \leq \nondensity^* \mid \Density =\density)\nonumber\\ 
    \label{Eq:step2-1}
    &+\prob(\delta_\pi(G_1,G_2)  \leq 0,\nonDensity > \nondensity^* \mid \Density =\density).
\end{align}
Here we set $\nondensity^*= C \E[\nonDensity\mid\Density =\density] =C \frac{\density\nontrivpin}{\vpin}$, where $C>0$ is some positive constant to be specified later. 
Note that $\vpin = \binom{n}{2}$ and $r =O(n\log n)$ from the assumption, then we have $\nondensity^* = O(\frac{\nontrivpin \log n}{n})$.

\begin{itemize} \itemsep .5em
    \item 
For the first term in~\eqref{Eq:step2-1}, we have
\begin{align}
    &\prob(\delta_\pi(G_1,G_2)  \leq 0,\nonDensity \leq \nondensity^* \mid \Density =\density) \nonumber\\
    \label{Lem 8_0}
    &= \sum_{\nondensity \leq \nondensity^*} \prob(\nonDensity = \nondensity \mid \Density = \density) \,\ \prob(\delta_\pi(G_1, G_2)  \leq 0 \mid \nonDensity = \nondensity) \\
    \label{Lem 8_11}
    &\leq \sum_{\nondensity \leq \nondensity^*}\prob(\nonDensity = \nondensity \mid \Density = \density) z_3^{\Tilde{n}}z_4^{\nondensity}z_5^{\Tilde{n}}\\
    &\leq  z_3^{\Tilde{n}}z_5^{\Tilde{n}} \sum_{\nondensity=0}^n\prob(\nonDensity = \nondensity \mid \Density = \density) z_4^{\nondensity} \nonumber\\
    \label{Lem 8_2}
    & = {z_3^{\Tilde{n}} z_5^{\Tilde{n}}  \Phi_{\mathrm{Hyp}}(z_4)}\\
    \label{Lem 8_3}
    & {\leq z_3^{\Tilde{n}} z_5^{\Tilde{n}}   \exp \left\{ \tfrac{\Tilde{n} \density}{n}(-2+\tfrac{e}{n-1}+2e z_4)\right\} }\\
    \label{Lem 8_4}
    & {= z_3^{\Tilde{n}} (e^{O(1)})^{ \Tilde{n}}  \exp \left\{- \tfrac{2\Tilde{n}  \density}{n} + \tfrac{e \Tilde{n}  \density}{n(n-1)} +O(\tfrac{1}{\log n}) \tfrac{\Tilde{n}  \density}{n} \right\} }\\
    \label{Lem 8_5}
    & {\leq z_3^{\Tilde{n}}  \exp \left \{\Tilde{n} \left(- \tfrac{2\density}{n} + O(1) \right)\right\} }
\end{align}
In \eqref{Lem 8_0}, we use the conditional independence of $\Density$ and $\delta_\pi(G_1,G_2)$ given $\nonDensity$, which can be proved as follows 
\begin{align}
    &\prob(\delta_{\pi}(G_1,G_2) \leq 0|\nonDensity=\nondensity, \Density=\density)\nonumber\\
    &=\frac{\prob(\delta_{\pi}(G_1,G_2) \leq 0, \nonDensity=\nondensity, \Density=\density)}{\prob(\nonDensity=\nondensity, \Density=\density)}\nonumber\\
    &= \frac{\prob(\delta_{\pi}(G_1,G_2) \leq 0, \nonDensity=\nondensity, \Density - \nonDensity=\density-\nondensity)}{\prob(\nonDensity=\nondensity,  \Density - \nonDensity=\density-\nondensity)}\nonumber\\
    &= \frac{\prob(\delta_{\pi}(G_1,G_2) \leq 0, \nonDensity=\nondensity) \prob(\Density - \nonDensity=\density-\nondensity)}{\prob(\nonDensity = \nondensity)\prob(\Density - \nonDensity=\density-\nondensity)}\label{Eq:indep}\\
    &= \prob((\delta_{\pi}(G_1,G_2) \leq 0 | \nonDensity = \nondensity),\nonumber
\end{align}
where \eqref{Eq:indep} follows from the fact that $\delta_\pi(G_1,G_2)$ and $\nonDensity$ are determined by $\mathcal{E}_\mathrm{m}$ while $\Density-\nonDensity$ is determined by those fixed vertex pairs.
In \eqref{Lem 8_11}, we have $\nondensity = O(\frac{\nontrivpin \log n}{n})$ and this inequality follows from \eqref{Eq: P(delta|nonDensity)} from Step 1. 
Equation \eqref{Lem 8_2} follows  from the definition of the probability generating function for $\Hyp(\density, \vpin, \Tilde{\vpin})$.
\eqref{Lem 8_3} follows form the conclusion about probability generating function of the hypergeometric distribution in \eqref{Eq:Hye gen upper bound} with $z_4 \in (0,1)$.
\eqref{Lem 8_4} is true since $z_4 = O \left(\frac{1}{\log n}\right)$ and $z_5 = O(1)$. 
\eqref{Lem 8_5} is true since $\density = O \left({n}{\log n}\right)$.

\item For the second term of \eqref{Eq:step2-1}, we have
\begin{align}
    &\prob(\delta_\pi(G_1,G_2)\leq 0,\nonDensity > \nondensity^* \mid \Density =\density)\nonumber\\ 
    &= \sum_{\nondensity > \nondensity^*}  \prob(\delta_\pi(G_1,G_2) \leq 0,\nonDensity = \nondensity \mid \Density =\density)\nonumber\\ 
    \label{Eq:Lem8_6}
    &= \sum_{\nondensity > \nondensity^*}  \prob(\delta_\pi(G_1, G_2) \leq 0 \mid \nonDensity = \nondensity)  \prob(\nonDensity = \nondensity \mid \Density = \density)\\
    &\leq \mspace{-5mu}\max_{0\le\nondensity\le n } \{\prob(\delta_\pi(G_1, G_2) \leq 0 | \nonDensity = \nondensity)\}  \prob(\nonDensity > \nondensity^* | \Density = \density).\nonumber
\end{align}

Here \eqref{Eq:Lem8_6} follows from the conditional independence of $\delta_{\pi}(G_1,G_2)$ and $\Density$ given $\nonDensity$.
To find this maximum probability, we consider the extreme case. 
Recall that $\delta_\pi = w_1(\Delta^\mathrm{u}(G_1,\pi(G_2))-\Delta^\mathrm{u}(G_1,G_2)) + w_2(\Delta^\mathrm{a}(G_1,\pi(G_2))-\Delta^\mathrm{a}(G_1,G_2))$. We have that $w_2(\Delta^\mathrm{a}(G_1,\pi(G_2))-\Delta^\mathrm{a}(G_1,G_2))$ is independent of $\nonDensity$. 
From the upper bound on generating function in \eqref{Lem7_2}, we consider $\pi^{\mathcal{E}}$ consisting of only 2-cycles.
Since $\Delta^\mathrm{u}(G_1,\pi(G_2))-\Delta^\mathrm{u}(G_1,G_2) > 0$ only if there exist user-user vertex pairs such that $(G_1(e), G_2(e)) = (1, 1)$ and $(G_1(\pi ^{\mathcal{E}}(e)) G_2(\pi ^{\mathcal{E}}(e))) = (0, 0)$, we have 
$ \Delta^\mathrm{u}(G_1,\pi(G_2))-\Delta^\mathrm{u}(G_1,G_2) \leq 0$ with probability $1$ given $\nonDensity = 0$.
Therefore, given $\nonDensity=0$ the probability that $\delta_\pi \leq 0$ is maximized.
We have
\begin{align}
    &\max_{0\le\nondensity\le n} \{\prob(\delta_\pi(G_1, G_2) \leq 0 \mid \nonDensity = \nondensity)\} \nonumber\\
    % \label{eq:add1}
    & \leq  \prob(\delta_\pi(G_1, G_2) \leq 0 \mid \nonDensity = 0)\nonumber\\
    \label{eq:add2}
    & \leq z_3^{\Tilde{n}}z_5^{\Tilde{n}},
    % & \ning{= \prob(\Delta^\mathrm{a}(G_1,\pi(G_2))-\Delta^\mathrm{a}(G_1,G_2) \leq 0| \nonDensity=\nontrivpin)}\\
    % \label{eq:add3}
    % & \ning{=\prob(\Delta^\mathrm{a}(G_1,\pi(G_2))-\Delta^\mathrm{a}(G_1,G_2) \leq 0)}\\
    % \label{eq:add4}
    % & \ning{= \sum_{d\leq0}[z^d] \prod_{l\geq1} \mathcal{A}_l(\bm{q},z)^{t_l^{\mathrm{a}}}}\\
    % \label{eq:add5}
    % & \ning{\leq \prod_{l\geq1} \mathcal{A}_l(\bm{q},z)^{t_l^{\mathrm{a}}}}\\
    % \label{eq:add6}
    % & \ning{\leq \mathcal{A}_2(\bm{q},z)^{m\Tilde{n}}}\\
    % \label{eq:add7}
    % & \ning{\leq (1-2\psia)^{\frac{m\Tilde{n}}{2}} = z_3^{\tilde{n}}}.
\end{align}
where \eqref{eq:add2} follows from \eqref{Eq: P(delta|nonDensity)} in Step 1 with $\nondensity=0$.
% For the equality \eqref{eq:add2}, recall the definition of $\delta_\pi$ in \eqref{Eq:def delta_pi} and notice that given $\nonDensity = \nontrivpin$, $ \Delta^\mathrm{u}(G_1,\pi(G_2))-\Delta^\mathrm{u}(G_1,G_2) =0$.
% \eqref{eq:add3} follows from the independence of $\Delta^\mathrm{a}$ and $\nonDensity$ since $\Delta^\mathrm{a}$ is a function on $\mathcal{E}_\mathrm{a}$ while $\nonDensity$ is a function on  $\mathcal{E}_\mathrm{u}$.
% \eqref{eq:add4} follows from the fact that $\prod_{l\geq1} \mathcal{A}_l(\bm{q},z)^{t_l^{\mathrm{a}}}$ is the probability generating function for $w_2 (\Delta^\mathrm{a}(G_1,\pi(G_2))-\Delta^\mathrm{a}(G_1,G_2)) $ with formal variable $z$.
% In \eqref{eq:add5}, we set $z\in (0,1)$ and the inequality follows from \eqref{eq:ogf coefficient} in Fact~\ref{Fact:ogf coefficient}.
% \eqref{eq:add6} follows from Fact~\ref{Fact:l and 2 cycle}.
% \eqref{eq:add7} follows from the upper bound on $\mathcal{A}_2(\bm{q},z)$ derived in \eqref{Eq:def-psiu}. 
Now we get
\begin{align}
    &\prob(\delta_\pi(G_1,G_2)\leq 0,\nonDensity > \nondensity^* \mid \Density =\density)\nonumber\\ 
    % &= \prob(\delta_\pi(G_1, G_2) \leq 0 \mid \nonDensity = \nontrivpin) \prob(\nonDensity > \nondensity^* \mid \Density = \density)\\
    & \leq z_3^{\Tilde{n}}z_5^{\Tilde{n}} \prob(\nonDensity > \nondensity^* \mid \Density = \density)\nonumber\\
    \label{Eq:Lem8_7}
    &= z_3^{\Tilde{n}}z_5^{\Tilde{n}}\sum_{i>\nondensity^*} [z^i] \Phi_\mathrm{Hyp}(z)\\
    \label{Eq:Lem8_7.2}
    & \leq  z_3^{\Tilde{n}} z_5^{\Tilde{n}} z^{-\nondensity^*} \Phi_{\Hyp}(z)\\
    \label{Eq:Lem8_8}
    & {\leq z_3^{\Tilde{n}}z_5^{\Tilde{n}} z^{-\nondensity^*} \Phi_\mathrm{Bin}(z)}\\
    \label{Eq:Lem8_9}
    & {= z_3^{\Tilde{n}}z_5^{\Tilde{n}} z^{-\nondensity^*} \left(1+\tfrac{\nontrivpin}{\vpin}(z-1)\right)^{\density}}\\
    \label{Eq:Lem8_10}
    & {\leq z_3^{\Tilde{n}}z_5^{\Tilde{n}} z^{-\nondensity^*} \exp \left \{ \tfrac{\density \nontrivpin}{\vpin}(z-1) \right\}}\\
    \label{Eq:Lem8_11}
    &{= z_3^{\Tilde{n}}z_5^{\Tilde{n}} \exp \left \{-\nondensity^* +\tfrac{\density \nontrivpin}{\vpin}(e-1) \right\}}\\
    \label{Eq:Lem8_12}
    &={z_3^{\Tilde{n}} z_5^{\Tilde{n}}\exp \left \{\tfrac{\density \nontrivpin}{\vpin}(-C-1+e) \right\}}\\
    \label{Eq:Lem8_13}
    &{\leq z_3^{\Tilde{n}}z_5^{\Tilde{n}} \exp \left\{\tfrac{\density \Tilde{n}(n-2)}{n(n-1)}(-C-1+e) \right\}}\\
    \label{Eq:Lem8_14}
    & {\leq z_3^{\Tilde{n}} \exp \left\{\Tilde{n} \left(\tfrac{\density }{n}(-C-1+e)+ O(1)\right) \right\}}\\
    \label{Eq:Lem8_15}
    & = o\left(z_3^{\Tilde{n}}  \exp \left \{\Tilde{n} \left(- \tfrac{2\density}{n} + O(1) \right)\right\}\right)
\end{align}
In \eqref{Eq:Lem8_7}, $\Phi_\mathrm{Hyp}(z)$ is a probability generating function for $\Hyp(\density,\vpin,\nontrivpin)$.
In \eqref{Eq:Lem8_7.2}, we set $z>1$ and the inequality follows from \eqref{eq:ogf sum-coefficient2} in Fact~\ref{Fact:ogf coefficient}.
In \eqref{Eq:Lem8_8}, $\Phi_\mathrm{Bin}(z)$ is a probability generating function for $\Bin(\density,\frac{\nontrivpin}{\vpin})$ and this inequality follows from the property of a Hypergeometric distribution. 
\eqref{Eq:Lem8_9} follows from the definition of $\Phi_\mathrm{Bin}(z)$. 
% $\Phi_\mathrm{Bin}(z) = \left(1+\frac{\nontrivpin}{\vpin}(z-1)\right)^{\density}$.
\eqref{Eq:Lem8_10} follows form the inequality $1+x \leq e^x$.
In \eqref{Eq:Lem8_11}, we set $z=e$.
In \eqref{Eq:Lem8_12}, we plug in $\nondensity^* = C\frac{\density \nontrivpin}{\vpin}$ where $C$ is larger than $(e-1)$.
In \eqref{Eq:Lem8_13}, we use the relation $\nontrivpin \geq \frac{\Tilde{n}(n-2)}{2}$ from \eqref{eqLem4_3} and $\vpin = {n \choose 2}$. 
In \eqref{Eq:Lem8_14}, we plug in $z_5 = O(1)$.
\eqref{Eq:Lem8_15} is true because we can always find $C>e+1$ such that \eqref{Eq:Lem8_14} is exponentially smaller than \eqref{Lem 8_5}.
\end{itemize}
\vspace{.5em}
We conclude that the second term of \eqref{Eq:step2-1} is negligible compared with the upper bound of the first term given in \eqref{Lem 8_5}.  
Combining the two terms,~\eqref{Eq:step2-1} can be bounded as
\begin{align*}
    &\prob(\delta_\pi(G_1,G_2)  \leq 0 \mid \Density =\density)\\
    &\leq \exp\left\{\Tilde{n}\left(-\tfrac{2\density}{n} + \tfrac{m}{2}\log(1-2\psi_{\mathrm{a}}) +O(1) \right)\right\}\\
    &= z_6^{\Tilde{n}}.
\end{align*}

\vspace{.5em}
\textbf{\emph{Step 3.}} {We now establish the desired error bound
\[
\prob(\exists \pi \in \Sn \setminus \{\pi_\mathrm{id}\}, \delta_\pi(G_1,G_2)  \leq 0\mid \Density = \density) \le 3n^2 z_6^2,
\]
where $z_6 = \exp\{-\frac{2\density}{n} + \frac{m}{2}\log(1-2\psi_{\mathrm{a}})+O(1)\}$.}

\vspace{.5em}
{When $n  z_6 > 2/3$, we have $$\prob(\exists \pi \in \Sn \setminus \{\pi_\mathrm{id}\}, \delta_\pi(G_1,G_2)  \leq 0\mid \Density = \density) \leq 1 \leq 3n^2  z_6^2.$$ Now assume that $n  z_6 \le 2/3$. We can bound}
\begin{align}
    &\prob(\exists \pi \in \Sn \setminus \{\pi_\mathrm{id}\}, \delta_\pi(G_1,G_2)  \leq 0 \mid \Density = \density) \nonumber\\ 
    \label{Lem 9_0}
    & \leq \sum_{\Tilde{n} =2}^{n} \sum_{\pi \in \SnTn} \prob(\delta_\pi(G_1,G_2)  \leq 0 \mid \Density = \density)\\
    &\leq \sum_{\Tilde{n} =2}^{n} |\SnTn| \max_{\pi \in \SnTn}\{\prob( \delta_\pi(G_1,G_2)  \leq 0 \mid \Density =\density)\}\nonumber\\
    \label{Lem 9_1}
    &\leq \sum_{\Tilde{n} =2}^{n} |\SnTn|  z_6^{\Tilde{n}}\\
    \label{Lem 9_2}
    &\leq \sum_{\Tilde{n} =2}^{n} n^{\Tilde{n}}   z_6^{\Tilde{n}}\\
    &\le \frac{(n   z_6)^2}{1-n   z_6} \nonumber\\
    \label{Lem 9_3}
    &\leq 3n^2  z_6^2,
\end{align}
where \eqref{Lem 9_0} follows from the union bound,
\eqref{Lem 9_1} follows from inequality \eqref{Eq:step2} proved in Step 2, \eqref{Lem 9_2} follows since $|\SnTn| \leq n^{\tilde{n}}$, and \eqref{Lem 9_3} holds since $n  z_6 \leq 2/3$.

In summary, $3n^2  z_6^2$ is always an upper bound on the conditional probability. This completes the proof of Lemma~\ref{Lemma:approximated union bound}.
\end{proof}

\section{Conclusion and Future Work}
\label{sec:discussion}

%\ning{Add more on conclusion (abstract style)}\\
%\ning{
%In this paper, we proposed the attributed \er\ pair model, characterized the information theoretic limits on exact alignment, and specialized our results for comparison with three well-studied graph alignment models. The current limitation is that there is still a gap between our achievability results and the converse result. Here we highlight several extensions and potential future directions for our work.
%}
In this paper, we proposed the attributed \er\ pair model to study the effect of publicly available side information for graph alignment.
We established information-theoretic limits for exact alignment, including achievability
and converse conditions that match for a certain range of parameters.
These conditions can be used to quantify the effect of side information.
We also specialized our results to three well-studied graph alignment models for comparison.

There are many more interesting questions to ask about the attributed graph alignment problem.
Here we give some example directions.
As discussed in Section~\ref{sec:main-result}, our achievability conditions and converse conditions do not match in the most general scenario.
We conjecture that the converse conditions can be potentially improved, especially given the recent developments on tighter converse conditions for the \er\ pair model in \cite{settling-TIT}.
Moreover, the achievability results in this work is based on the MAP estimator, which has no polynomial-time implementation. Therefore, a natural question is whether there is any efficient algorithms for attributed graph alignment, and whether there is any fundamental gap between the achievable region by efficient algorithms and the information-theoretic achievable region. This question is partially answered in the concurrent work \cite{wang-2022-feasible} by proposing two efficient algorithms for attributed graph alignment and analyzing their feasible region.
Another direction that is worth further investigation is graph alignment under more general attributed graph models.
Our model has assumed that the user-attribute edges are independent of the user-user edges.
However, in the social network example, users attending the same institute are more likely to be friends than users attending different institutes.
Therefore, it would be interesting to consider graph models in which user-attribute edges are correlated with user-user edges, and to investigate how the correlation affects graph alignment.
We comment that a starting point can be the multiplicative attribute graph model proposed in \cite{kim-2012-multiplicative}, where the probability of a user-user edge depends on the product of individual attribute-attribute similarity.

\appendices
\label{sec:appendix}
% \subsection{MAP estimator}\label{sec:appendix}
% \begin{replemma}{Lemma:MAP}
% Let $(G_1,G_2')$ be an observable pair generated from the attributed \er\ pair $\mathcal{G}(n,\bm{p};m,\bm{q})$. The MAP estimator of the permutation $\Pi^*$ based on $(G_1,G_2')$ simplifies to
% \begin{align*}
%     &\hat{\pi}_\mathrm{MAP}(G_1,G_2') \\
%     &= \argmin_{\pi \in \mathcal{S}_n}\{w_1 \Delta ^{\mathrm{u}}(G_1,\pi^{-1}(G_2')) 
%     + w_2 \Delta ^{\mathrm{a}}(G_1,\pi^{-1}(G_2'))\},
%     \end{align*}
% where $w_1 = \log\left(\frac{p_{11}p_{00}}{p_{10}p_{01}}\right)$, $w_2 = \log \left(\frac{q_{11}q_{00}}{q_{10}q_{01}}\right)$, and
% \begin{align*}
%     \Delta ^{\mathrm{u}}(G_1,\pi^{-1}(G_2')) &=\mspace{-12mu} \sum_{(i,j)\in \mathcal{E}_{\mathrm{u}}} \mspace{-12mu} \mathbb{1} \{G_1((i,j)) \neq G_2'((\pi(i),\pi(j)))\},\\
%     \Delta ^{\mathrm{a}}(G_1,\pi^{-1}(G_2')) &= \mspace{-12mu}\sum_{(i,v) \in \mathcal{E}_{\mathrm{a}}} \mspace{-12mu} \mathbb{1} \{G_1((i,v)) \neq G_2'((\pi(i),v))\}.
% \end{align*}
% \end{replemma}
\newcommand{\vuu}{\mathcal{V}^{u'}}
\newcommand{\vs}{\mathcal{V}^{s}}

\section{MAP estimator}\label{appd:pfMAP}
In this section, we derive the expression of the MAP for the attributed \er{} graph pair model $\mathcal{G}(n, \bm{p};m,\bm{q})$.
\LemmaMAP*
\begin{proof}
\label{Proof:Lemma MAP}
Let $(g_1,g_2')$ be a realization of an observable pair $(G_1,G_2')$ from $\mathcal{G}(n, \bm{p};m,\bm{q})$.  Then the posterior of the permutation $\Pi^*$ can be written as:
\begin{align}
    & \prob (\Pi^* = \pi| G_1 = g_1, G_2' = g_2') \nonumber\\
    =& \frac{\prob(G_1 = g_1, G_2' = g_2'|\Pi^* = \pi) \prob (\Pi^* =\pi)}{P(G_1 = g_1, G_2' = g_2')}\nonumber \\
    \label{eqLem 1_2}
    \propto& \prob (G_1 = g_1, G_2' = g_2'|\Pi^* = \pi)\\
    \label{eqLem 1_3}
    =& \prob(G_1 = g_1, G_2 = \pi^{-1}(g_2'))\\
    \label{eqLem 1_4}
    =& \prod _{(i,j)\in \{0,1\}^2}p_{ij}^{\muu_{ij}(g_1,\pi^{-1}(g_2'))} q_{ij}^{\mua_{ij}(g_1,\pi^{-1}(g_2'))}.
\end{align}
Here equation \eqref{eqLem 1_2} follows from the fact that $\Pi^*$ is uniformly drawn from $\mathcal{S}_n$ and $\prob(G_1=g_1, G_2'=g_2')$ does not depend on $\pi$. Equation \eqref{eqLem 1_3} is due to the independence between $\Pi^*$ and $(G_1,G_2)$. 

To further simplify equation \eqref{eqLem 1_4}, note that the total number of edges in a graph is invariant under any permutation. 
We define $\beta ^{\mathrm{u}}(G_1)$ as the total number of user-user edges in graph $G_1$ and $\beta ^{\mathrm{u}}(\pi^{-1}(G_2'))$ for graph $\pi^{-1}(G_2')$. Similarly, we define $\beta ^{\mathrm{a}}(G_1)$ and $\beta ^{\mathrm{a}}(\pi^{-1}(G_2'))$ as the total number of user-attribute edges for graph $G_1$ and $\pi^{-1}(G_2')$, respectively.
Recall our definitions on Hamming distance $\Delta^{\mathrm{u}}(G_1,\pi^{-1}(G_2'))$ and $\bm{\mu}(G_1,\pi^{-1}(G_2'))$, and notice that $\Delta^{\mathrm{u}}(G_1,\pi^{-1}(G_2')) =\muu_{10}+\muu_{01}$. Moreover, we have $\beta ^{\mathrm{u}}(G_1) = \muu_{11} + \muu_{10} $ and $\beta ^{\mathrm{u}}(G_2) =  \beta ^{\mathrm{u}}(\pi^{-1}(G_2'))=\muu_{11} + \muu_{01}$. Then, for the user-user set $\mathcal{E}_{\mathrm{u}}$, we have
\begin{align*}
    \mu_{11}  &= \frac{\beta ^{\mathrm{u}}(G_1)+\beta ^{\mathrm{u}}(\pi^{-1}(G_2'))}{2} - \frac{\Delta ^{\mathrm{u}}(G_1, \pi^{-1}(G_2'))}{2}\\
    \mu_{10}  &= \frac{\beta ^{\mathrm{u}}(G_1)-\beta ^{\mathrm{u}}(\pi^{-1}(G_2'))}{2} + \frac{\Delta ^{\mathrm{u}}(G_1, \pi^{-1}(G_2'))}{2}\\
    \mu_{01}  &= \frac{\beta ^{\mathrm{u}}(\pi^{-1}(G_2'))-\beta ^{\mathrm{u}}(G_1)}{2} + \frac{\Delta ^{\mathrm{u}}(G_1, \pi^{-1}(G_2'))}{2}\\
    \mu_{00}  &= {n \choose 2} \mspace{-5mu}-\mspace{-5mu} \frac{\beta ^{\mathrm{u}}(G_1)+\beta ^{\mathrm{u}}(\pi^{-1}(G_2'))}{2} \mspace{-5mu}-\mspace{-5mu} \frac{\Delta ^{\mathrm{u}}(G_1, \pi^{-1}(G_2'))}{2}.
\end{align*}
Similarly, for the user-attribute set $\mathcal{E}_{\mathrm{a}}$,  we have $\Delta^\mathrm{a}(G_1,\pi^{-1}(G_2')) = \mua_{10}+\mua_{01}$, $\beta^{\mathrm{a}}(G_1) =\mua_{11}+\mua_{10}$ and  $\beta^{\mathrm{a}}(G_2) =\beta^{\mathrm{a}}(\pi^{-1}(G_2')) =\mua_{11}+\mua_{01}$. Therefore, we get the following.
\begin{align*}
    \mua_{11} &= \frac{\beta ^{\mathrm{a}}(G_1)+\beta ^{\mathrm{a}}(\pi^{-1}(G_2'))}{2} - \frac{\Delta ^{\mathrm{a}}(G_1, \pi^{-1}(G_2'))}{2}\\
    \mua_{10} &= \frac{\beta ^{\mathrm{a}}(G_1)-\beta ^{\mathrm{a}}(\pi^{-1}(G_2'))}{2} + \frac{\Delta ^{\mathrm{a}}(G_1, \pi^{-1}(G_2'))}{2}\\
    \mua_{01} &= \frac{\beta ^{\mathrm{a}}(\pi^{-1}(G_2'))-\beta ^{\mathrm{a}}(G_1)}{2} + \frac{\Delta ^{\mathrm{a}}(G_1, \pi^{-1}(G_2'))}{2}\\
    \mua_{00} &= nm -  \frac{\beta ^{\mathrm{a}}(G_1)+\beta ^{\mathrm{a}}(\pi^{-1}(G_2'))}{2} - \frac{\Delta ^{\mathrm{a}}(G_1, \pi^{-1}(G_2'))}{2}.
\end{align*}
Since $\beta^{\mathrm{u}}(G_1),\beta ^{\mathrm{u}}(\pi^{-1}(G_2')),\beta^{\mathrm{a}}(G_1)$, and $\beta ^{\mathrm{a}}(\pi^{-1}(G_2'))$ do not depend on $\pi$, we can further simplify the posterior as follows
\begin{align}
    & \prob (\Pi^* = \pi| G_1 = G_1, G_2' = G_2') \nonumber\\
    \propto& \prod _{(i,j)\in \{0,1\}^2}p_{ij}^{\muu_{ij}(G_1,\pi^{-1}(G_2'))} q_{ij}^{\mua_{ij}(G_1,\pi^{-1}(G_2'))} \nonumber\\
    \label{Lem 1_5}
    \propto & \left( \frac{p_{11}p_{00}}{p_{10}p_{01}}\right)^{-\frac{\Delta ^{\mathrm{u}}(G_1,\pi^{-1}(G_2' ))}{2}} \left(\frac{q_{11}q_{00}}{q_{10}q_{01}}\right)^{- \frac{\Delta ^{\mathrm{a}}(G_1,\pi^{-1}(G_2' ))}{2}} \\
    \label{Lem 1_6}
    =& \exp{\left\{ -w_1 \frac{\Delta ^{\mathrm{u}}(G_1,\pi^{-1}(G_2' ))}{2} - w_2\frac{\Delta ^{\mathrm{a}}(G_1,\pi^{-1}(G_2' ))}{2} \right\}},
\end{align}
where $w_1 \triangleq \log\left(\frac{p_{11}p_{00}}{p_{10}p_{01}}\right)$ and $w_2 \triangleq \log\left(\frac{q_{11}q_{00}}{q_{10}q_{01}}\right)$. Note that $w_1>0$ and $w_2>0$ since we assume that the edges of $G_1$ and $G_2$ are positively correlated. Therefore, of all the permutations in $S_n$, the one that minimizes the weighted Hamming distance $w_1 \Delta ^{\mathrm{u}}(G_1,\pi^{-1}(G_2')) + w_2\Delta ^{\mathrm{a}}(G_1,\pi^{-1}(G_2'))$ achieves the maximum posterior probability.
\end{proof}

\section{Proof of Lemma~\ref{lemma:seed-attri}}\label{appx:PfLem1}
Recall that when we specialize the attributed \er\ pair model by setting $\bm{p}=\bm{q}$, we can treat the $m$ attributes as $m$ seeds.
The only difference between the $\mathcal{G}(n,\bm{p}; m ,\bm{p})$ model and the seeded model $\mathcal{G}(n,m,\bm{p})$ is that there are no edges between seeds in the specialized model, but those edges exist in the seeded model. Here, we show that this distinction has no influence on the information-theoretic limit of exact alignment. To see this, we prove that the optimal estimator for the seeded \er\ pair---the MAP estimator---also simplifies to minimizing the Hamming distance of the user-user edges and user-seed edges.
\begin{lemma}
\label{appx:lemMAP2}
Let $(G_1,G_2')$ be a pair of graphs generated from the seeded \er\ pair $\mathcal{G}(n,m, \bm{p})$. The MAP estimator of the permutation $\Pi^*$ based on $(G_1,G_2')$ simplifies to
\begin{align*}
    &\hat{\pi}_\mathrm{MAP}(G_1,G_2') \\
    &= \argmin_{\pi \in \mathcal{S}_n}\{\Delta ^{\mathrm{u}}(G_1,\pi^{-1}(G_2')) 
    + \Delta ^{\mathrm{a}}(G_1,\pi^{-1}(G_2'))\},
    \end{align*}
where 
\begin{align*}
    \Delta ^{\mathrm{u}}(G_1,\pi^{-1}(G_2')) &=\mspace{-12mu} \sum_{(i,j)\in \mathcal{E}_{\mathrm{u}}} \mspace{-12mu} \mathbb{1} \{G_1((i,j)) \neq G_2'((\pi(i),\pi(j)))\},\\
    \Delta ^{\mathrm{a}}(G_1,\pi^{-1}(G_2')) &= \mspace{-12mu}\sum_{(i,v) \in \mathcal{E}_{\mathrm{a}}} \mspace{-12mu} \mathbb{1} \{G_1((i,v)) \neq G_2'((\pi(i),v))\}.
\end{align*}
\end{lemma}
\begin{proof}
To start, we have the posterior of the underlying permutation.
\begin{align}
    & \prob (\Pi^* = \pi| G_1 = g_1, G_2' = g_2') \nonumber\\
    &= \frac{\prob(G_1 = g_1, G_2' = g_2'|\Pi^* = \pi) \prob (\Pi^* =\pi)}{P(G_1 = g_1, G_2' = g_2')}\nonumber \\
    \label{eq-app1_1}
    &\propto \prob (G_1 = g_1, G_2' = g_2'|\Pi^* = \pi)\\
    \label{eq-app1_2}
    &= \prob(G_1 = g_1, G_2 = \pi^{-1}(g_2')).
\end{align}
Here \eqref{eq-app1_1} follows since $\Pi^*$ is uniformly drawn. \eqref{eq-app1_2} follows since $\Pi^*$ is independent of $G_1$ and $G_2$. 
For ease of notation, we use $g_2^{\pi}$ to denote $\pi^{-1}(g_2')$.
Then according to the seeded graph model in Section~\ref{sec:model}, we have 
\begin{align}
    & \prob(G_1 = g_1, G_2 = g_2^{\pi})\nonumber\\
    \label{eq-app1_3}
    =& p_{11}^{\mu_{11}(g_1,g_2^{\pi} )} 
    p_{00}^{\mu_{00}(g_1, g_2^{\pi})} p_{10} ^{\mu_{10}(g_1, g_2^{\pi})}p_{01}^{\mu_{01} (g_1, g_2^{\pi})}.
\end{align}
In \eqref{eq-app1_3}, we define  
\begin{align*}
    &\mu_{11}(g_1, g_2^{\pi})  \triangleq \sum_{i,j \in \vuu} \mathbb{1}_{\left\{i\stackrel{g_1}{\sim}j, i\stackrel{g_2^{\pi}}{\sim}j \right\}} + \sum_{i,j \in \vs} \mathbb{1}_{\left\{i\stackrel{g_1}{\sim}j, i\stackrel{g_2^{\pi}}{\sim}j \right\}}\\
    &\hspace{2em} +\sum_{i \in \vuu, j \in \vs} \mathbb{1}_{\left\{i\stackrel{g_1}{\sim}j, i\stackrel{g_2^{\pi}}{\sim}j \right\}}+\sum_{i \in \vs, j \in \vuu} \mathbb{1}_{\left\{i\stackrel{g_1}{\sim}j, i\stackrel{g_2^{\pi}}{\sim}j \right\}}\\
    &\mu_{10}(g_1, g_2^{\pi}) \triangleq  \sum_{i,j \in \vuu} \mathbb{1}_{\left\{i\stackrel{g_1}{\sim}j, i\stackrel{g_2^{\pi}}{\not\sim}j \right\}} + \sum_{i,j \in \vs} \mathbb{1}_{\left\{i\stackrel{g_1}{\sim}j, i\stackrel{g_2^{\pi}}{\not\sim}j \right\}}\\
    &\hspace{2em} +\sum_{i \in \vuu, j \in \vs} \mathbb{1}_{\left\{i\stackrel{g_1}{\sim}j, i\stackrel{g_2^{\pi}}{\not\sim}j \right\}}+\sum_{i \in \vs, j \in \vuu} \mathbb{1}_{\left\{i\stackrel{g_1}{\sim}j, i\stackrel{g_2^{\pi}}{\not\sim}j \right\}}\\
    &\mu_{01}(g_1, g_2^{\pi})  \triangleq \sum_{i,j \in \vuu} \mathbb{1}_{\left\{i\stackrel{g_1}{\not\sim}j, i\stackrel{g_2^{\pi}}{\sim}j \right\}} + \sum_{i,j \in \vs} \mathbb{1}_{\left\{i\stackrel{g_1}{\not\sim}j, i\stackrel{g_2^{\pi}}{\sim}j \right\}}\\
    &\hspace{2em}  +\sum_{i \in \vuu, j \in \vs} \mathbb{1}_{\left\{i\stackrel{g_1}{\not\sim}j, i\stackrel{g_2^{\pi}}{\sim}j \right\}}+\sum_{i \in \vs, j \in \vuu} \mathbb{1}_{\left\{i\stackrel{g_1}{\not\sim}j, i\stackrel{g_2^{\pi}}{\sim}j \right\}}\\
    &\mu_{00}(g_1, g_2^{\pi})  \triangleq \sum_{i,j \in \vuu} \mathbb{1}_{\left\{i\stackrel{g_1}{\not\sim}j, i\stackrel{g_2^{\pi}}{\not\sim}j \right\}} + \sum_{i,j \in \vs} \mathbb{1}_{\left\{i\stackrel{g_1}{\not\sim}j, i\stackrel{g_2^{\pi}}{\not\sim}j \right\}}\\
    &\hspace{2em}  +\sum_{i \in \vuu, j \in \vs} \mathbb{1}_{\left\{i\stackrel{g_1}{\not\sim}j, i\stackrel{g_2^{\pi}}{\not\sim}j \right\}}+\sum_{i \in \vs, j \in \vuu} \mathbb{1}_{\left\{i\stackrel{g_1}{\not\sim}j, i\stackrel{g_2^{\pi}}{\not\sim}j \right\}}.
\end{align*}
where {$\vuu \triangleq \vu \setminus \vs$ is the set of unmatched user vertices} and $\vs$ is the set of seed vertices. Notice that the term summing seed-seed edges is always the same for every $\pi \in \mathcal{S}_{\mathrm{u}}$ since we only permute user vertices. 
Here, we define 
\begin{align*}
    &\mu_{11}'(g_1, g_2^{\pi}) \triangleq \sum_{i,j \in \vuu} \mathbb{1}_{\left\{i\stackrel{g_1}{\sim}j, i\stackrel{g_2^{\pi}}{\sim}j \right\}}\\
    & \hspace{2em} +\sum_{i \in \vuu, j \in \vs} \mathbb{1}_{\left\{i\stackrel{g_1}{\sim}j, i\stackrel{g_2^{\pi}}{\sim}j \right\}}+\sum_{i \in \vs, j \in \vuu} \mathbb{1}_{\left\{i\stackrel{g_1}{\sim}j, i\stackrel{g_2^{\pi}}{\sim}j \right\}}\\
    &\mu_{10}'(g_1, g_2^{\pi}) \triangleq  \sum_{i,j \in \vuu} \mathbb{1}_{\left\{i\stackrel{g_1}{\sim}j, i\stackrel{g_2^{\pi}}{\not\sim}j \right\}} \\
    &\hspace{2em}  +\sum_{i \in \vuu, j \in \vs} \mathbb{1}_{\left\{i\stackrel{g_1}{\sim}j, i\stackrel{g_2^{\pi}}{\not\sim}j \right\}}+\sum_{i \in \vs, j \in \vuu} \mathbb{1}_{\left\{i\stackrel{g_1}{\sim}j, i\stackrel{g_2^{\pi}}{\not\sim}j \right\}}\\
    &\mu_{01}'(g_1, g_2^{\pi})  \triangleq \sum_{i,j \in \vuu} \mathbb{1}_{\left\{i\stackrel{g_1}{\not\sim}j, i\stackrel{g_2^{\pi}}{\sim}j \right\}} \\
    &\hspace{2em}  +\sum_{i \in \vuu, j \in \vs} \mathbb{1}_{\left\{i\stackrel{g_1}{\not\sim}j, i\stackrel{g_2^{\pi}}{\sim}j \right\}}+\sum_{i \in \vs, j \in \vuu} \mathbb{1}_{\left\{i\stackrel{g_1}{\not\sim}j, i\stackrel{g_2^{\pi}}{\sim}j \right\}}\\
    &\mu_{00}'(g_1, g_2^{\pi})  \triangleq \sum_{i,j \in \vuu} \mathbb{1}_{\left\{i\stackrel{g_1}{\not\sim}j, i\stackrel{g_2^{\pi}}{\not\sim}j \right\}} \\
    & \hspace{2em} +\sum_{i \in \vuu, j \in \vs} \mathbb{1}_{\left\{i\stackrel{g_1}{\not\sim}j, i\stackrel{g_2^{\pi}}{\not\sim}j \right\}}+\sum_{i \in \vs, j \in \vuu} \mathbb{1}_{\left\{i\stackrel{g_1}{\not\sim}j, i\stackrel{g_2^{\pi}}{\not\sim}j \right\}}.
\end{align*}
We therefore have
\begin{align}
    & \prob(G_1 = g_1, G_2 = g_2^{\pi})\nonumber\\
    &\propto p_{11}^{\mu_{11}'(g_1, g_2^{\pi})}
    p_{00}^{\mu_{00}'(g_1, g_2^{\pi})} p_{10} ^{\mu_{10}'(g_1, g_2^{\pi})}p_{01}^{\mu_{01}'(g_1, g_2^{\pi})}
    \label{eq-app1_4}
    % & \propto \left(\frac{ps^s [1-ps^2-2ps(1-s)]}{[ps(1-s)]^2}\right)^{-\mu_{10}'(g_1, g_2^{\pi}) - \mu_{01}'(g_1, g_2^{\pi})}
\end{align}
% Here \eqref{eq-app1_4} follows since $\mu_{11}+\mu_{10}$ is a fixed constant
So far the MAP estimator we derived here is exactly the same as the estimator for attributed graph alignment. Applying Lemma~\ref{Lemma:MAP}, we then get
\begin{align*}
    \hat{\pi}_\mathrm{MAP} (g_1, g_2')= \argmin_{\pi \in \mathcal{S}_{\mathrm{u}}} \{\mu_{10}'(g_1, g_2^{\pi}) + \mu_{01}'(g_1, g_2^{\pi})\}.
\end{align*}
\end{proof}

\section{Proof of Lemma~\ref{lemma:Aut-error-prob}}\label{appx:PfLem2}
\LemmaAUT*
\begin{proof}
    We assume without loss of generality that the true underlying permutation $\Pi^*$ us the identity permutation, i.e., $G_2=G_2'$. Recall that from Lemma~\ref{Lemma:MAP} we have 
    \begin{align*}
    &\hat{\pi}_\mathrm{MAP}(G_1,G_2') \\
    &= \argmin_{\pi \in \mathcal{S}_n}\{w_1 \Delta ^{\mathrm{u}}(G_1,\pi^{-1}(G_2')) 
    + w_2 \Delta ^{\mathrm{a}}(G_1,\pi^{-1}(G_2'))\},
    \end{align*}
where $w_1 = \log\left(\frac{p_{11}p_{00}}{p_{10}p_{01}}\right)$, $w_2 = \log \left(\frac{q_{11}q_{00}}{q_{10}q_{01}}\right)$, and
\begin{align*}
    \Delta ^{\mathrm{u}}(G_1,\pi^{-1}(G_2')) &=\mspace{-12mu} \sum_{(i,j)\in \mathcal{E}_{\mathrm{u}}} \mspace{-12mu} \mathbb{1} \{G_1((i,j)) \neq G_2'((\pi(i),\pi(j)))\},\\
    \Delta ^{\mathrm{a}}(G_1,\pi^{-1}(G_2')) &= \mspace{-12mu}\sum_{(i,v) \in \mathcal{E}_{\mathrm{a}}} \mspace{-12mu} \mathbb{1} \{G_1((i,v)) \neq G_2'((\pi(i),v))\}.
\end{align*}
It suffices to show that for any $\sigma\in \mathrm{Aut}(G_1 \wedge G_2)$, we have $\Delta ^{\mathrm{u}}(G_1,G_2)= \Delta ^{\mathrm{u}}(G_1,\sigma(G_2))$ and $\Delta ^{\mathrm{a}}(G_1,G_2)= \Delta ^{\mathrm{a}}(G_1,\sigma(G_2))$. This would imply that permutation $\sigma^{-1}$ has the same posterior as the identity permutation, and thus the estimator cannot do better than a random guess between these permutations.

We firstly show that $\Delta ^{\mathrm{u}}(G_1,G_2)= \Delta ^{\mathrm{u}}(G_1,\sigma(G_2))$. Consider a user pair $(i,j)\in \mathcal{E}_{\mathrm{u}}$. Suppose that $G_1(i,j)=G_2(i,j)=1$, i.e., $(G_1\wedge G_2)(i,j)=1$. Because $\sigma\in \mathrm{Aut}(G_1 \wedge G_2)$, we know that $G_1(\sigma(i),\sigma(j))=G_2(\sigma(i),\sigma(j))=1$. So the contribution of $(i,j)$ to $\Delta ^{\mathrm{u}}(G_1,G_2)$ and $\Delta ^{\mathrm{u}}(G_1,\sigma(G_2))$ are both zero. Next we consider the case of $(G_1\wedge G_2)(i,j)=0$. This includes subcases of $G_1(i,j)=0,G_2(i,j)=1$ and $G_1(i,j)=1,G_2(i,j)=0$ and $G_1(i,j)=0,G_2(i,j)=0$. Let $S_\mathrm{u}$ denote the edge orbit in the permutation $\sigma^\mathcal{E}$ that contains $(i,j)$. Because $\sigma\in \mathrm{Aut}(G_1 \wedge G_2)$, we know that $(G_1\wedge G_2)(e)=0$ for each $e\in S_\mathrm{u}$. Note that the contribution of $e$ to $\Delta^\mathrm{u}(G_1, G_2)$ is one if $G_1(e)=0,G_2(e)=1$ or $G_1(e)=1,G_2(e)=0$ and the contribution is zero if $G_1(e)=0,G_2(e)=0$. Therefore, the total contribution of the orbit $S_\mathrm{u}$ to $\Delta ^{\mathrm{u}}(G_1,G_2)$ is given by the total number of edges in $G_1$ and $G_2$ on the orbit $S_\mathrm{u}$. Similarly, the contribution of $S_\mathrm{u}$ to $\Delta ^{\mathrm{u}}(G_1,\sigma(G_2))$ is the same. Thus, we have $\Delta ^{\mathrm{u}}(G_1,G_2)= \Delta ^{\mathrm{u}}(G_1,\sigma(G_2))$.

Secondly, we show that $\Delta ^{\mathrm{a}}(G_1,G_2)= \Delta ^{\mathrm{a}}(G_1,\sigma(G_2))$. Consider a user-attribute pair $(i,a)\in \mathcal{E}_\mathrm{a}$. Then if $(G_1\wedge G_2)(i,a)=1$, we have $G_1(i,a)=G_2(i,a)=G_1(\sigma(i),a)=G_2(\sigma(i),a)=1$. So the contribution of $(i,a)$ to $\Delta ^{\mathrm{a}}(G_1,G_2)$ and $\Delta ^{\mathrm{a}}(G_1,\sigma(G_2))$ are both zero. Next, suppose $(G_1\wedge G_2)(i,a)=0$. Let $S_\mathrm{a}$ denote the edge orbit that contains $(i,a)$. It is not hard to see that the contribution of the user-attribute pairs in $S_\mathrm{a}$ to $\Delta ^{\mathrm{a}}(G_1,G_2)$ equals to the total number of edges in $G_1$ and $G_2$ on $S_\mathrm{a}$, and its contribution of $\Delta ^{\mathrm{a}}(G_1,\sigma(G_2))$ is the same. So we have $\Delta ^{\mathrm{a}}(G_1,G_2)= \Delta ^{\mathrm{a}}(G_1,\sigma(G_2))$.

\end{proof}

\section{Proof of Fact~\ref{Fact:ogf cycle decomposition}}
% Orbit decomposition}
\ogfcycle*
% \begin{repfact}{Fact:OFG cycle decomposition}
% \label{Appendix:fact1}
% The generating function $\mathcal{A}(\bm{x},\bm{y},z)$ for permutation $\pi$ can be decomposed into
% \begin{align*}
%     \mathcal{A}(\bm{x},\bm{y},z)= & \prod_{l\geq1} \mathcal{A}_{l}(\bm{x},z)^{t_l ^{\mathrm{u}}} \mathcal{A}_{l} (\bm{y},z)^{t_l ^{\mathrm{a}}},
% \end{align*}
% where $t_l^\mathrm{u}$ is the number of $l$-cycle of user-user pairs, $t_l^\mathrm{a}$ is the number of $l$-cycle of user-attribute pairs. Here $\mathcal{A}_l(\bm{x},z)$ is defined as the generating function for a $l$-cycle of user-user pair
% \begin{align}
% % \label{Eq:l-cycle OGFu}
%     \mathcal{A}_l(\bm{x},z) = \sum_{g\in \{0,1\}^{\mathcal{E}_l}}\sum_{h \in \{0,1\}^{\mathcal{E}_l}} z^{\delta_{\pi}(g,h)}\bm{x}^{\bm{\muu} (g,h)}.
% \end{align}
% The generating function for a $l$-cycle of user-user pair is defined as
% \begin{align}
% % \label{Eq:l-cycle OGFa}
%     \mathcal{A}_l(\bm{y},z) = \sum_{g\in \{0,1\}^{\mathcal{E}_l}}\sum_{h \in \{0,1\}^{\mathcal{E}_l}} z^{\delta_{\pi}(g,h)}\bm{y}^{\bm{\mua} (g,h)}.
% \end{align}
% \end{repfact}
\begin{proof}
\label{Appendix:fact1}
Recall the definition of $\mathcal{A}(\bm{x},\bm{y},z)$ for a given $\pi$
\begin{align*}
    \mathcal{A}(\bm{x},\bm{y},z) = \sum_{g\in \{0,1\}^\mathcal{E}}\sum_{h \in \{0,1\}^\mathcal{E}} z^{\delta_{\pi}(g,h)}\bm{x}^{\bm{\muu} (g,h)} \bm{y}^{\bm{\mua}(g,h)}.
\end{align*}
According to the cycle decomposition on $\pi^{\mathcal{E}}$, we write $\mathcal{E} = \cup_{i\geq 1} \orbit_i$, where use $\orbit_i$ is the $i$th orbit and there are $N$ orbits in total. Then we have the following.
\begin{align}
    &\mathcal{A}(\bm{x},\bm{y},z) \nonumber\\ 
    &= \sum_{g\in \{0,1\}^\mathcal{E}}\sum_{h \in \{0,1\}^\mathcal{E}} z^{\delta_{\pi}(g,h)}\bm{x}^{\bm{\muu} (g,h)} \bm{y}^{\bm{\mua}(g,h)}\nonumber\\
    \label{eqFact1_0}
    &= \sum_{g\in \{0,1\}^\mathcal{E}}\sum_{h \in \{0,1\}^\mathcal{E}} \prod_{e \in \mathcal{E}} z^{\delta_\pi(g_e,h_e)} \bm{x}^{\bm{\muu}(g_e,h_e)} \bm{y}^{\bm{\mua}(g_e,h_e)}\\
    \label{eqFact1_1}
    &= \sum_{g\in \{0,1\}^\mathcal{E}}\sum_{h \in \{0,1\}^\mathcal{E}} \prod_{i=1}^{N} f(g_{\orbit_i},h_{\orbit_i})\\
    \label{eqFact1_2}
    &= \mspace{-15mu} \sum_{\substack{g_{\orbit_1} \in \{0,1\}^{\orbit_1}}}\sum_{h_{\orbit_1} \in \{0,1\}^{\orbit_1}}
    \mspace{-10mu} \ldots \mspace{-10mu}
    \sum_{h_{\orbit_N} \in \{0,1\}^{\orbit_N}} \prod_{i=1}^{N} f(g_{\orbit_i},h_{\orbit_i}) \\
    \label{eqFact1_3}
    &=\prod_{i=1}^{N} \left(\sum_{g_{\orbit_i}\in \{0,1\}^{\orbit_i}}\sum_{h_{\orbit_i} \in \{0,1\}^{\orbit_i}} f(g_{\orbit_i},h_{\orbit_i})\right) \\
    \label{eqFact1_4}
    &=\prod_{i=1}^{N} \mathcal{A}_{\orbit_i} (\bm{x},\bm{y},z)\\
    \label{eqFact1_5}
    &=  \prod_{l\geq1} \mathcal{A}_{l}(\bm{x},z^{w_1})^{t_l ^{\mathrm{u}}} \mathcal{A}_{l} (\bm{y},z^{w_2})^{t_l ^{\mathrm{a}}}.
\end{align}
Here we use $g_{\mathcal{E'}}$ to denote a subset of $g$ that contains only vertex pairs in $\mathcal{E}'$ and $h_{\mathcal{E'}}$ to denote a subset of $h$ that contains only vertex pairs in $\mathcal{E}'$, where $\mathcal{E'}$ can be any set of vertex pairs. In \eqref{eqFact1_0}, $g_e$ (resp. $h_e$) represent a subset of $g$ (resp. $h$) that contains a single vertex pair $e$. 
In \eqref{eqFact1_1}, $g_{\orbit_i}$ (resp. $h_{\orbit_i}$) represents the subset of $g$ (resp. $h$) that contains only vertex pairs in the orbit $\orbit_i$. 
We define $f(g_{\orbit_i},h_{\orbit_i})$ as a function of $g_{\orbit_i}$ and $h_{\orbit_i}$ where $f(g_{\orbit_i},h_{\orbit_i}) = \prod_{e \in \orbit_i}z^{\delta_\pi(g_e,h_e)} \bm{x}^{\bm{\muu}(g_e,h_e)}$ if $\orbit_i$ only contains user-user pairs, and $f(g_{\orbit_i},h_{\orbit_i}) = \prod_{e \in \orbit_i}z^{\delta_\pi(g_e,h_e)} \bm{y}^{\bm{\mua}(g_e,h_e)}$ if $\orbit_i$ only contains user-attribute pairs. 
% Therefore, we have $\prod_{i=1}^{N} f(g_{\orbit_i},h_{\orbit_i}) = z^{\delta_{\pi}(g,h)}\bm{x}^{\bm{\muu} (g,h) } \bm{y}^{\bm{\mua}(g,h)}$.
Equation \eqref{eqFact1_2} follows because $\orbit_i$'s  are disjoint and their union is $\mathcal{E}$.  
Note that $f(g_{\orbit_i},h_{\orbit_i})$ only concerns vertex pairs in the cycle $\orbit_i$ since for $e \in \orbit_i$ we have $\pi^{\mathcal{E}}(e) \in \orbit_i$.
Then, \eqref{eqFact1_3} follows because $f(g_{\orbit_i},h_{G_i})$'s are independent functions.
In \eqref{eqFact1_4}, we use $\mathcal{A}_{\orbit_i} (\bm{x}, \bm{y},z)$ to denote the generating function for the orbit $\orbit_i$ where $\mathcal{A}_{\orbit_i} (\bm{x}, \bm{y},z) = \mathcal{A}_{\orbit_i} (\bm{x},z)$ if $\orbit_i$ contains user-user vertex pairs; $\mathcal{A}_{\orbit_i} (\bm{x}, \bm{y},z) = \mathcal{A}_{\orbit_i} (\bm{y},z)$ if $\orbit_i$ contains user-attribute vertex pairs. To see why this equation follows,
note that if $\orbit_i$ contains only user-user vertex pairs, then
\begin{align*}
    & \sum_{g_{\orbit_i}\in \{0,1\}^{\orbit_i}}\sum_{h_{\orbit_i} \in \{0,1\}^{\orbit_i}} f(g_{\orbit_i},h_{\orbit_i}) \\
    &= \sum_{g_{\orbit_i}\in \{0,1\}^{\orbit_i}}\sum_{h_{\orbit_i} \in \{0,1\}^{\orbit_i}}  \prod_{e \in \orbit_i} z^{\delta_\pi(g_e,h_e)} \bm{x}^{\bm{\muu}(g_e,h_e)}\\
    & = \sum_{g_{\orbit_i}\in \{0,1\}^{\orbit_i}}\sum_{h_{\orbit_i} \in \{0,1\}^{\orbit_i}}  z^{\delta_\pi(g_{\orbit_i},h_{\orbit_i})} \bm{x}^{\bm{\muu}(g_{\orbit_i},h_{\orbit_i})}\\
    & = \mathcal{A}_{\orbit_i} (\bm{x},z).
\end{align*}
If $\orbit_i$ contains only user-attribute vertex pairs, then
\begin{align*}
    & \sum_{g_{\orbit_i}\in \{0,1\}^{\orbit_i}}\sum_{h_{\orbit_i} \in \{0,1\}^{\orbit_i}} f(g_{\orbit_i},h_{\orbit_i}) \\
    &= \sum_{g_{\orbit_i}\in \{0,1\}^{\orbit_i}}\sum_{h_{\orbit_i} \in \{0,1\}^{\orbit_i}}  \prod_{e \in \orbit_i} z^{\delta_\pi(g_e,h_e)} \bm{y}^{\bm{\mua}(g_e,h_e)}\\
    & = \sum_{g_{\orbit_i}\in \{0,1\}^{\orbit_i}}\sum_{h_{\orbit_i} \in \{0,1\}^{\orbit_i}}  z^{\delta_\pi(g_{\orbit_i},h_{\orbit_i})} \bm{y}^{\bm{\mua}(g_{\orbit_i},h_{\orbit_i})}\\
    & = \mathcal{A}_{\orbit_i} (\bm{y},z).
\end{align*}
In \eqref{eqFact1_5}, we apply the fact that orbits of the same size have the same generating function. 
\end{proof}

\section{A useful fact for Corollaries~\ref{Coro:achievability} and~\ref{coro:seed_tight}}\label{appx:fact4}
\begin{fact}
\label{fact-subsample}
Consider the subsampling representation of the graph parameters
\begin{align*}
    \begin{pmatrix}
    p_{11} & p_{10}\\
    p_{01} & p_{00}
    \end{pmatrix}
    &=\begin{pmatrix}
    p s_{\mathrm{u},1} s_{\mathrm{u},2} & p s_{\mathrm{u},1} (1-s_{\mathrm{u},2})\\
    p (1-s_{\mathrm{u},1}) s_{\mathrm{u},2} & p(1-s_{\mathrm{u},1})(1-s_{\mathrm{u},2}) + 1-p
    \end{pmatrix},
    % \\
    % \begin{pmatrix}
    % q_{11} & q_{10}\\
    % q_{01} & q_{00}
    % \end{pmatrix}
    % &=\begin{pmatrix}
    % q s'_1s'_2 & q s'_1 (1-s'_2)\\
    % q (1-s'_1) s'_2 & q(1-s'_1)(1-s'_2) + 1-q
    % \end{pmatrix}.
\end{align*}
If 
$1-p=\Theta(1),$
then we have $\psiu = \Theta(p_{11})$ and $\psiu = p_{11} - \Theta(p_{11}p^{1/2})$. The statement holds if we exchange $\bm{p}$ to $\bm{q}$.
\end{fact}

\begin{proof}
To see $\psiu = \Theta(p_{11})$, we write $\psiu$ using parameters from the subsampling model and we have
\begin{align}
    \psiu &= (\sqrt{p_{11}p_{00}}-\sqrt{p_{10}p_{01}})^2\nonumber\\
    &=( \sqrt{p_{11}((1-p)+p(1-s_{\mathrm{u},1})(1-s_{\mathrm{u},2}))} \nonumber\\ 
    & \;\ \;\ -\sqrt{p^2s_{\mathrm{u},1}s_{\mathrm{u},2}(1-s_{\mathrm{u},1})(1-s_{\mathrm{u},2})}  )^2 \nonumber\\
    &=(1-p)p_{11} \left(\sqrt{1+\tfrac{p(1-s_{\mathrm{u},1})(1-s_{\mathrm{u},2})}{1-p}} - \sqrt{\tfrac{p(1-s_{\mathrm{u},1})(1-s_{\mathrm{u},2})}{1-p}}\right)^2 \nonumber\\
    \label{eqf4_1}
    &= (1-p)p_{11}
    \frac{1}{\left(\sqrt{1+\frac{p(1-s_{\mathrm{u},1})(1-s_{\mathrm{u},2})}{1-p}} + \sqrt{\frac{p(1-s_{\mathrm{u},1})(1-s_{\mathrm{u},2})}{1-p}}\right)^2}.
\end{align}
In \eqref{eqf4_1}, we have that $(1-p) = \Theta(1)$ and   $\frac{1}{\sqrt{1+\frac{p(1-s_{\mathrm{u},1})(1-s_{\mathrm{u},2})}{1-p}} + \sqrt{\frac{p(1-s_{\mathrm{u},1})(1-s_{\mathrm{u},2})}{1-p}}} = \Theta(1)$. Therefore, $\psiu = \Theta(p_{11})$.

To see $\psiu = p_{11} - \Theta(p_{11}^{3/2})$, we take
\begin{align*}
    \psiu &= (\sqrt{p_{11}p_{00}}-\sqrt{p_{10}p_{01}})^2\nonumber\\
    & = p_{11}p_{00}+p_{10}p_{01} -2\sqrt{p_{11}p_{00}p_{10}p_{01}}\\
    & = p_{11} ((1-p)+p(1-s_{\mathrm{u},1})(1-s_{\mathrm{u},2})) \\
    & \;\ + p^2s_{\mathrm{u},1}s_{\mathrm{u},2}(1-s_{\mathrm{u},1})(1-s_{\mathrm{u},2})\\
    &\mspace{-20mu} -\sqrt{p_{11}^2 ((1-p)+p(1-s_{\mathrm{u},1})(1-s_{\mathrm{u},2}))p(1-s_{\mathrm{u},1})(1-s_{\mathrm{u},2})}\\
    & = p_{11} - O(p_{11}p^{1/2}).
\end{align*}
% where the last step follows from $s_1 = \Theta(1)$ and $s_2 = \Theta(1)$.

\end{proof}

\section{Proof of Corollary~\ref{Coro:achievability}}
\label{appx:coro-proof}
% \begin{proof}[Proof of Corollary~\ref{Coro:achievability}]

\paragraph{Achievability} In this proof, we first show that, under the assumptions on the user-user edges in condition~\eqref{eq:cond-nonedgeProb} and \eqref{eq:cond-rho}, the achievability result becomes
\[
n p_{11} +m\psia-\log{n} = \omega(1).
\]
% Next, we apply the assumptions on the user-attribute edges from~\eqref{eq:cond-nonedgeProb} and \eqref{eq:cond-rho}, and derive the two cases in this Corollary by approximating $\psia$.
Next, we apply conditions~\eqref{eq:cond-attedgeProb} and~\eqref{eq:cond-rhoa} to bound difference between $q_{11}$ and $\psia$, and complete the proof.

For the user-user edge part, we check the two regimes $p_{11}=\omega(\frac{\log{n}}{n})$ and $p_{11}=O(\frac{\log{n}}{n})$ separately.
If $p_{11} = \omega(\frac{\log{n}}{n})$, then with the assumption on the user-user edge density \eqref{eq:cond-nonedgeProb},  we also have $\psiu=\omega(\frac{\log{n}}{n})$ because $\psiu=\Theta(p_{11})$ (see Fact~\ref{fact-subsample} in Appendix~\ref{appx:fact4}). Therefore exact alignment is achievable according to Theorem~\ref{Thm:achievability-general}: $\tfrac{n\psiu}{2} + m\psia -\log{n} = \omega(\log n) + m\psia - \log n = \omega(1)$. 
Now we check the case when $p_{11}=O(\frac{\log{n}}{n})$.
We will see that all the conditions in Theorem~\ref{Thm:achievability-sparse} are satisfied. Notice that $p_{10}=ps_{\mathrm{u},1}(1-s_{\mathrm{u},2})\le ps_{\mathrm{u},1}=\frac{p_{11}}{s_{\mathrm{u},2}}$ and $p_{01}=ps_{\mathrm{u},2}(1-s_{\mathrm{u},1})\le ps_{\mathrm{u},2}=\frac{p_{11}}{s_{\mathrm{u},1}}$. Under assumption \eqref{eq:cond-rho}, we know that $s_{\mathrm{u},1}=\Omega(\frac{(\log n)^4}{n})$ and $s_{\mathrm{u},2}=\Omega(\frac{(\log n)^4}{n})$. Because $p_{11}=O(\frac{\log n}{n})$, we have $p_{01}=O(\frac{1}{\log n})$ and $p_{10}=O(\frac{1}{\log n})$. Moreover, note that
\begin{align*}
    \frac{p_{10}p_{01}}{p_{00}p_{11}}=O\left(\frac{p^2s_{\mathrm{u},1}s_{\mathrm{u},2}(1-s_{\mathrm{u},1})(1-s_{\mathrm{u},2})}{ps_{\mathrm{u},1}s_{\mathrm{u},2}}\right)=O(p).
\end{align*}
Because $ps_{\mathrm{u},1}s_{\mathrm{u},2}=O\left(\frac{\log n}{n}\right)$ and $s_{\mathrm{u},1}s_{\mathrm{u},2}=\Omega\left(\frac{(\log n)^4}{n}\right)$, we have $p=O\left(\frac{1}{(\log n)^3}\right)$. Then it follows that the sparsity conditions in Theorem~\ref{Thm:achievability-sparse} ~\eqref{eq:lem31},~\eqref{eq:lem32} and~\eqref{eq:lem33}  are all satisfied. Therefore, we just need $n p_{11} +m\psia-\log{n} = \omega(1)$ to guarantee that exact alignment is achievable. Combining the two case, we come to the conclusion that, under the assumptions \eqref{eq:cond-nonedgeProb} \eqref{eq:cond-rho}, the achievability results in Theorem~\ref{Thm:achievability-general} and Theorem~\ref{Thm:achievability-sparse} simplifies to 
\begin{equation}
\label{eq:simp-ach}
n p_{11} +m\psia-\log{n} = \omega(1).
\end{equation}

% Notice that under the assumption on the edge correlation \eqref{eq:cond-rho}, we have $p_{10} = O(p_{11})$ and $p_{01} = O(p_{11})$  (see Fact~\ref{fact-subsample} in Appendix~\ref{appx:fact4}). Then it follows that the sparsity constrains in Theorem~\ref{Thm:achievability-sparse} ~\eqref{eq:lem31}~\eqref{eq:lem32}~\eqref{eq:lem33}  are all satisfied. Therefore, we just need $n p_{11} +m\psia-\log{n} = \omega(1)$ to guarantee that exact alignment is achievable. Combining the two case, we come to the conclusion that, under the assumptions \eqref{eq:cond-nonedgeProb} \eqref{eq:cond-rho}, the achievability results in Theorem~\ref{Thm:achievability-general} and Theorem~\ref{Thm:achievability-sparse} simplifies to $n p_{11} +m\psia-\log{n} = \omega(1)$.

% From the above discussion, we further simplify the achievability results so that we can see how it influenced by only on $np_{11}$ and $mq_{11}$,
From the above discussion, we further simplify the achievability results so that it depends only on $np_{11}$ and $mq_{11}$, which are the two only parameters in the converse bound.
Then we can show in what regime the achievability and converse are tight (up to $\pm\omega(1)$).
% and thus show when the achievability and converse are tight (up to $\pm\omega(1)$).
From the achievability in last step: $np_{11}+m\psia-\log n =\omega(1)$, we then need to determine the difference between $m\psia$ and $mq_{11}$. 
Firstly, consider the case when $mq_{11}=\omega(\log n)$. In this case, we immediately have that $np_{11}+m\psia-\log n =\omega(1)$ because $\psia=\Theta(q_{11})$ by Fact~\ref{fact-subsample}.
Now, consider the case when $mq_{11}=O(\log n)$. Suppose $np_{11}+mq_{11}-\log n =\omega(1)$ implies $mq_{11}q^{1/2}=O(1)$, then it implies $np_{11}+m\psia-\log n =\omega(1)$ as well. This is because $mq_{11}-m\psia = O (mq_{11}q^{1/2})$. Therefore, we need to find the condition for $np_{11}+mq_{11}-\log n =\omega(1)$ to imply $mq_{11}q^{1/2}=O(1)$. Because $mq_{11}=mq s_{\mathrm{a},1}s_{\mathrm{a},2}=O(\log n)$, we have $mq_{11}q^{1/2}=mq_{11}\left(\frac{mq_{11}}{ms_{\mathrm{a},1}s_{\mathrm{a},2}}\right)^{1/2}=O\left(\frac{(\log n)^{3/2}}{(ms_{\mathrm{a},1}s_{\mathrm{a},2})^{1/2}}\right)$. Condition~\eqref{eq:cond-rhoa} implies that $ms_{\mathrm{a},1}s_{\mathrm{a},2}=\Omega((\log n)^3)$, so we have that $mq_{11}q^{1/2}=O(1)$, which completes the proof.

\paragraph{Converse} In this proof, we will show that \[np_{11}+mq_{11}-\log n\rightarrow -\infty\] implies condition~\eqref{eq:converse} in Theorem~\ref{Thm:converse}. Note that $-\log(x^2+(1-x)^2)\ge 2x$ for any $x\in [0,1]$. Therefore, we have \[
-n\log(p_{11}^2+(1-p_{11})^2)-m\log(q_{11}^2+(1-q_{11})^2)\le 2np_{11}+2mq_{11}\le 2\log n-\omega(1),
\]
which completes the proof.

\section{Proof of Corollary~\ref{coro:seed_tight}}
\label{proof-seed_tight}
% \begin{proof}[Proof of Corollary~\ref{coro:seed_tight}]
\underline{Proof for the achievability condition~\eqref{eq:ach-seed}:}
Recall that from~\eqref{eq:simp-ach} in the proof of Corollary~\ref{Coro:achievability}, we have, under assumptions \eqref{eq:cond-nonedgeProb} and \eqref{eq:cond-rho}, our achievability results (Theorem~\ref{Thm:achievability-general} and Theorem~\ref{Thm:achievability-sparse}) simplify to the following condition
\begin{align*}
    n p_{11} +m\psia \geq \log{n}+ \infty,
\end{align*}
% where $\psia - q_{11} = O(q_{11}^{3/2})$.
where $\psia - q_{11} = O(q_{11}q^{1/2})$.

% Now, in the seeded \er\ setting, we have $\bm{p} = \bm{q}$. Correspondingly, we obtain the following achievability for seeded alignment
% \begin{align}
% \label{eq:seed_tight1}
%     n p_{11} +m \psiu \geq \log{n}+ \omega(1),
% \end{align}
% where  $\psiu - p_{11} = O(p_{11}^{3/2})$.

Now, in the seeded \er\ setting, we have $\bm{p} = \bm{q}$. Because condition~\eqref{eq:sparsity} implies conditions \eqref{eq:cond-nonedgeProb} and \eqref{eq:cond-rho}, we obtain the following achievability for seeded alignment
\begin{align}
\label{eq:seed_tight1}
    n p_{11} +m \psiu \geq \log{n}+ \omega(1),
\end{align}
where  $\psiu - p_{11} = O(p_{11}p^{1/2})$.

For the above achievability condition~\eqref{eq:seed_tight1}, we show that it is equivalent to
\begin{align}
    \label{eq:seed_tight2}
    (n+m)p_{11}  \geq   \log n+\omega(1)
\end{align}
by comparing them in the following three regimes.
\begin{enumerate}
    \item For the regime $(n+m)p_{11} = \omega(\log n)$, we show that it is strictly contained in both \eqref{eq:seed_tight1} and \eqref{eq:seed_tight2}. We can easily see that $(n+m)p_{11} = \omega(\log n)$ satisfy condition~\eqref{eq:seed_tight2}. For condition~\eqref{eq:seed_tight1}, recall that we have $\psiu=\Theta(p_{11})$ from Fact~\ref{fact-subsample} (cf. Appendix~\ref{appx:fact4}). Thus, we also have $(n+m)p_{11} = \omega(\log n)$ satisfy condition~\eqref{eq:seed_tight2}.
    \item In the regime $(n+m)p_{11} = \Theta(\log n)$, we have
    % \begin{align*}
    %      m p_{11}^{3/2} &= \Theta\left(m\left(\frac{\log n}{n+m}\right)^{3/2}\right) \\
    %      &\leq \Theta \left( m  \frac{\log n }{m} \left(\frac{\log n}{n}\right)^{1/2}\right)\\
    %      &= \Theta\left( \frac{\log n^{3/2}}{n^{1/2}}\right) = O(1).
    % \end{align*}
    \begin{align*}
         m p_{11}p^{1/2} &= O\left(mp_{11}\left(\frac{\log n}{(m+n)s^2}\right)^{1/2}\right)\\
         &=O\left(\log n\left(\frac{\log n}{ns^2}\right)^{1/2}\right)\\
         &=O\left(\frac{(\log n)^{3/2}}{(\log n)^2}\right)\\
         &=O(1),
    \end{align*}
    where the penultimate equality follows by assumption~\eqref{eq:cond-rho}. 
    For the condition \eqref{eq:seed_tight1}, we have $m\psiu = mp_{11} - O(mp_{11}p^{1/2}) = mp_{11} - O(1)$. Therefore, condition \eqref{eq:seed_tight1} can be simplified to $np_{11} + mp_{11} \geq \log n +\omega(1)$, which is exactly condition \eqref{eq:seed_tight2}.
    \item For the regime $(n+m)p_{11} = o(\log n)$, it is not contained by neither \eqref{eq:seed_tight1} nor \eqref{eq:seed_tight2}.
\end{enumerate}

\underline{Proof for the converse condition~\eqref{eq:conv-seed}:}
From Theorem~\ref{Thm:converse}, we have the converse condition for attributed \er\ alignment
\begin{align*}
    np_{11} + m q_{11} \leq \log n -\omega(1).
\end{align*}
Now, in the seeded \er\ setting, we have $\bm{p} = \bm{q}$ and we directly obtain the following converse for seeded alignment
\begin{align}
\label{eq:seed_tight3}
    (n+m) p_{11} \leq \log n -\omega(1).
\end{align}
% \end{proof}
% \end{proof}
% \section{Appendix-B}
% \input{proof-achievability}
% \appendix

 \section*{Acknowledgment}
This work was supported in part by the NSERC Discovery Grant No.\ RGPIN-2019-05448, the NSERC Collaborative Research and Development Grant CRDPJ 54367619, and the NSF grant CNS-2007733.

\bibliographystyle{IEEEtranN}
\bibliography{references}
\end{document}